\documentclass[11pt,journal,onecolumn]{IEEEtran}
\IEEEoverridecommandlockouts
\usepackage{amsmath}
\usepackage{amsfonts}
\usepackage{amssymb}
\usepackage{amscd}
\usepackage[dvips]{graphicx}
\usepackage{epsfig}
\usepackage{subfigure}
\usepackage{newlfont}
\usepackage{cite}

\renewcommand{\baselinestretch}{1.5}
\newtheorem{theorem}{Theorem}
\newtheorem{lemma}{Lemma}

\newtheorem{corollary}{Corollary}

\newcommand{\usrind}{\lcb i+1, \ldots, K\rcb}
\newcommand{\Hjj}{\mathbf{H}_{jj}}
\newcommand{\Hnew}{\mathbf{H}}

\newcommand{\Pnew}{\mathbf{P}}
\newcommand{\Hii}{\mathbf{H}_{ii}}
\newcommand{\Yj}{\mathbf{y}_{j}}
\newcommand{\Yi}{\mathbf{y}_{i}}
\newcommand{\Honebar}{\mathbf{\overline{H}}_{11}}
\newcommand{\Htwobar}{\mathbf{\overline{H}}_{12}}
\newcommand{\Hthreebar}{\mathbf{\overline{H}}_{22}}
\newcommand{\Ponebar}{\mathbf{\overline{P}}_{1}}
\newcommand{\Ptwobar}{\mathbf{\overline{P}}_{2}}
\newcommand{\Yonebar}{\mathbf{\overline{y}}_{1}}
\newcommand{\Ytwobar}{\mathbf{\overline{y}}_{2}}
\newcommand{\Xonebar}{\mathbf{\overline{x}}_{1}}
\newcommand{\Xtwobar}{\mathbf{\overline{x}}_{2}}
\newcommand{\Zonebar}{\mathbf{\overline{z}}_{1}}
\newcommand{\Ztwobar}{\mathbf{\overline{z}}_{2}}
\newcommand{\Sbar}{\mathbf{\overline{s}}}
\newcommand{\Hji}{\mathbf{H}_{ji}}
\newcommand{\Hij}{\mathbf{H}_{ij}}
\newcommand{\Hplus}{\mathbf{H}_{i+1,i}}
\newcommand{\Hminus}{\mathbf{H}_{i-1,i}}
\newcommand{\siplus}{\mathbf{s}_{i+1,i}}
\newcommand{\siminus}{\mathbf{s}_{i-1,i}}
\newcommand{\ziplus}{\mathbf{z}_{i+1}}
\newcommand{\ziminus}{\mathbf{z}_{i-1}}
\newcommand{\Pj}{\mathbf{P}_{j}}
\newcommand{\Pdashi}{\mathbf{P}_{i}}
\newcommand{\myxi}{\mathbf{x}_{i}}
\newcommand{\myxj}{\mathbf{x}_{j}}
\newcommand{\Zj}{\mathbf{z}_{j}}
\newcommand{\Zi}{\mathbf{z}_{i}}
\newcommand{\Honej}{\mathbf{H}_{1j}}
\newcommand{\HKj}{\mathbf{H}_{Kj}}
\newcommand{\Honeone}{\mathbf{H}_{11}}
\newcommand{\Pone}{\mathbf{P}_{1}}
\newcommand{\PK}{\mathbf{P}_{K}}
\newcommand{\Pk}{\mathbf{P}_{K}}
\newcommand{\Honek}{\mathbf{H}_{1K}}
\newcommand{\Hkone}{\mathbf{H}_{K1}}
\newcommand{\Hkk}{\mathbf{H}_{KK}}
\newcommand{\Hionebar}{\mathbf{\overline{H}}_{i1}}
\newcommand{\Pionebar}{\mathbf{\overline{P}}_{i1}}
\newcommand{\Hkonebar}{\mathbf{\overline{H}}_{Ki}}
\newcommand{\Hitwobar}{\mathbf{\overline{H}}_{i,i+1}}
\newcommand{\Pitwobar}{\mathbf{\overline{P}}_{i2}}
\newcommand{\Honekbar}{\mathbf{\overline{H}}_{1,i+1}}
\newcommand{\Hithreebar}{\mathbf{\overline{H}}_{iK}}
\newcommand{\Pithreebar}{\mathbf{\overline{P}}_{i3}}
\newcommand{\Honetwobar}{\mathbf{\overline{H}}_{1i}}
\newcommand{\Hifourbar}{\mathbf{\overline{H}}_{i,K-1}}
\newcommand{\Pifourbar}{\mathbf{\overline{P}}_{i4}}
\newcommand{\HKthreebar}{\mathbf{\overline{H}}_{K,i+1}}
\newcommand{\Yone}{\mathbf{y}_{1}}
\newcommand{\Yk}{\mathbf{y}_{K}}
\newcommand{\Zk}{\mathbf{z}_{K}}
\newcommand{\Zone}{\mathbf{z}_{1}}
\newcommand{\Skone}{\mathbf{s}_{K,1}}
\newcommand{\Sonek}{\mathbf{s}_{1,K}}
\newcommand{\Skonetwoi}{\mathbf{s}_{K,\{1,2,\ldots,i\}}}
\newcommand{\Soneacomp}{\mathbf{s}_{1,\usrind}}
\newcommand{\Sonektwoi}{\mathbf{s}_{1,\{K,2,\ldots,i\}}}
\newcommand{\Skacomp}{\mathbf{s}_{K,\mathcal{A}-\{K,2,3,\ldots,i\}}}
\newcommand{\sumneq}{\displaystyle\sum_{j=1,\,j \neq i}^{K}}
\newcommand{\sumneqmod}{\displaystyle\sum_{i=1,\,i \neq j}^{K}}
\newcommand{\mysum}{\displaystyle\sum_{i=1}^{L}}
\newcommand{\mysumK}{\displaystyle\sum_{i=1}^{K}}

\newcommand{\Iden}{\mathbf{I}}

\newcommand{\labs}{\left\lvert}
\newcommand{\rabs}{\right\rvert}
\newcommand{\lb}{\left(}
\newcommand{\rb}{\right)}
\newcommand{\lsqb}{\left[}
\newcommand{\rsqb}{\right]}
\newcommand{\lcb}{\left\{}
\newcommand{\rcb}{\right\}}
\newcommand{\mutuali}{I(\mathbf{x}_{i}^{n};\mathbf{y}_{i}^{n})}
\newcommand{\expec}{\mathbf{E}}
\newcommand{\oone}{\mathcal{O}(1)}
\newcommand{\sums}{\displaystyle\sum_{j \in S}}

\newcommand{\Uonetwobar}{\mathbf{\overline{U}}_{12}}
\newcommand{\eigm}{\overline{\Sigma}}

\newcommand{\Htildetwtwo}{\widetilde{\mathbf{H}}_{22}}
\newcommand{\Hitilde}{\tilde{\mathbf{H}}_{ii}}
\newcommand{\Complex}{\mathbb{C}}
\newcommand{\Real}{\mathbb{R}}
\newcommand{\dimone}{\displaystyle\sum_{j=1}^{i}M_{j}}
\newcommand{\dimtwo}{\displaystyle\sum_{j=i+1}^{K}M_{j}}
\newcommand{\Mri}{M_{r_{i}}}
\newcommand{\Msi}{M_{s_{i}}}
\newcommand{\mean}[1]{\mathbb{E}\left\{#1\right\}}

\newcommand{\Hbar}{\mathbf{\overline{H}}}

\newcommand{\setxysizeo}{\epsfxsize=6.2in \epsfysize=4in}
\newcommand{\setxysizeone}{\epsfxsize=6.5in}
\title{On the Generalized Degrees of Freedom of the $~K$-user Symmetric MIMO Gaussian Interference Channel}
\author{
\IEEEauthorblockN{Parthajit Mohapatra and Chandra R. Murthy\\}
\IEEEauthorblockA{Dept. of ECE, Indian Institute of Science\\
Bangalore  560 012, India\\
Email: \{partha, cmurthy\}@ece.iisc.ernet.in}}
\begin{document}
\maketitle
\begin{abstract}
The $K$-user symmetric multiple input multiple output (MIMO) Gaussian interference channel (IC) where each transmitter has $M$ antennas and each receiver has $N$ antennas is studied from a generalized degrees of freedom (GDOF) perspective. An inner bound on the GDOF is derived using a combination of techniques such as treating interference as noise, zero forcing (ZF) at the receivers, interference alignment (IA), and extending the Han-Kobayashi (HK) scheme to $K$ users, as a function of the number of antennas and the $\log \text{INR} / \log \text{SNR}$ level. Three outer bounds are derived, under different assumptions of cooperation and providing side information to receivers. The novelty in the derivation lies in the careful selection of side information, which results in the cancellation of the negative differential entropy terms containing signal components, leading to a tractable outer bound. The overall outer bound is obtained by taking the minimum of the three outer bounds. The derived bounds are simplified for the MIMO Gaussian symmetric IC to obtain outer bounds on the generalized degrees of freedom (GDOF). Several interesting conclusions are drawn from the derived bounds. For example, when $K > \frac{N}{M}+1$, a combination of the HK and IA schemes performs the best among the schemes considered. When $\frac{N}{M} < K \leq \frac{N}{M}+1$, the HK-scheme outperforms other schemes and is shown to be GDOF optimal. In addition, when the SNR and INR are at the same level, ZF-receiving and the HK-scheme have the same GDOF performance. It is also shown that many of the existing results on the GDOF of the Gaussian IC can be obtained as special cases of the bounds, e.g., by setting $K=2$ or the number of antennas at each user to $1$. 
\end{abstract}
\section{Introduction}\label{sec:introduction}
Approximate capacity characterization of the interference channel has recently received considerable research attention, both as a means to analyze the capacity scaling behavior as well as to obtain guidelines for interference management in a multi-user environment. Towards this, the concept of generalized degrees of freedom (GDOF) was introduced in \cite{etkin1} as a means of quantifying the extent of interference management in terms of the number of free signaling dimensions in a two-user interference channel (IC). In a multiuser MIMO setup, the use of multiple antennas at the transmitters and receivers can provide additional dimensions for signaling, which can in turn improve the GDOF performance of the IC. Characterizing the GDOF performance of a multiuser MIMO IC is therefore an important problem, and is the focus of this work.

Among the different possible methods to mitigate the effect of interference, two main approaches have typically been adopted in the literature. The first is based on the notion of splitting the message into private and public parts (also known as the Han-Kobayashi (HK) scheme) \cite{han1}, \cite{etkin1}. The second is based on the idea of interference alignment \cite{maddahali1, maddahali2, cadambe1}. These schemes are based on different ideas: the former allows part of the interference to be decoded and canceled at the unintended receivers, while the latter makes the interfering signals cast \emph{overlapping shadows} \cite{cadambe1} at the unintended receivers, allowing them to project the received signal in an orthogonal direction and remove the effect of interference. 

The HK-scheme proposed in \cite{han1} is known to achieve the largest possible rate region for the two-user single input single output (SISO) IC. Further, it can achieve a rate that is within 1 bit/s/Hz of the capacity of the channel for all values of the channel parameters \cite{etkin1}. Different variants of the HK-scheme for the two user IC can be found in \cite{telatar1, tarokh1, mahesh3}. The concept of interference alignment (IA) originated from the work of Maddah-Ali \textit{et al.} in \cite{maddahali1}, and was subsequently used in the DOF analysis of the $X$-channel in \cite{maddahali2} and \cite{jafar3}. This notion of IA was crystallized by Cadambe and Jafar in \cite{cadambe1}. Here, the precoding matrix is designed such that the interfering signals occupy a reduced dimension at all of the unintended receivers, while the desired signal remains decodable at the intended receiver. 
The idea of IA was extended to the $K$-user MIMO scenario in \cite{gou1}. More works on IA can be found in \cite{cadambe2, cadambe3, cadambe4}.  

The GDOF performance of the two-user MIMO IC was characterized in \cite{tarokh1}. It was extended to the $X$-channel and the $K$-user SISO IC in \cite{huang1} and \cite{jafar2}, respectively. In \cite{gou3}, the idea of message splitting was used to derive the GDOF in a SIMO setting when $K = N+1$, where $N$ is the number of receive antennas at each user. However, none of the existing studies consider the GDOF performance of the $K$-user MIMO Gaussian IC for $K > 2$. Moreover, the achievable GDOF performance of the HK-scheme and IA has not been contrasted in the literature.

Past work by several researchers has provided bounds on the degrees of freedom (DOF) and GDOF for multiuser ICs  (e.g., \cite{jafar0, cadambe1, gou1}). In \cite{jafar0}, a MIMO multiple access channel (MAC) outer bound on the sum capacity of the MIMO GIC (Gaussian IC) was derived, and simplified to obtain a bound on the DOF. It was also shown that zero forcing (ZF) receiving/precoding is sufficient to achieve all the available DOF. In \cite{cadambe1}, an outer bound on the DOF for the $K$-user SISO Gaussian symmetric IC was presented, and the novel idea of interference alignment (IA) developed in this work was found to be DOF optimal. Subsequently, in \cite{gou1}, an outer bound on the DOF for the $K$-user symmetric MIMO Gaussian IC was developed, and found to be tight when $R \triangleq \frac{\max\lb M,N\rb}{\min\lb M,N \rb}$ is an integer, where $M$ and $N$ are the number of transmitting and receiving antennas, respectively. The outer bound in \cite{gou1} was  improved in \cite{ghasemi1} by considering multiple ways of cooperation among users. The achievable scheme derived in \cite{ghasemi1} was found to be tight when $K\geq \frac{M+N}{\text{gcd}(M,N)}$, where $\text{gcd}(M,N)$ denotes the greatest common divisor of $M$ and $N$.  However, although several outer bounds have been derived for the DOF/GDOF, general outer bounds on the sum rate for the $K$-user MIMO GIC for $K>2$ that are valid for all values of the channel parameters are not available in the existing literature. Deriving such bounds can offer important insight into the performance limits of multiuser ICs.

In this paper, three new outer bounds on the sum rate are proposed, which are valid for all values of channel parameters. Further, these outer bounds are simplified to obtain outer bounds on the GDOF in the symmetric case. The overall outer bound on the GDOF is obtained by taking the minimum of the three bounds and the interference-free GDOF of $\min(M,N)$ per user. The first outer bound is based on using a combination of user cooperation similar in flavor to \cite{ghasemi1}, in conjunction with providing a subset of receivers with side information. The other two outer bounds are based on providing carefully selected side information to the receivers in such a way that the the negative differential entropy terms in the sum rate bound that contain a signal component cancel out, due to which, it is possible to obtain a single letter characterization. 

The three bounds on the GDOF perform differently depending on the values of the parameters $\alpha \triangleq \frac{\log \text{INR}}{\log \text{SNR}}$, $M, N$ and $K$; and this in turn provides insights into the performance limits of the system under different schemes for interference management. Several useful and interesting insights on the relative merits of the different schemes are obtained from the bounds. 
For example, when $K > \frac{N}{M} + 1$, neither the HK-scheme nor IA can uniformly outperform the other; which scheme is the better of the two depends on the $\log \text{INR} / \log \text{SNR}$ level. The performance of the proposed achievable schemes is compared with the outer bounds.  Using this, the GDOF-optimality of the achievable scheme is  established in some cases. Further, many of the existing results in the literature can be obtained as special cases of this work.

In summary, the major contributions of this paper are as follows:
\begin{enumerate}
 \item Three outer bounds on the sum rate are derived, presented as Theorems \ref{theorem-outer1}, \ref{theorem-outer2} and \ref{theorem-outer3}. These theorems apply to all channel conditions when the channel coefficients are drawn from a continuous distribution such as the  Gaussian distribution. 
 \item The three theorems are specialized to the symmetric MIMO Gaussian IC to obtain outer bounds on the per user GDOF, stated as Lemmas \ref{lemma-outer1}, \ref{lemma-outer2} and \ref{lemma-outer3}.  To the best of the authors' knowledge, result derived here represents the tightest known outer bound on the per user GDOF of the  $K$ $(K>2)$ user symmetric MIMO Gaussian IC, except for some specific cases mentioned in Section~\ref{sec:compexisting}.
\item The scheme for providing side information employed in Theorem~\ref{theorem-outer2} is new. 
 \item An inner bound is derived for the symmetric MIMO Gaussian IC as a combination of the HK-scheme, IA, zero-forcing (ZF) receiving, and treating interference as noise. To the best of the authors' knowledge, the extension of the HK-scheme to the multiuser MIMO scenario presented here is new.
 \item The interplay between the HK-scheme and IA is explored from an achievable GDOF perspective.
 \item Lemmas \ref{lemma-outer1} and \ref{lemma-outer3} are used to establish the optimality of the achievable scheme, when $\frac{N}{M} < K \leq \frac{N}{M}+1$. The corresponding GDOF result in Lemma~\ref{lemma-outer2} establishes that treating interference as noise is GDOF optimal when $M=N$ and for all $K$, in the weak interference regime.
\end{enumerate}

The following notation is used in the sequel. Lower case or upper case letters are used to represent scalars. Small boldface letters represent a vector, whereas capital boldface letters represent matrices. $\mathbf{x}^{n} = \lsqb \mathbf{x}_{1}^{T}, \mathbf{x}_{2}^{T}, \ldots,\mathbf{x}_{n}^{T}\rsqb^{T}$ represents a long vector consisting of the sequence of vectors $\mathbf{x}_{i}$, $i = 1, 2, \ldots, n$. $h(.)$ represents differential entropy, $I( \cdot;\cdot)$ represents mutual information, $\mathbf{I}_{L}$ is the $L\times L$ identity matrix, and $\text{blkdiag}(\mathbf{H}_{11}, \mathbf{H}_{22}, \ldots,\mathbf{H}_{L,L})$  represents a matrix which is obtained by block diagonal concatenation of matrices $\mathbf{H}_{11},\mathbf{H}_{22},\ldots,\mathbf{H}_{L,L}$. 

The rest of the paper is organized as follows. Section \ref{sec:prelim} describes formally the system model and defines GDOF. In Section \ref{sec:outerbound}, three outer bounds are derived and specialized to the case of a symmetric MIMO Gaussian IC.  Section \ref{sec:innerbound} presents the main results on the achievable GDOF performance. In Sec. \ref{sec:discussion}, some numerical examples are considered to obtain better insight into the bounds and to compare the performance of the various schemes. Concluding remarks are offered in Sec. \ref{sec:conclusion}; and the proofs of the lemmas and theorems are presented in the Appendix.
\section{Preliminaries}\label{sec:prelim}
\subsection{System Model}\label{sec:sysmod}
Consider a MIMO GIC with $K$ transmitter-receiver pairs, with $M_{i}$ antennas at the $i$-th transmitter and $N_{j}$ antennas at the $j$-th receiver. Let $\Hji$ represent the  $N_{j} \times M_{i}$ channel gain matrix from transmitter $i$ to receiver $j$. The channel coefficients are assumed to be drawn from a continuous distribution such as the Gaussian distribution. The received signal at the $j$-th receiver, denoted $\Yj$, is modeled as 
\begin{equation}
\Yj=\Hjj\mathbf{x}_{j}+\sum_{i=1, i \neq j}^K \Hji\mathbf{x}_{i} +  \Zj,
\end{equation}
where $\Zj$ is the complex symmetric Gaussian noise vector, distributed as $\Zj\sim \mathcal{CN} (\mathbf{0}, \mathbf{I}_{N_j})$ and  $\mathbf{x}_{i}$ is the signal transmitted by the $i$-th user, satisfying $\mean{\mathbf{x}_i^H\mathbf{x}_i} = P_i$. As in past work on the MIMO GIC, global channel state information is assumed to be available at every node. For the symmetric case considered later in the paper, with a slight abuse of notation, $\Hji\:(j \neq i)$ is replaced with $\sqrt{\rho^{\alpha}}\Hji$ and $\Hjj$ is replaced with $\sqrt{\rho}\Hjj$. The quantity $\rho^{\alpha_{ji}}$ represents the received signal power from user $i$; and in the symmetric case, it is assumed that $\alpha_{ji}=1$ when $j=i$ and $\alpha_{ji}=\alpha$ otherwise. That is, $\alpha > 0$ represents the ratio of the logarithm of the INR to the logarithm of the SNR. For the inner bound, attention will be restricted to the symmetric case with $M$ antennas at every transmitter and $N$ antennas at every receiver, with $M \le N$. Further, it is assumed that  $\mean{\mathbf{x}_i \mathbf{x}_i^H}$ is full rank and $\mean{\mathbf{x}_i^H\mathbf{x}_i} = 1$.
\subsection{Generalized Degrees of Freedom (GDOF)}\label{sec:gdof}
The GDOF, introduced in \cite{etkin1}, is an asymptotic quantity in the limit of high SNR and INR. For symmetric case, it is defined as:
\begin{equation}
 d(\alpha) = \frac{1}{K}\lim_{\substack{\rho\rightarrow \infty}}\frac{C_{\Sigma}(\rho,\alpha)}{\log\rho}, 
\end{equation}
and $C_{\Sigma}(\rho,\alpha)$ is the sum capacity of the $K$ user symmetric MIMO GIC defined above. When $\alpha=1$, the GDOF reduces to the degrees of freedom (DOF) defined in \cite{gou1}. 
\section{Outer Bound}\label{sec:outerbound}
In this section, three outer bounds on the sum rate of the $K$ user MIMO GIC are stated as Theorems \ref{theorem-outer1}, \ref{theorem-outer2} and \ref{theorem-outer3}. The bounds are general in the sense that they are valid for all values of the channel parameters. Then, the bounds are specialized to the case of the symmetric MIMO GIC to obtain outer bounds on the per user GDOF; these are stated as Lemmas \ref{lemma-outer1}, \ref{lemma-outer2} and  \ref{lemma-outer3}. Finally, the overall outer bound on the GDOF is obtained by taking the minimum of the three outer bounds and the interference free GDOF bound of $\min(M,N)$ per user. 

The first outer bound is obtained by considering cooperation among subsets of users. The idea of using cooperation among users has been explored in \cite{ghasemi1} for obtaining outer bounds on the DOF of the $K$-user MIMO GIC. However, it turns out that cooperation by itself is not sufficient for obtaining outer bounds on the sum rate of the $K$-user symmetric MIMO GIC. When $\alpha \neq 1$, the symmetric assumption on the resulting $2$-user GIC is no longer valid when the users are allowed to cooperate among themselves. Hence, this technique cannot be directly used to obtain an outer bound on the GDOF or the sum rate. It is necessary to provide a judiciously chosen signal as  side information to a subset of the receivers in addition to cooperation, to convert the system into a MIMO $Z$-GIC, whose capacity cannot be worse than the original MIMO IC. Then, an outer bound on the $Z$-GIC is derived. Taking the minimum of the outer bounds obtained by considering all possible combinations of cooperating users results in an outer bound on the sum rate of the MIMO GIC.

Thus, the $K$-user system is divided into two disjoint groups; group-1 containing $L_1\:(0 \leq L_{1} \leq K)$ users and group-2 containing $L_2\:(0 \leq L_{2} \leq K)$ users, with $L \triangleq L_1 + L_2$ such that $0 <  L \leq K$. The receivers within a given group are provided the messages of the other users in the same group, due to which,  interference between users within a group is eliminated. In group-1, all $L_1$ users are allowed to cooperate among themselves but they experience interference from group-2. Similarly, users in group-2 are allowed to cooperate among themselves. In group-2, all the receivers are given the messages of users $1, \ldots, L_1$ by  a genie as side information. As a result, group-2 does not see any interference from the users in group-1. To simplify the equation, it is assumed that each transmitter is equipped with $M$ antennas and each receiver is equipped with $N$ antennas in stating Theorem~\ref{theorem-outer1}. 
\begin{theorem}
\label{theorem-outer1}
The sum rate of the $K$-user  MIMO GIC is upper bounded as follows:
\begin{eqnarray}
 & \mysum R_{i} & \leq \log \labs \Iden_{L_{1}N} +  \Honebar\Ponebar\Honebar^{H} + \Htwobar\Ptwobar\Htwobar^{H}\rabs + \nonumber \\
& & \qquad \log \labs \Iden_{L_2N} + \Hthreebar\Ptwobar^{1/2} \lcb \Iden_{L_2M} + \Ptwobar^{1/2}\Htwobar^{H}\Htwobar\Ptwobar^{1/2} \rcb^{-1}\Ptwobar^{1/2}\Hthreebar^{H}\rabs, \label{eq:outerth1} 
\end{eqnarray}
\begin{eqnarray}
& & \text{where } L_{1} + L_{2} = L \leq K, \: 0 \leq L_{1} \leq K, \: 0 \leq L_{2} \leq K, \nonumber \\
& & \Honebar \triangleq \text{blockdiag}(\mathbf{H}_{11}, \mathbf{H}_{22}, \ldots, \mathbf{H}_{L_{1},L_{1}}), \:
\Hthreebar \triangleq \text{blockdiag}(\mathbf{H}_{L_{1}+1,L_{1}+1}, \mathbf{H}_{L_{1}+2,L_{1}+2}, \ldots,  \mathbf{H}_{L,L}),  \nonumber \\
& & \Htwobar \triangleq \lsqb\begin{array}{llll}
\mathbf{H}_{1,L_{1}+1} & \mathbf{H}_{1,L_{1}+2} &\cdots & \mathbf{H}_{1,L} \\
\mathbf{H}_{2,L_{1}+1} & \mathbf{H}_{2,L_{1}+2} &\cdots & \mathbf{H}_{2,L} \\
& \vdots & \vdots &  \\
\mathbf{H}_{L_1,L_{1}+1} & \mathbf{H}_{L_1,L_{1}+2} &\cdots & \mathbf{H}_{L_1,L} \\
           \end{array}\rsqb, \nonumber \\ 
& & \Ponebar \triangleq \text{blockdiag}(\mathbf{P}_{1}, \mathbf{P}_{2}, \ldots  \mathbf{P}_{L_{1}}), 
 \Ptwobar \triangleq \text{blockdiag}(\mathbf{P}_{L_{1} + 1}, \mathbf{P}_{L_{1} + 2}, \ldots \mathbf{P}_{L_{2}}), \overline{\mathbf{H}}_{ij} \in \Complex^{L_{i}N \times L_{j}M},\nonumber \\
& &  \Hij \in \Complex^{N \times M}, \mathbf{P}_{j} \in \Complex^{M \times M} \text{ is the input covariance matrix of } jth \text{ user and } \overline{\mathbf{P}}_{j} \in \Complex^{L_{j}M \times L_{j}M}, j=1,2. \label{eq:outerth1b}
\end{eqnarray}
\end{theorem}
\begin{proof}
See Appendix \ref{sec:appendouter1-th}.
\end{proof}
Recall that, in order to obtain (\ref{eq:outerth1}), $L_1$ users in group-1 are allowed to cooperate, while in group-2 $L_2$ users are allowed to cooperate with each other. There are $3^K - 2^{K+1}$ ways of choosing the user groups for cooperation (each user can be in group-1, group-2, or neither, and both groups should have at least one user). Hence, the minimum sum rate obtained out of all possible ways of cooperation leads to the tightest outer bound on the sum rate obtainable from this method. Since the users have different power constraints and users see different SNRs and INRs, obtaining a closed-form outer bound becomes a formidable task. However, for the symmetric case, a simplified solution exists, as given by the following Lemma.
\begin{lemma}
\label{lemma-outer1}
In the symmetric case, the upper bound of Theorem \ref{theorem-outer1} can be expressed as an upper bound on the per user GDOF as follows:
\begin{enumerate}
 \item When $M \leq N$ and $0 \leq \alpha \leq 1$:
     \begin{eqnarray}
 d(\alpha)  \leq \displaystyle\min_{L_{1}, L_{2}} \frac{1}{L} \lsqb L_1M + \min \lcb r, L_{1}(N - M)\rcb\alpha + (L_2M - r)^{+}  \right. \nonumber \\
+ \left. \min\lcb r, L_2N - (L_2M - r)^{+}\rcb(1-\alpha)\rsqb.  \label{eq:outerlm1a}
\end{eqnarray}
 \item When $M \leq N$ and $ \alpha > 1$:
\begin{eqnarray}
d(\alpha)  \leq \displaystyle\min_{L_{1}, L_{2}} \frac{1}{L} \lsqb r \alpha + \min\lcb L_1M, L_1N - r \rcb + (L_2M - r)^{+}\rsqb. \label{eq:outerlm1b}
\end{eqnarray}
 \item When $M > N$ and $0 \leq \alpha \leq 1$:
\begin{eqnarray}
& d(\alpha) &\leq \displaystyle\min_{L_{1}, L_{2}} \frac{1}{L} \lsqb L_1N + \min \lcb L_2N, (L_2M - r)^{+}\rcb + \min\lcb \min\lcb L_2N,r \rcb, \right. \right. \nonumber \\
& & \qquad \qquad \qquad \left. \left. L_2N -\min \lcb L_2N, (L_2M - r)^{+}\rcb \rcb (1-\alpha)\rsqb. \label{eq:outerlm1c}
\end{eqnarray}
 \item When $M > N$ and $ \alpha > 1$:
\begin{eqnarray}
& d(\alpha) &\leq \displaystyle\min_{L_{1}, L_{2}} \frac{1}{L} \lsqb L_1N + r(\alpha - 1) + \min\lcb L_2N, (L_2M - r)^{+}\rcb\rsqb, \label{eq:outerlm1d}
\end{eqnarray}
\end{enumerate}
where $r \triangleq \min\lcb L_2M, L_1N \rcb$.
\end{lemma}
\begin{proof}
See Appendix \ref{sec:appendouter1-lm}.
\end{proof}
The result below provides another outer bound on the sum rate, by providing side information in the form of a noisy version of the intended message at the receivers. 
To the best of the authors' knowledge the scheme for providing side information employed here new, and leads to the tightest known bounds for some parameter values as mentioned in Theorem~\ref{th:compareouter}.. 
Let $\mathbf{s}_{j,\mathcal{B}} \triangleq \displaystyle\sum_{i \in \mathcal{B}}\Hji\mathbf{x}_{i} + \mathbf{z}_{j},$  where $\mathcal{B} \subseteq \{1,2,\ldots,K\}$ is a subset users. Then, user $1$ is provided $\mathbf{s}_{2,1}$ and user $K$ is provided $\mathbf{s}_{K-1,K}$. 
Users $i = 2, 3, \ldots, K-1$ are provided $\mathbf{s}_{i-1,i}$ and $\mathbf{s}_{i+1,i}$ in succession to obtain two sets of rate bounds.  
It turns out that by doing so, all the negative differential entropy terms containing a signal component cancel out, leading to the outer bound given by Theorem \ref{theorem-outer2} below. Further remarks on the choice of side information are offered in Section~\ref{sec:discussion}.
\begin{theorem}
\label{theorem-outer2}
For the $K$-user MIMO GIC, the following rate bound is applicable:
\begin{eqnarray}
& & R_{1} + 2\displaystyle\sum_{i=2}^{K-1}R_{i} + R_{K} \nonumber \\
& & \leq \displaystyle\sum_{i=1}^{K-1}\log\labs  \Iden_{N_i} + \sumneq \Hij\Pj\Hij^{H} + \Hii\Pdashi^{1/2}\lb \Iden_{M_i} + \Pdashi^{1/2}\Hplus^{H}\Hplus\Pdashi^{1/2} \rb^{-1}\Pdashi^{1/2}\Hii^{H}\rabs  \nonumber \\
& &  + \displaystyle\sum_{i=2}^{K}\log\labs \Iden_{N_i} + \sumneq \Hij\Pj\Hij^{H} + \Hii\Pdashi^{1/2}\lb \Iden_{M_i} + \Pdashi^{1/2}\Hminus^{H}\Hminus\Pdashi^{1/2} \rb^{-1}\Pdashi^{1/2}\Hii^{H}\rabs. \label{eq:outerth2}
\end{eqnarray}
\end{theorem}
\begin{proof}
See Appendix \ref{sec:appendouter2-th}. 
\end{proof}
\emph{Remark:} Note that the above theorem presents a bound on $R_{1} + 2\sum_{i=2}^{K-1}R_{i} + R_{K}$, rather than on the sum rate, i.e., $\sum_{i=1}^{K} R_i$. Clearly, one can obtain $\frac{K(K-1)}{2}$ inequalities of the form (\ref{eq:outerth2}), for each possible choice of the first and $K^{\text{th}}$ user. Bounds on the sum rate can then be obtained from the above by summing all such inequalities and dividing by $\frac{3K(K-1)}{2}$.
\begin{lemma}
\label{lemma-outer2}
In the symmetric case, the sum rate upper bound of Theorem \ref{theorem-outer2} can be reduced to the following per user GDOF upper bound:
\begin{eqnarray}
& & d(\alpha) \leq \left\{\begin{array}{lll}
r_{\min}(1-\alpha) + \min\lcb r^{'}, r_{\max} - r_{\min}\rcb \alpha & \mbox{for $0 \leq \alpha \leq \frac{1}{2} $} \\
r^{'}\alpha + \min\lcb r_{\min}, r_{\max} - r^{'}\rcb(1-\alpha) & \mbox{for $\frac{1}{2} \leq \alpha \leq 1 $}
                               \end{array}\right. ,\label{eq:outerlm2}
\end{eqnarray}
where $r_{\min} \triangleq \min\lcb M,N\rcb$, $r_{\max} \triangleq \max\lcb M,N \rcb$ and $r^{'} \triangleq \min\lcb N, (K-1)M\rcb$.
\end{lemma}
\begin{proof}
See Appendix \ref{sec:appendouter2-lemma}.
\end{proof}
The third outer bound is based on providing each receiver with side information comprised of a noisy version of a carefully chosen part of the interference experienced by it. For the SIMO case, and when $K = N + 1$, this bound reduces to the outer bound presented in \cite{gou3}.
\begin{theorem}
\label{theorem-outer3}
For the $K$-user MIMO GIC, the following rate bound is applicable:
\begin{eqnarray}
& & R_{1} + \displaystyle\sum_{i=2}^{K-1}R_{i} + R_{K} \nonumber \\
& &  \leq \log\labs \Iden_{N_1} + \displaystyle\sum_{j=2}^{K}\Honej\Pj\Honej^{H} + \Honeone \Pone^{1/2}\lcb \Iden_{M_1}+ \Pone^{1/2}\Hkone^{H}\Hkone\Pone^{1/2} \rcb^{-1}\Pone^{1/2}\Honeone^{H}\rabs  \nonumber \\
& & + \displaystyle\sum_{i=2}^{K-1}\log \labs \Iden_{N_i} + \Hionebar\Pionebar^{1/2}\lcb \Iden_{\Mri} + \Pionebar^{1/2}\Hkonebar^{H}\Hkonebar\Pionebar^{1/2}\rcb^{-1}\Pionebar^{1/2}\Hionebar^{H}  \right. \nonumber \\ 
& & \left. +\Hitwobar\Pitwobar^{1/2}\lcb \Iden_{\Msi}+ \Pitwobar^{1/2}\Honekbar^{H}\Honekbar\Pitwobar^{1/2} \rcb^{-1}\Pitwobar^{1/2}\Hitwobar^{H}\rabs  \nonumber \\
& & +\displaystyle\sum_{i=2}^{K-1}\log \labs \Iden_{N_{i}} + \Hithreebar\Pithreebar^{1/2} \lcb \Iden_{\Mri^{'}} + \Pithreebar^{1/2}\Honetwobar^{H}\Honetwobar\Pithreebar^{1/2}\rcb^{-1}\Pithreebar^{1/2}\Hithreebar^{H}  \right.\nonumber \\
& & \left.+\Hifourbar\Pifourbar^{1/2}\lcb \Iden_{\Msi^{'}} + \Pifourbar^{1/2}\HKthreebar^{H}\HKthreebar\Pifourbar^{1/2} \rcb^{-1}\Pifourbar^{1/2}\Hifourbar^{H} \rabs   \nonumber \\
& & + \log\labs \Iden_{N_K} + \displaystyle\sum_{j=1}^{K-1}\HKj\Pj\HKj^{H} + \Hkk\PK^{1/2}\lcb \Iden_{M_K} + \PK^{1/2}\Honek^{H}\Honek\PK^{1/2} \rcb^{-1}\PK^{1/2}\Hkk^{H}\rabs, \nonumber \\
\label{eq:outerth3}
\end{eqnarray}
where
\begin{eqnarray}
& & \Hionebar \triangleq \lsqb \begin{array}{llll}
                     \mathbf{H}_{i1} & \mathbf{H}_{i2} &\ldots &\mathbf{H}_{ii}
                    \end{array} \rsqb, \quad \Hitwobar\triangleq \lsqb \begin{array}{llll}
		\mathbf{H}_{i,i+1} &\mathbf{H}_{i,i+2} &\ldots &\mathbf{H}_{iK} \end{array} \rsqb, \nonumber \\
& & \Hkonebar \triangleq \lsqb \begin{array}{llll}
                   \mathbf{H}_{K1} & \mathbf{H}_{K2} &\ldots &\mathbf{H}_{Ki}
                  \end{array} \rsqb,\: \Honekbar \triangleq \lsqb \begin{array}{llll}
                   \mathbf{H}_{1,i+1} & \mathbf{H}_{1,i+2} &\ldots &\mathbf{H}_{1K}
                  \end{array} \rsqb, \nonumber \\
& & \Honetwobar \triangleq \lsqb\begin{array}{llll}
                   \mathbf{H}_{1K} &\mathbf{H}_{12} &\ldots &\mathbf{H}_{1i}
                  \end{array}\rsqb ,\: \HKthreebar \triangleq \lsqb\begin{array}{llll}
		 \mathbf{H}_{K1} &\mathbf{H}_{K,i+1} &\ldots \mathbf{H}_{K,K-1} \end{array}\rsqb, \nonumber \\
& & \Hithreebar \triangleq \lsqb\begin{array}{llll}
                        \mathbf{H}_{iK} &\mathbf{H}_{i2} &\ldots &\mathbf{H}_{ii}
                       \end{array}\rsqb, \: \Hifourbar \triangleq \lsqb\begin{array}{llll}
			\mathbf{H}_{i1} &\mathbf{H}_{i,i+1} &\ldots &\mathbf{H}_{i,K-1} \end{array}\rsqb, \nonumber \\
& & \Pionebar \triangleq  \text{blockdiag}\lb  \mathbf{P}_{1}\: \mathbf{P}_{2}\: \ldots\: \mathbf{P}_{i} \rb, \: \Pitwobar \triangleq \text{blockdiag}\lb \mathbf{P}_{i+2} \:\mathbf{P}_{i+3} \: \ldots \: \mathbf{P}_{K}\rb, \nonumber \\
& & \Pithreebar\triangleq \text{blockdiag}\lb \mathbf{P}_{K} \: \mathbf{P}_{2} \: \ldots \: \mathbf{P}_{i} \rb, \Pifourbar \triangleq \text{blockdiag}\lb \mathbf{P}_{1} \: \mathbf{P}_{i+1} \: \ldots \: \mathbf{P}_{K-1} \rb \nonumber, \nonumber \\
& &  \Mri \triangleq \dimone, \: \Msi \triangleq \dimtwo, \: \Mri^{'}\triangleq \displaystyle\sum_{j=2}^{i}M_{j} + M_{K} \text{ and } \Msi^{'}\triangleq M_{1} + \displaystyle\sum_{j=i+1}^{K-1}M_{j}.\label{eq:outerth3b}
\end{eqnarray}
\end{theorem}
\begin{proof}
See Appendix \ref{sec:appendouter3-th}.
\end{proof}
\emph{Remark:} A bound on the sum rate $(\sum_{i=1}^{K}R_{i})$ can be obtained in a similar manner as in Theorem~\ref{theorem-outer2}. 
The above result can be used to obtain an outer bound of the GDOF of the $K$-user symmetric MIMO GIC only for $\frac{N}{M} < K \leq \frac{N}{M} + 1$ because the form of the above outer bound results in rank deficient matrices when $K > \frac{N}{M} + 1$, which make finding the inverse and computing the GDOF complicated.
\begin{lemma}
\label{lemma-outer3}
In the symmetric case, when $\frac{N}{M} < K \leq \frac{N}{M}+1$, the sum rate upper bound of Theorem \ref{theorem-outer3} can be expressed as an upper bound on the per user GDOF as follows:
\begin{eqnarray}
& & d(\alpha) \leq \left\{\begin{array}{lll}
M(1-\alpha) + \frac{1}{K-1}\lb N-M \rb \alpha &\mbox{for $0 \leq \alpha \leq \frac{1}{2} $} \\ 	
M\alpha + \frac{1}{K-1}(N-M)(1-\alpha) &\mbox{for $\frac{1}{2} \leq \alpha \leq 1 $}. \\
\end{array}\right. \label{eq:outerlm3}
\end{eqnarray}
\end{lemma}
\begin{proof}
See Appendix \ref{sec:appendouter3-lemma}.
\end{proof}
The overall outer bound is obtained by taking minimum of the outer bounds in Lemmas~\ref{lemma-outer1}, \ref{lemma-outer2} and~\ref{lemma-outer3}. Due to minimization involved in Lemma~\ref{lemma-outer1}, analytical characterization of the outer bound is not possible in all cases. However, in Theorem~\ref{th:compareouter} below, an expression for the combined outer bound is obtained when $K \geq N+M$ and $\frac{N}{M} < K \leq \frac{N}{M}+1$. Also, a unified expression is presented for case $\frac{N}{M}+1 < K < M+N$, when $\frac{N}{M}$ is integer-valued.  In stating the theorem, three interference regimes are considered, as in the past work~\cite{etkin1, gou3, tarokh1}. The result follows by first analytically solving the minimization in Lemma~\ref{lemma-outer1} and then carefully comparing the three outer bounds to determine which bound is tightest for different values of $K, M, N$ and $\alpha$. 
\begin{theorem}\label{th:compareouter}
The outer bound on the per user GDOF of the $K$-user symmetric MIMO $(M \leq N)$ GIC, obtained by taking the minimum of the outer bounds derived in this work, is
\begin{enumerate}
 \item When $(K \geq M + N)$ or $(\frac{N}{M} + 1 < K < M + N$, where $\frac{N}{M}$ is an integer):
\begin{enumerate}
 \item Weak interference regime $(0 \leq \alpha \leq \frac{1}{2})$: When $MN < N^2-M^2$, Lemma \ref{lemma-outer1} is active, otherwise Lemma~\ref{lemma-outer2} is active, and the outer bound is of the following form:
 \begin{equation}
 d(\alpha) \leq \lcb\begin{array}{l l}
	 M - \frac{M^2\alpha}{M+N} &\mbox{for $MN < N^2-M^2$} \\
         M(1-\alpha) + (N-M)\alpha & \mbox{for $MN \geq N^2-M^2$.}
\end{array}\right. \label{eq:minouter1}
 \end{equation}
 \item Moderate interference regime $(\frac{1}{2} \leq \alpha \leq 1)$:
 \begin{enumerate}
 \item When $MN < N^2 - M^2$, Lemma~\ref{lemma-outer1} is active, and the outer bound is of the following form:
         \begin{equation}
          d(\alpha) \leq  M - \frac{M^2\alpha}{M+N}. \label{eq:minouter2a}
         \end{equation}
\item When $MN \geq N^2 - M^2$, Lemma \ref{lemma-outer2} is active for $\frac{1}{2} \leq \alpha \leq \frac{M(M+N)}{N(M+N)+M^2}$, whereas Lemma \ref{lemma-outer1} is active for $\frac{M(M+N)}{N(M+N)+M^2} < \alpha \leq 1$,  and the outer bound becomes
        \begin{equation}
         d(\alpha) \leq \lcb\begin{array}{l l}
	 N\alpha &\mbox{for $\frac{1}{2} \leq \alpha \leq \frac{M(M+N)}{N(M+N)+M^2}$} \\
         M - \frac{M^2\alpha}{M+N} & \mbox{for $\frac{M(M+N)}{N(M+N)+M^2} < \alpha \leq 1$.}
\end{array}\right. \label{eq:minouter2}
 \end{equation}
 \end{enumerate}
\item High interference regime $(\alpha \geq 1)$: In this case, Lemma \ref{lemma-outer1} is active and the outer bound is of the following form:
\begin{equation}
d(\alpha) \leq \lcb\begin{array}{l l}
	 \frac{MN\alpha}{M+N} &\mbox{for $1 \leq \alpha \leq \frac{M+N}{N}$} \\
         M & \mbox{for $\alpha > \frac{M+N}{N}$.}
\end{array}\right. \label{eq:minouter3}
\end{equation}
\end{enumerate}
\item When $\frac{N}{M} < K \leq \frac{N}{M}+1$:
\begin{enumerate}
\item Weak interference regime $(0 \leq \alpha \leq \frac{1}{2})$: In this case, Lemma \ref{lemma-outer3} is active and the outer bound is of the following form:
\begin{equation}
d(\alpha) \leq M(1-\alpha) + \frac{1}{K-1}(N-M)\alpha. \label{eq:minouter4}
\end{equation}
\item Moderate interference regime $(\frac{1}{2} \leq \alpha \leq 1)$: Lemma \ref{lemma-outer3} is active for $\frac{1}{2} \leq \alpha \leq \frac{K}{2K-1}$, and Lemma \ref{lemma-outer1} is active for $\frac{K}{2K-1}< \alpha \leq 1$.
The outer bound becomes
\begin{equation}
         d(\alpha) \leq \lcb\begin{array}{l l}
	 M\alpha + \frac{1}{K-1}(N-M)(1-\alpha) &\mbox{for $\frac{1}{2} \leq \alpha \leq \frac{K}{2K-1}$} \\
         M(1-\alpha) + \frac{N\alpha}{K} & \mbox{for $\frac{K}{2K-1}< \alpha \leq 1$.}
\end{array}\right. \label{eq:minouter5}
 \end{equation} 
\item High interference regime $(\alpha \geq 1)$: In this case, Lemma \ref{lemma-outer1} is active and the outer bound is of the following form:
\begin{equation}
d(\alpha) \leq \lcb\begin{array}{l l}
	 \frac{1}{K}\lsqb N + (K-1)M(\alpha-1)\rsqb &\mbox{for $1 \leq \alpha \leq \frac{2KM-(M+N)}{(K-1)M}$} \\
         M & \mbox{for $\alpha \geq \frac{2KM-(M+N)}{(K-1)M}$.}
\end{array}\right. \label{eq:minouter6}
\end{equation}
\end{enumerate}
\end{enumerate}
\end{theorem}
\begin{proof}
See Appendix \ref{sec:th-combineouter}.
\end{proof}
The following lemmas are used in derivation of the outer bound.
\begin{lemma}\cite{mahesh1}\label{lemmaused1}
Let $\mathbf{R}_{1}$ and $\mathbf{R}_{2}$ be $N \times N$ covariance matrices with rank $r_{1}$ and $r_{2}$, respectively. Let $\mathbf{R}_{1} = \mathbf{U}_{1}\Lambda_{1}\mathbf{U}_{1}^{H}$ and $\mathbf{R}_{2} =  \mathbf{U}_{2}\Lambda_{2}\mathbf{U}_{2}^{H}$ represent the EVD of $\mathbf{R}_{1}$ and $\mathbf{R}_{2}$, with $\mathbf{U}_{1} \in \Complex^{N\times r_{1}}$ and $\mathbf{U}_{2} \in \Complex^{N\times r_{2}}$. If $\text{rank}[ \mathbf{U}_{1} \: \mathbf{U}_{2}] = \min(r_{1} + r_{2}, N)$, then for $\eta \geq \beta$,
\begin{eqnarray}
 & J_{1} & \triangleq \log|\Iden_{N}+\rho^{\eta}\mathbf{R}_{1}+\rho^{\beta}\mathbf{R}_{2}| \nonumber \\
 & & = r_{1}\eta\log\rho + \min\lcb r_{2},N-r_{1}\rcb\beta\log\rho+ \oone.
\end{eqnarray}
\end{lemma}
\begin{lemma}\cite{poor1}\label{lemmaused2}
Let $\mathbf{x}^{n}$ and $\mathbf{y}^{n}$ be two sequences of random vectors and let $\mathbf{x}^{*}, \mathbf{y}^{*}, \mathbf{\hat{x}} \text{ and } \mathbf{\hat{y}}$ be Gaussian vectors with covariance matrices satisfying
\begin{eqnarray}
& & \text{Cov}\left[\begin{array}{c}
\mathbf{\hat{x}} \\
\mathbf{\hat{y}} 
           \end{array}\right] = \frac{1}{n}\displaystyle\sum_{i=1}^{n}\text{Cov}\left[\begin{array}{c}
\mathbf{x}_{i} \\
\mathbf{y}_{i}
           \end{array}\right]\preceq\text{Cov}\left[\begin{array}{c}
\mathbf{x}^{*} \\
\mathbf{y}^{*}
           \end{array}\right], \nonumber 
\end{eqnarray}
then we get the following bounds
\begin{eqnarray}
& & h(\mathbf{x}^{n}) \leq nh(\mathbf{\hat{x}}) \leq nh(\mathbf{x}^{*}), \nonumber \\
& & h(\mathbf{y}^{n}|\mathbf{x}^{n}) \leq nh(\mathbf{\hat{y}}|\mathbf{\hat{x}}) \leq nh(\mathbf{y}^{*}|\mathbf{x}^{*}). \nonumber
\end{eqnarray}
\end{lemma}
\begin{lemma}\cite{mahesh1}\label{lemmaused3}
Let $0 \preceq \mathbf{G}_{1} \preceq \mathbf{G}_{2}$ and $\mathbf{0} \preceq \mathbf{A}$ be positive semi-definite matrices of size $N\times N$. For any given $\pi \in \Real^{+}$, 
\begin{equation}
\mathbf{G}_{1}\lcb \Iden + \pi \mathbf{G}_{1}\mathbf{A}\mathbf{G}_{1} \rcb ^{-1}\mathbf{G}_{1} \preceq  \mathbf{G}_{2}\lcb \Iden + \pi \mathbf{G}_{2}\mathbf{A}\mathbf{G}_{2} \rcb ^{-1}\mathbf{G}_{2}.
\end{equation}
\end{lemma}
\begin{lemma}\cite{mahesh2}\label{lemmaused4}
Let $\mathbf{R}_{1}$, $\mathbf{R}_{2}$ and $\mathbf{R}_{3}$ be $N \times N$ covariance matrices with rank $r_{1}$, $r_{2}$ and $r_3$, respectively. Let $\mathbf{R}_{i} = \mathbf{U}_{i}\Lambda_{i}\mathbf{U}_{i}^{H}$ represents the EVD of $\mathbf{R}_{i}$, with $\mathbf{U}_{i} \in \Complex^{N\times r_{i}}$. If $\text{rank}\lsqb\mathbf{U}_{1} \:\: \mathbf{U}_{2} \rsqb = \min\lcb r_{1} + r_{2}, N\rcb$ and $\text{rank}\lsqb\mathbf{U}_{1} \:\: \mathbf{U}_{2}\:\: \mathbf{U}_{3} \rsqb = \min\lcb r_{1} + r_{2} + r_{3}, N\rcb$, then for $\eta \geq \beta \geq \gamma$,
\begin{eqnarray}
 & J_{1} & \triangleq \log|\Iden_{N}+\rho^{\eta}\mathbf{R}_{1}+\rho^{\beta}\mathbf{R}_{2}+\rho^{\gamma}\mathbf{R}_{3}| \nonumber \\
 & & = r_{1}\eta\log\rho + \min\lcb r_{2},N-r_{1}\rcb \beta\log\rho + \min\lcb r_{3}, (N-r_{1} - r_{2})^{+}\rcb \gamma\log\rho + \oone.
\end{eqnarray}
\end{lemma}
\section{Inner Bound}\label{sec:innerbound}
In this section, an inner bound is derived for the $K$-user symmetric MIMO $(M \leq N \text{ and } KM > N)$\footnote{Note that, if $KM \le N$, one can trivially achieve the interference-free GDOF of $M$ per user, by using a ZF receiver.} GIC. The main results are stated as theorems; and the proofs are provided in the Appendix. Also, the detailed discussion and interpretation of the results is relegated to the next section. For vector space IA, the channel is required to be time-varying \cite{gou1}. The results for the HK-scheme, treating interference as noise and ZF-receiving are applicable in both the time-varying and the constant-channel cases.
Before stating the inner bounds, the following known results on the achievable DOF using IA and Zero-Forcing (ZF) reception are recapitulated.
\subsection{Known Results}
\subsubsection{Interference Alignment (IA)}\label{sec:interalignment}
In \cite{gou1}, it is shown that using vector space IA, the achievable per user DOF for a $K$-user symmetric MIMO GIC is
\begin{equation}
 d_{\text{IA}} = \frac{R}{R+1}\min\lcb M,N\rcb, \text{ when } K > R, \text{ where } R = \left\lfloor\frac{\max\lcb M,N\rcb}{\min\lcb M,N\rcb}\right\rfloor. \label{eq:inneralign1}
\end{equation}
It requires global channel knowledge at every node and the channel to be time varying. 
\subsubsection{Zero-Forcing (ZF) Receiving}
The achievable DOF by ZF-receiving is given by:
\begin{equation}
d_{\text{ZF}} = \min\lcb M, \frac{N}{K} \rcb.  \label{eq:innerzf1}
\end{equation}

Note that, for vector space IA and ZF-receiving, the relative strength between the signal and interference does not matter, and hence the above DOF is achievable for all values of $\alpha$.
\subsection{Treating Interference as Noise}\label{sec:innerintasnoise}
Treating interference as noise is one of the simplest methods of dealing with interference, and may work well when the interference is weak. The following theorem summarizes the GDOF obtained by treating interference as noise.
\begin{theorem}The following per user GDOF is achievable for the $K$-user symmetric MIMO GIC when interference is treated as noise:
\label{th:theorem1_intnoise}
\begin{enumerate}
 \item When $\frac{N}{M} < K \leq \frac{N}{M} + 1$, 
\begin{eqnarray}
& & d(\alpha) \geq  M + (N-KM)\alpha. \label{eq:teatintnoiseth1a}
\end{eqnarray}
 \item When $K > \frac{N}{M} + 1$,
\begin{eqnarray}
& & d(\alpha) \geq M(1-\alpha).  \label{eq:teatintnoiseth1b}
\end{eqnarray}
\end{enumerate}
\end{theorem}
\begin{proof}
See Appendix \ref{sec:th-intasnoise}.
\end{proof}
\subsection{Han-Kobayashi (HK) Scheme}\label{sec:innerHKscheme}
In this section, an achievable GDOF is derived by extending the HK-scheme to the symmetric $K$-user MIMO GIC. As in past work in the two-user and SIMO case \cite{etkin1}, \cite{tarokh1}, and \cite{gou3}, three different interference regimes are considered: strong, moderate, and weak interference. A key idea in the proof is to minimize the achievable GDOF per user from the common part of the message over all possible subsets of users, which does not enter into the picture in the $2$ user case considered in past work. Also, the results stated in this subsection are applicable even when $\frac{N}{M}$ is not an integer.
\subsubsection{Strong Interference Case $(\alpha \geq 1)$} \label{sec:strong-int}
When $\alpha \geq 1$, each receiver can decode both the unintended messages as well as  the intended message. Hence, a $K$ user MAC channel is formed at each receiver, and the achievable rate region is the intersection of the $K$ MAC regions obtained. This results in the following inner bound on the per user GDOF. 
\begin{theorem}
 \label{th-highint1}
In the strong interference case $(\alpha \geq 1)$, the following per user GDOF is achievable by the HK-scheme:
\begin{enumerate}
 \item When $\frac{N}{M} < K \leq \frac{N}{M} + 1$,
\begin{eqnarray}
& & d(\alpha) \geq \min\lcb M, \frac{1}{K}\lsqb(K-1)M\alpha + N - (K-1)M \rsqb\rcb. \label{eq:strongth1a}
\end{eqnarray}
 \item When $K > \frac{N}{M} + 1$,
\begin{eqnarray}
& & d(\alpha) \geq \min\lcb M, \frac{N\alpha}{K}\rcb. \label{eq:strongth1b}
\end{eqnarray}
\end{enumerate}
\end{theorem}
\begin{proof}
See Appendix \ref{sec:appendhigint}.
\end{proof}
\subsubsection{Moderate Interference Case $(1/2 \leq \alpha \leq 1)$}\label{sec:moderate-int}
In the moderate interference regime, an achievable scheme based on HK-type message splitting is as follows. The transmitter $j$ splits its message $W_{j}$ into two sub-messages: a common message $W_{c,j}$ that is decodable at every receiver, and a private message $W_{p,j}$ that is required to be decodable only at the desired receiver. The common message is encoded using a Gaussian code book with rate $R_{c,j}$ and power $P_{c,j}$. Similarly, the private message is encoded using a Gaussian code book with rate $R_{p,j}$ and power $P_{p,j}$. Further, it is assumed that the rates are symmetric, i.e., $R_{c,j} = R_{c}$ and $R_{p,j}=R_{p}$. Also, $P_{c,j} = P_{c}$ and $P_{p,j} = P_{p}$. The powers on the private and common messages satisfy the constraint $P_{c} + P_{p} = 1$. 
The codewords are transmitted using superposition coding, and hence, the transmitted signal $X_{j}$ is a superposition of the private message and the public message.

Similar to \cite{tarokh1}, the power in the private message is set such that it is received at the noise floor of the unintended receivers, resulting in $\text{INR}_p = 1$. Coupled with the transmit power constraint at each of the users, the SNRs of the common and private parts at the desired receiver (denoted $\text{SNR}_c$ and $\text{SNR}_p$) and the INRs of the common and private parts  at unintended receivers (denoted $\text{INR}_c$ and $\text{INR}_p$) are given by 
\begin{eqnarray}
 \text{SNR}_{c} = \rho -\rho^{1-\alpha},  \text{SNR}_{p} = \rho^{1-\alpha},  
\text{INR}_{c} = \rho^{\alpha} - 1,  \text{INR}_{p} = 1. \label{eq:modrateint1a}
\end{eqnarray}
The transmit covariance of the common message is assumed to be the same as that of the private message. The decoding order is such that the common message is decoded first, followed by the private message. While decoding the common message, all the users' private messages are treated as noise (including its own private message). 
The rate achievable from the private message is obtained by treating all the other users' private messages as noise. 

The GDOF is contributed by both the private and public parts of the message:
\begin{equation}
d(\alpha) = d_{p}(\alpha) + d_{c}(\alpha), 
\end{equation}
where $d_{p}(\alpha)$ and $d_{c}(\alpha)$ are the GDOF contributed by the private and public parts of the message, respectively. The following theorem summarizes the per user GDOF achievable by this scheme.
\begin{theorem}\label{th:moderateint1}
In the moderate interference regime ($1/2 \le \alpha \le 1$), the $K$-user symmetric MIMO GIC achieves the following per user GDOF:
\begin{enumerate}
 \item When $\frac{N}{M} < K \leq \frac{N}{M} + 1$,
\begin{eqnarray}
&d(\alpha)  \geq M(1-\alpha) + \min\lcb \frac{N\alpha}{K}, \frac{\lsqb M  \lcb (2K-1)\alpha - K\rcb + N(1-\alpha)\rsqb}{K-1} \rcb. \label{eq:moderateth1a} 
\end{eqnarray}
\item When $K > \frac{N}{M}+1$,
\begin{eqnarray}
d(\alpha)  \geq M(1-\alpha) + \min\lcb\frac{N\alpha}{K},\frac{\lsqb N\alpha  - M(1-\alpha)\rsqb}{K-1}\rcb. \label{eq:moderateth1b}
\end{eqnarray}
\end{enumerate}
\end{theorem}
\begin{proof}
See Appendix \ref{sec:appendmoderateint1}.
\end{proof} 
\subsubsection{Weak Interference Case $(0 \leq \alpha \leq 1/2)$}\label{sec:weak-int}
In this case, the received SNR and INR of the common and private messages are set the same way as in the moderate interference regime. The per user GDOF achieved is summarized in the following theorem.
\begin{theorem}\label{th:lowinter1}
In the weak interference regime $\lb 0 \le \alpha \le \frac{1}{2}\rb$, the $K$-user symmetric MIMO GIC achieves the following per user GDOF:
\begin{equation}
 d(\alpha) \geq M(1-\alpha) + \frac{1}{K-1}(N-M)\alpha. \label{eq:weakth1}
\end{equation}
\end{theorem}
\begin{proof}
See Appendix \ref{sec:appendlowint1}.
\end{proof}
\emph{Remark:} The expressions for the GDOF in \eqref{eq:moderateth1b} and \eqref{eq:weakth1} are different
because $\alpha \geq 1-\alpha$ in the former case while $\alpha \le 1-\alpha$ the latter case, and this has been used to simplify the equations. 
\subsection{Achievable GDOF as a Combination of HK-scheme, IA, ZF-Receiving and Treating Interference as Noise}
In this subsection, the performance of the various schemes considered above is consolidated in terms of the parameters $\alpha$, $K$, $M$ and $N$. Here, the channel is assumed to be time-varying in order to include IA along with the other schemes considered in this paper. Further, to simplify the presentation, it is assumed that $\frac{N}{M}$ is an integer in Theorems \ref{innermaxhigh-th}, \ref{innermaxmoderate-th} and \ref{innermaxweak-th}. It is straightforward to extend the result to non-integer values of $\frac{N}{M}$; however, the expressions become cumbersome with the floor of $\frac{N}{M}$ appearing in the expressions, and offer little additional insight on the achievable GDOF. In Theorem \ref{thm:NbyMnotInteger}, the achievable GDOF for the case where $K \ge \frac{N}{M} + 4$ is presented without assuming that $\frac{N}{M}$ is an integer. 

The achievable per user GDOF with IA and ZF-receiving are:
\begin{eqnarray}
& & d_{\text{IA}} = \frac{MN}{M+N}, \label{eq:innercombine1} \\
 \text{and }& &  d_{\text{ZF}} = \min\lcb M, \frac{N}{K}\rcb. \label{eq:innercombine2}
\end{eqnarray}
The maximum achievable GDOF for different interference regimes are stated in following Theorems.
\begin{theorem}\label{innermaxhigh-th}
The achievable GDOF in high interference case $\lb \alpha > 1\rb$ obtained by taking maximum of all the schemes considered in this work is
\begin{enumerate}
 \item When $\frac{N}{M} < K \leq \frac{N}{M} + 1$,
\begin{eqnarray}
& & d(\alpha) \geq \lcb\begin{array}{ll}
    \frac{1}{K}[\alpha(K-1)M+N-(K-1)M] &\mbox{for $ 1 < \alpha < \frac{M(2K-1)-N}{M(K-1)}$} \\ 	
    M & \mbox{for $ \alpha \geq \frac{M(2K-1)-N}{M(K-1)}$}
 \end{array}\right. \label{eq:innermaxhigh1}
\end{eqnarray}
\item When $K > \frac{N}{M} + 1$,
\end{enumerate}
\begin{eqnarray}
 d(\alpha) \geq \left\{\begin{array}{lll}
     \frac{MN}{M+N }    &\mbox{for $ 1 \leq \alpha \leq \frac{KM}{M+N}$} \\ 	
     \frac{\alpha N}{K} & \mbox{for $ \frac{KM}{M+N} < \alpha < \frac{KM}{N} $} \\
     M                   & \mbox{for $\alpha \geq \frac{KM}{N} $}
 \end{array}\right.  \label{eq:innermaxhigh2}
\end{eqnarray}
\end{theorem}
\begin{proof}
See Appendix \ref{sec:innermax-high}.
\end{proof}
\begin{theorem}\label{innermaxmoderate-th}
The achievable GDOF in the moderate interference case $\lb \frac{1}{2}\leq \alpha \leq 1\rb$ obtained by taking maximum of all the achievable schemes considered in this work is
\begin{enumerate}
 \item When $\frac{N}{M} < K \leq \frac{N}{M} + 1$,
 \begin{eqnarray}
& & d(\alpha) \geq \lcb\begin{array}{ll}
    M(1-\alpha) + \frac{1}{K-1}\lsqb M\lcb \alpha(2K-1)-K\rcb+N(1-\alpha)\rsqb&\mbox{for $\frac{1}{2} \leq \alpha \leq \frac{K}{2K-1} $} \\  M(1-\alpha) + \frac{N\alpha}{K}& \mbox{for $ \frac{K}{2K-1} \leq \alpha \leq 1 $} \end{array}\right. \label{eq:innermaxmoderate1}
\end{eqnarray}
\item When $ \frac{N}{M} + 1 < K \leq \frac{N}{M} + 2$,
\begin{eqnarray}
d(\alpha) \geq \lcb\begin{array}{lll}
      M(1-\alpha) +  \frac{N\alpha - M(1-\alpha)}{K-1} &\mbox{for $\frac{1}{2} \leq \alpha \leq \frac{KM}{N+KM}$}  \\
      M(1-\alpha) + \frac{N\alpha}{K} &\mbox{for $\frac{KM}{N+KM} \leq \alpha \leq \frac{KM^{2}}{(M+N)(KM-N)}$} \\
      \frac{MN}{M+N} &\mbox{for $\frac{KM^{2}}{(M+N)(KM-N)} < \alpha \leq 1$}
                   \end{array}\right. \label{eq:innermaxmoderate3}
\end{eqnarray}
\item When $K > \frac{N}{M}+2$,
\begin{eqnarray}
d(\alpha) \geq \frac{MN}{M+N} \label{eq:innermaxmoderate4}
\end{eqnarray}
\end{enumerate}
\end{theorem}
\begin{proof}
See Appendix \ref{sec:innermax-moderate}.
\end{proof}
\begin{theorem}\label{innermaxweak-th}
The achievable GDOF in the weak interference case $\lb0 \leq \alpha \leq \frac{1}{2}\rb$ obtained by taking maximum of all the achievable schemes considered in this work is
 \begin{enumerate}
\item When $K > \frac{N}{M}+2$,
\begin{eqnarray}
& d(\alpha) & \geq \left\{\begin{array}{ll}
    M(1-\alpha) + \frac{1}{K-1}(N-M)\alpha &\mbox{for $0 \leq \alpha \leq \frac{M^{2}}{M(N+M)-\frac{N^{2}-M^{2}}{K-1}} $} \\ 	
    \frac{NM}{N+M} & \mbox{for $ \frac{M^{2}}{M(N+M)-\frac{N^{2}-M^{2}}{K-1}} < \alpha \leq \frac{1}{2} $}
 \end{array}\right. \label{eq:innermaxweak2}
\end{eqnarray}
\item When $\frac{N}{M} < K \leq \frac{N}{M}+2$,
\begin{eqnarray}
& d(\alpha) & \geq  M(1-\alpha) + \frac{1}{K-1}(N-M)\alpha \label{eq:innermaxweak3}
\end{eqnarray}
\end{enumerate}
\end{theorem}
\begin{proof}
See Appendix \ref{sec:innermax-weak}.
\end{proof}
From the expressions in the previous section, it is easy to see that the maximum of the achievable GDOF from the HK-scheme and IA outperforms the achievable GDOF from treating interference as noise or ZF-receiving for all values of $M$, $N$, $K$ and $\alpha$. The following result follows from carefully comparing the achievable GDOF from the HK-scheme and IA in the weak, moderate, and strong interference cases. 
\begin{theorem} \label{thm:NbyMnotInteger}
Recall that $R \triangleq \left\lfloor \frac{N}{M} \right\rfloor$. When $K \geq \frac{N}{M} + 4$, the $K$-user symmetric MIMO GIC achieves the following per-user GDOF. 
\begin{enumerate}
\item When $R = 1$: 
\begin{enumerate}
\item The HK-scheme is active in the weak interference case and in the initial part of the moderate interference case, and achieves 
\begin{equation}
d(\alpha) \geq \lcb\begin{array}{l l}
    M(1-\alpha) + \frac{(N-M)\alpha}{K-1} &\mbox{for $ 0 \le  \alpha \le \frac{1}{2}$} \\ 	
    M(1-\alpha) + \frac{N\alpha - M(1-\alpha)}{K-1}  & \mbox{for $\frac{1}{2} < \alpha \le \frac{(K-1) - (R+1)}{\lb R+1\rb \lb \lb K-1\rb - \lb\tfrac{N}{M} + 1\rb \rb}.$}
 \end{array}\right. 
 \end{equation}
 \item IA is active in the later part of the moderate interference case and the initial part of the strong interference case, and achieves
\begin{equation}
d(\alpha) \geq 
    \frac{MR}{R+1}   \quad \mbox{for $\frac{(K-1) - (R+1)}{\lb R+1\rb \lb \lb K-1\rb - \lb\tfrac{N}{M} + 1\rb \rb} < \alpha \le \frac{MKR}{N(R+1)}.$}
\end{equation} 
\item The HK-scheme is active in the later part of the strong interference case, and achieves
\begin{equation}
d(\alpha) \geq \lcb\begin{array}{l l}
    \frac{N\alpha}{K} &\mbox{for $ \frac{MKR}{N(R+1)} <  \alpha \le \frac{MK}{N}$} \\ 	
    M & \mbox{for $ \alpha > \frac{MK}{N}.$}
 \end{array}\right. 
 \end{equation}
\end{enumerate}
\item When $R > 1$:
\begin{enumerate}
\item The HK-scheme is active in the initial part of the weak interference case, and achieves
\begin{equation}
d(\alpha) \geq M(1-\alpha) + \frac{(N-M)\alpha}{K-1}    \quad \mbox{for $0 \le \alpha \le \frac{(K-1)}{\lb R+1\rb \lb K - \tfrac{N}{M} \rb}.$}
\end{equation} 
\item IA is active in the later part of the weak interference case, in the moderate interference case, and in the initial part of the strong interference case, and achieves
\begin{equation}
d(\alpha) \geq 
    \frac{MR}{R+1}   \quad \mbox{for $\frac{(K-1)}{\lb R+1\rb \lb K - \tfrac{N}{M} \rb} < \alpha \le \frac{MKR}{N(R+1)}.$}
\end{equation} 
\item The HK-scheme is active for the later part of the strong interference case, and achieves
\begin{equation}
d(\alpha) \geq \lcb\begin{array}{l l}
    \frac{N\alpha}{K} &\mbox{for $ \frac{MKR}{N(R+1)} <  \alpha \le \frac{MK}{N}$} \\ 	
    M & \mbox{for $ \alpha > \frac{MK}{N}.$}
 \end{array}\right. 
\end{equation}
\end{enumerate}
\end{enumerate}
\end{theorem}
\begin{proof}
See Appendix \ref{sec:new result}.
\end{proof}
The above theorem is interesting because it exactly characterizes the regimes of $\alpha$ where the HK-scheme and IA are active for $K \geq \frac{N}{M} + 4$, even when $\frac{N}{M}$ is not an integer. 
It can be used, for example, to study the effect of varying the number of transmit and receive antennas on the achievable GDOF, or the scaling of the achievable GDOF as the number of transmit and receive antennas per user is increased while keeping their ratio fixed.
\subsection{Tightness of the Outer Bounds}\label{sec:tightouter}
\begin{corollary}\label{tightcor}
The outer bound is tight when $M=N$ for any $K$, in the weak interference regime $(0 \leq \alpha \leq \frac{1}{2})$, and when $\frac{N}{M} < K \leq \frac{N}{M}+1$, for all values of $\alpha$.
\end{corollary}
\begin{proof}
See Appendix \ref{sec:tightoutercorl}.
\end{proof}
\section{Discussion on The Bounds}\label{sec:discussion}
%
\subsection{Comparison with Existing Results}\label{sec:compexisting}
Some observations on how the inner and outer bounds on the GDOF derived above stand in relation to existing work are as follows:
\begin{enumerate}
\item When $M = 1$ and $K=N+1$, the Theorem \ref{lemma-outer3} and the HK-scheme in Section \ref{sec:innerHKscheme} reduce to the corresponding SIMO GDOF results in \cite{gou3}. 
\item When $K=2$, the inner and outer bounds derived here reduce to the corresponding two-user symmetric GDOF result in \cite{tarokh1}.
\item When $M=N=1$ and $K=2$, the inner and outer bounds derived here reduce to the corresponding GDOF results derived in \cite{etkin1}.
\item When $M=N=1$, the inner bounds derived here match with the result in \cite{jafar2} only in the weak interference regime. In \cite{jafar2} assumes the constant IC model and uses multilevel coding with nested lattice structure to achieve a higher GDOF. Also, the outer bound derived here reduces to the $K$-user symmetric SISO GIC GDOF result in \cite{jafar2}.
\item When $\alpha = 1$, the cooperative outer bound of Lemma~\ref{lemma-outer1} matches with the DOF outer bound in \cite{ghasemi1} for many cases of $K$, $M$ and $N$ (e.g., $K=3, M=2, N=5$). Theorem~\ref{lemma-outer1} uses genie-aided message sharing in addition to cooperation, to handle the $\alpha\neq 1$ cases. The bound in \cite{ghasemi1} only requires cooperation, due to which it is lower for some values of $M, N$ and $K$. Hence, when $\alpha = 1$, the minimum of the outer bound derived here and the one in \cite{ghasemi1} is plotted in the graphs presented in the next subsection. The outer bound derived here does not match with the DOF-optimal outer bound in \cite{wang2} for the $K = 3$ and $\frac{N}{M}+1 < K \leq \frac{M+N}{\text{gcd}(M,N)}$ case. The outer bound in \cite{wang2} uses the concept of subspace alignment chains to identify the extra dimension to be provided by  a genie to a receiver, which does not easily generalize to arbitrary $K$, $M$, $N$ and~$\alpha$.
\item When $K=2$, the outer bound in Lemma~\ref{lemma-outer1} reduces to the DOF outer bound on MIMO $Z$-GIC in \cite{zhengdao1}. See Appendix \ref{sec:compz} for details.
\end{enumerate} 
In Fig.~\ref{fig:compare1}, the outer bound derived in this work is compared with some of the existing results mentioned above.
\subsection{Numerical Examples}
Now, some numerical examples are considered to get better insight into the bounds for various values of $K, M, N$, and $\alpha$.

In Fig. \ref{fig:outercomp1}, the outer bounds on the per user GDOF in Lemmas \ref{lemma-outer1}, \ref{lemma-outer2} and \ref{lemma-outer3} are contrasted as a function $\alpha$, for $(M,N) = (2,2)$ and $(2,4)$. 
When $K=3$ and $(M,N) = (2,2)$, the outer bound in Lemma~\ref{lemma-outer2} is active in the weak interference regime and the initial part of the moderate interference regime. The outer bound in Lemma~\ref{lemma-outer1} is not tight in this regime, as a result of the genie giving too much information to the receiver. As the interference level increases, it is necessary to provide the unintended message completely as in Theorem~\ref{theorem-outer1} to obtain a tractable outer bound; and hence Lemma~\ref{lemma-outer1} is active in the later part of the moderate interference regime and the high interference regime.
As the number of receive dimensions increases $(N=4)$, the outer bound in Lemma~\ref{lemma-outer2} is found to be loose. Hence, another outer bound is derived, where a carefully chosen part of the interference is provided as side information to the receiver, as in Theorem~\ref{theorem-outer3}. The corresponding GDOF outer bound in Lemma~\ref{lemma-outer3} is tight in the weak interference regime $(0 \leq \alpha \leq \frac{1}{2})$ and in the initial part of the moderate interference regime $(\frac{1}{2} \leq \alpha \leq \frac{3}{5})$. For $\alpha > \frac{3}{5}$, the outer bound in Lemma~\ref{lemma-outer1} is active, as in the previous case.

In Fig. \ref{fig:P1_K3M2N2}, the per user GDOF is plotted versus $\alpha$ for $K=3$ and $M=N=2$. The achievable GDOF by IA (curve labeled as \texttt{IA}), HK-scheme (curve labeled as \texttt{HK-scheme}), treating interference as noise (curve labeled as \texttt{Intf. as noise}) and ZF-receiving (curve labeled as \texttt{ZF-receiving}) are plotted along with the outer bound (curve labeled as \texttt{Outer bound}). In the low interference regime, treating interference as noise coincides with the outer bound. In this case, treating interference as noise performs as well as the HK-scheme and the outer bound in Lemma \ref{lemma-outer2} is active. Also, IA and ZF-receiving are suboptimal in this regime. At $\alpha = {1}/{2}$, IA, HK and treating interference as noise all coincide. In the moderate interference regime, the flat segment is  due to IA. In the initial part of the moderate interference regime, the HK-scheme and IA coincide but in the later part, IA performs better than HK-scheme. IA performs better than the other schemes and is optimal at $\alpha = {1}/{2}$ and $1$. In terms of outer bounds, initially, the side-information based bound of Lemma \ref{lemma-outer2} is active, and as $\alpha$ increases, the cooperative bound of Lemma \ref{lemma-outer1} is active. 
In the high interference regime, IA initially performs the best, and as $\alpha$ increases, the HK-scheme performs the best, and finally achieves the interference free GDOF. There exists gap between the inner and outer bounds in the moderate and high interference regimes.

In Fig. \ref{fig:P_KMN}, the achievable per-user GDOF is plotted against $\alpha$ for the $K=3$ user symmetric MIMO GIC with various antenna configurations and compared with existing results. The inner bound derived in this paper is compared with the result in \cite{jafar2} for the symmetric SISO GIC case and with the result in \cite{gou3} for the symmetric SIMO GIC with $K=N+1$. Since the achievable GDOF in \cite{jafar2}  is discontinuous at $\alpha=1$, it is represented by the filled circle in the plot. Note that the scheme in \cite{jafar2} assumes that the channel remains constant over time. Hence, the performance of IA is not included in the comparison.  Further, the achievable GDOF is plotted for the $2 \times 3, 2 \times 4, 2 \times 5$  and $2 \times 6$ antenna configurations.  Also, the outer bound is plotted for these antenna configurations to verify the optimality of the inner bound. 

The figure illustrates the benefits of having additional antennas at the transmitter and receiver in improving the achievable GDOF. For the symmetric SISO GIC, the proposed inner bound matches with the result in \cite{jafar2} in the weak interference case. 
There exists a gap between the two schemes in the moderate interference case and in the initial part of the strong interference case, as noted in the previous subsection. 
For the SIMO case, the achievable GDOF of the proposed scheme matches with that of the scheme in \cite{gou3} and is also GDOF optimal. As receive antennas are added, in the strong interference regime, the HK-scheme achieves the interference-free GDOF at a smaller value of $\alpha$. In the $2\times 6$ system, as $N=KM$, ZF-receiving achieves the interference free GDOF for all values of $\alpha$. Finally, note that the inner bound is GDOF optimal for the $2 \times 4$, $2 \times 5$ and $2 \times 6$ symmetric MIMO GIC cases. In Fig. \ref{fig:P1_K4MN}, the achievable per-user GDOF is compared with the outer bound for many more cases for the $K=4$ user symmetric MIMO GIC with various antenna configurations. 

In Figs. \ref{fig:K3MN7} and  \ref{fig:K3MN10}, the per user achievable GDOF performance is compared for different antenna configurations with a total of $7$ and $10$ antennas per user pair, respectively. The figures illustrate the effect of different combinations of the number of antennas at the transmitter and receiver on the achievable GDOF. When the interference is either low or very high, an equal or nearly equal (in Fig.~\ref{fig:K3MN7}) distribution of antennas achieves the best GDOF. The behavior for intermediate values of $\alpha$ depends on the specific values of $M$, $N$, $K$ and $\alpha$. 
\subsection{Further Remarks}
From the expressions obtained for the bound, the following useful observations can be made. In particular, these insights are not be obtainable from the existing results for the two user symmetric MIMO GIC or the $K$-user symmetric SIMO GIC. 
\begin{enumerate}
\item The outer bounds on the sum rate in Theorems \ref{theorem-outer2} and \ref{theorem-outer3} hold for any number of transmit and receive antennas. Although Theorem~\ref{theorem-outer1} was presented for $M$ antennas at each transmitter and $N$ antennas at each receiver, it is straightforward to extend it to the case of arbitrary number of antennas at each transmitter and receiver. These results are new as there are no existing outer bounds on the sum rate of the $K$-user MIMO GIC for $K\geq 3$.
 \item No single outer bound on the GDOF is universally the tightest among the three. Theorem~\ref{th:compareouter} characterizes the performance of the outer bounds as a function of $K$, $M$, $N$ and $\alpha$ when $K\geq M+N$ and $\frac{N}{M} < K \leq \frac{N}{M}+1$, and when $\frac{N}{M}+1 < K < N+M$ for integer-valued~$\frac{N}{M}$.
 \item Treating interference as noise was known to be GDOF optimal in the weak interference regime in the two user SISO case \cite{etkin1}, two user symmetric MIMO case \cite{tarokh1} and the $K$-user SISO real-valued constant channel case \cite{jafar2}. The above results show that it is GDOF optimal in the weak interference regime only when $M=N$. When $N>M$, the HK-scheme performs better. Moreover,  the maximum of the HK-scheme and IA outperforms treating interference as noise and ZF-receiving for all values of $M$, $N$, $\alpha$ and $K$. 
\item When $\frac{N}{M} < K \leq \frac{N}{M} + 1$, Theorem~\ref{th:compareouter} establishes that the achievable scheme is GDOF optimal for all $\alpha$. The proof can be found in Corollary \ref{tightcor}. Moreover, the HK-scheme does not assume a time-varying channel, and hence it is optimal even for the constant channel case. 
 \item When $K>3$ and $M=N$, IA outperforms the HK-scheme for $\frac{1}{2}\leq \alpha \leq 1$. Also, IA is  GDOF optimal at $\alpha = \frac{1}{2}$ when $M=N$. 
 \item When $K > \frac{N}{M}+1$, depending on the value of $\alpha$, one or the other of the HK-scheme and IA performs the best. When $K \ge \frac{N}{M} + 4$,  Theorem~\ref{thm:NbyMnotInteger}  characterizes the interplay between the two schemes and determines the range of $\alpha$ for which either scheme is active.  
\item When $\frac{N}{M} < K \leq \frac{N}{M} + 1$, ZF-receiving coincides with the HK-scheme only at $\alpha=1$ when $K > 2$. 
In contrast, when $K=2$, ZF-receiving is optimal for $\alpha = \frac{1}{2}$ and $1$ (see \cite{tarokh1}). 
\end{enumerate}
In general, it is found that IA performs well over a fairly wide range of parameters around $\alpha=1$, and it offers a performance that does not depend on the interference level. Hence, it may be a good approach for managing the interference, especially when the number of receive antennas is comparable to the number of transmit antennas. As the number of receive dimensions increases, the HK-scheme becomes a better choice for interference management.
\section{Conclusion}\label{sec:conclusion}
This work derived inner and outer bounds on the GDOF of the $K$-user symmetric MIMO interference channel as a function of $\alpha = \log \text{INR} / \log \text{SNR}$. The outer bound was based on a combination of three schemes, one of which was derived using the notion of cooperation, and the two other outer bounds were based on providing partial side information at the receivers. The inner bound was derived using a combination of ZF-receiving, treating interference as noise, interference alignment, and the Han-Kobayashi scheme. Several interesting insights were obtained from the derived bounds. For example, it was found that when $M=N$, treating interference as noise performs as well as the HK scheme and outperforms both IA and the ZF bound. However, when  $N > M$, treating interference as noise is always suboptimal. For $K > \frac{N}{M}+1$, a combination of HK and IA performs the best in the moderate interference regime. Finally, when $\frac{N}{M} < K \leq \frac{N}{M}+1$, HK scheme is GDOF optimal for all values of $\alpha$. In contrast to two user IC, ZF-receiving is found to be GDOF optimal at $\alpha=1$ when $\frac{N}{M} < K \leq \frac{N}{M}+1$. The outer bound was shown to be tight in the weak interference case ($0 \le \alpha \le \frac{1}{2}$) when $M=N$ for any $K$, and for all values of $\alpha$ when $\frac{N}{M} < K \le \frac{N}{M} + 1$.
\appendix
\subsection{Proof of Theorem \ref{theorem-outer1}}\label{sec:appendouter1-th}
Given the stated assumptions on user cooperation and the genie-provided side information, the system model becomes:
\begin{eqnarray}
& & \Yonebar = \Honebar\Xonebar + \Htwobar\Xtwobar + \Zonebar, \nonumber \\
& & \Ytwobar = \Hthreebar\Xtwobar + \Ztwobar, \label{eq:coopouter1} 
\end{eqnarray}
\begin{eqnarray}
& & \text{where } \Yonebar \triangleq \lsqb \begin{array}{lll}
                                  \mathbf{y}_{1}^{T}, \cdots, \mathbf{y}_{L_{1}}^{T}
                                    \end{array} \rsqb^{T}, \:
\Ytwobar \triangleq \lsqb \begin{array}{lll}
                                  \mathbf{y}_{L_{1} + 1}^{T}, \cdots, \mathbf{y}_{L}^{T}
                                    \end{array} \rsqb^{T}, \:
\Xonebar \triangleq \lsqb \begin{array}{lll}
                                  \mathbf{x}_{1}^{T}, \cdots, \mathbf{x}_{L_{1}}^{T}
                                    \end{array} \rsqb^{T}, \: \nonumber \\
& & \Xtwobar \triangleq \lsqb \begin{array}{lll}
                                  \mathbf{x}_{L_{1}+1}^{T}, \cdots, \mathbf{x}_{L}^{T}
                                    \end{array} \rsqb^{T},
 \Zonebar \triangleq \lsqb \begin{array}{lll}
                                  \mathbf{z}_{1}^{T}, \cdots, \mathbf{z}_{L_{1}}^{T}
                                    \end{array} \rsqb^{T} \text{ and } \:
\Ztwobar \triangleq \lsqb \begin{array}{lll}
                                  \mathbf{z}_{L_{1}+1}^{T}, \cdots, \mathbf{z}_{L}^{T}
                                    \end{array} \rsqb^{T}. \nonumber
\end{eqnarray}
Also, $\Hbar_{ij}$ are stacked channel matrices, as defined in the statement of the theorem. The above system model is equivalent to a two-user MIMO $Z$-interference channel with each transmitter having $L_{1}M$ and $L_{2}M$ antennas and each receiver having $L_{1}N$ and $L_{2}N$ antennas. The outer bound derived for this modified system is an outer bound for the $K$-user MIMO GIC. By using Fano's inequality, the sum rate of the modified system is upper bounded as given below:
\begin{eqnarray}
& n\mysum R_{i} -n\epsilon_{n} & \leq I\lb \Xonebar^{n};\Yonebar^{n}\rb + I \lb \Xtwobar^{n};\Ytwobar^{n} \rb, \nonumber \\
& & \stackrel{(a)}{\leq} I\lb \Xonebar^{n};\Yonebar^{n}\rb + I \lb \Xtwobar^{n};\Ytwobar^{n},\Sbar^{n} \rb, \text{ where } \Sbar  \triangleq \Htwobar\Xtwobar + \Zonebar, \nonumber \\ 
& & = h\lb \Yonebar^{n} \rb - h \lb \Yonebar^{n} | \Xonebar^{n}\rb + h \lb \Sbar^{n} \rb - h \lb \Sbar^{n}|\Xtwobar^{n} \rb + h\lb\Ytwobar^{n}|\Sbar^{n}\rb - h\lb\Ytwobar^{n}|\Sbar^{n}, \Xtwobar^{n}\rb, \nonumber \\
& & = h \lb\Yonebar^{n}\rb - h\lb\Sbar^{n}\rb + h\lb\Sbar^{n}\rb - h \lb \Zonebar^{n}\rb + h \lb \Ytwobar^{n}|\Sbar^{n}\rb - h\lb \Ztwobar^{n}\rb, \nonumber \\ 
& & \stackrel{(b)}{\leq} n h\lb \Yonebar^{*} \rb - n h \lb \Zonebar \rb + n h \lb \Ytwobar^{*} | \Sbar^{*}\rb - n h \lb \Ztwobar\rb, \nonumber \\
& \text{or }\mysum R_{i}  & \leq h\lb \Yonebar^{*} \rb -  h \lb \Zonebar \rb +  h \lb \Ytwobar^{*} | \Sbar^{*}\rb -  h \lb \Ztwobar\rb, \label{eq:coopouter2}
\end{eqnarray}
where \textit{(a)} is due to the genie giving side information to Receiver $2$ and \textit{(b)} follows from the Lemma  \ref{lemmaused2}. In the above equation, the superscript $*$ indicates that the inputs are i.i.d. Gaussian i.e., $\mathbf{x}_{i}^{*}\sim CN(\mathbf{0},\mathbf{\overline{P}}_{i})$ and the quantities $\Sbar^{*}, \Yonebar^{*} \text{ and } \Ytwobar^{*}$ are the signals obtained due to Gaussian inputs. Each term in (\ref{eq:coopouter2}) is simplified as follows:
\begin{eqnarray}
& & h\lb \Yonebar^{*} \rb = \log\labs \pi e \lsqb \Iden_{L_{1}N} +  \Honebar\Ponebar\Honebar^{H} + \Htwobar\Ptwobar\Htwobar^{H}\rsqb\rabs, \label{eq:coopouter3} 
\end{eqnarray}
\begin{eqnarray}
& & h \lb \Ytwobar^{*} | \Sbar^{*}\rb = \log\labs \pi e \Sigma_{\Ytwobar^{*}|\Sbar^{*}}\rabs, \label{eq:coopouter4} 
\end{eqnarray}
\text{where }
\begin{eqnarray}
 & \Sigma_{\Ytwobar^{*}|\Sbar^{*}} & = \expec\lsqb \Ytwobar^{*}\Ytwobar^{*H} \rsqb - \expec\lsqb \Ytwobar^{*}\Sbar^{*H}\rsqb \expec \lsqb \Sbar^{*}\Sbar^{*H}\rsqb^{-1}\expec\lsqb \Sbar^{*}\Ytwobar^{*H}\rsqb, \nonumber \\
& & = \Iden_{L_2N} + \Hthreebar\Ptwobar\Hthreebar^{H} - \Hthreebar\Ptwobar\Htwobar^{H} \lcb \Iden_{L_2M} + \Htwobar\Ptwobar\Htwobar^{H}\rcb^{-1}\Htwobar\Ptwobar\Hthreebar^{H}, \nonumber \\
& & = \Iden_{L_2N} + \Hthreebar\Ptwobar^{1/2}\lsqb \Iden_{L_2M} - \Ptwobar^{1/2}\Htwobar^{H}\lcb \Iden_{L_2M} + 
 \Htwobar\Ptwobar^{1/2}\Ptwobar^{1/2}\Htwobar^{H}\rcb^{-1}\Htwobar\Ptwobar^{1/2}\rsqb\Ptwobar^{1/2}\Hthreebar^{H}, \nonumber \\
& & = \Iden_{L_2N} + \Hthreebar\Ptwobar^{1/2} \lcb \Iden_{L_2M} + \Ptwobar^{1/2}\Htwobar^{H}\Htwobar\Ptwobar^{1/2} \rcb^{-1}\Ptwobar^{1/2}\Hthreebar^{H}.  \label{eq:coopouter5}
\end{eqnarray}
In the above, (\ref{eq:coopouter5}) is obtained using the Woodbury matrix identity \cite{seber1}.

The conditional differential entropy in (\ref{eq:coopouter4}) thus reduces to:
\begin{equation}
h \lb \Ytwobar^{*} | \Sbar^{*}\rb = \log\labs \pi e \lsqb \Iden_{L_2N} + \Hthreebar\Ptwobar^{1/2} \lcb \Iden_{L_2M} + \Ptwobar^{1/2}\Htwobar^{H}\Htwobar\Ptwobar^{1/2} \rcb^{-1}\Ptwobar^{1/2}\Hthreebar^{H}\rsqb \rabs. \label{eq:coopouter6}
\end{equation}
From (\ref{eq:coopouter3}) and (\ref{eq:coopouter6}), the sum rate bound in (\ref{eq:coopouter2}) becomes:
\begin{eqnarray}
 & \mysum R_{i} & \leq \log \labs \Iden_{L_{1}N} +  \Honebar\Ponebar\Honebar^{H} + \Htwobar\Ptwobar\Htwobar^{H}\rabs + \nonumber \\
& & \qquad \log \labs \Iden_{L_2N} + \Hthreebar\Ptwobar^{1/2} \lcb \Iden_{L_2M} + \Ptwobar^{1/2}\Htwobar^{H}\Htwobar\Ptwobar^{1/2} \rcb^{-1}\Ptwobar^{1/2}\Hthreebar^{H}\rabs, \label{eq:coopouter7}
\end{eqnarray}
which concludes the proof.
\subsection{Proof of Lemma \ref{lemma-outer1}}\label{sec:appendouter1-lm}
In the symmetric case, with a slight abuse of notation, the system model in (\ref{eq:coopouter1}) reduces to:
\begin{eqnarray}
& & \Yonebar = \sqrt{\rho}\Honebar\Xonebar + \sqrt{\rho^{\alpha}}\Htwobar\Xtwobar + \Zonebar, \nonumber \\
& & \Ytwobar = \sqrt{\rho}\Hthreebar\Xtwobar + \Ztwobar. \label{eq:lmcoopouter1}
\end{eqnarray}
Under the symmetric assumption, the sum rate in (\ref{eq:outerth1}) in Theorem \ref{theorem-outer1} is bounded as follows:
\begin{eqnarray}
 & \mysum R_{i}  & \leq \log \labs \Iden_{L_{1}N} +  \rho\Honebar\Ponebar\Honebar^{H} + \rho^{\alpha} \Htwobar\Ptwobar\Htwobar^{H}\rabs + \nonumber \\
& & \qquad \log \labs \Iden_{L_2N} + \rho\Hthreebar\Ptwobar^{1/2} \lcb \Iden_{L_2M} + \rho^{\alpha}\Ptwobar^{1/2}\Htwobar^{H}\Htwobar\Ptwobar^{1/2} \rcb^{-1}\Ptwobar^{1/2}\Hthreebar^{H}\rabs, \nonumber \\
& & \leq \log \labs \Iden_{L_{1}N} +  \rho\Honebar\Honebar^{H} + \rho^{\alpha} \Htwobar\Htwobar^{H}\rabs + \log \labs \Iden_{L_2N} + \rho\Hthreebar \lcb \Iden_{L_2M} + \rho^{\alpha}\Htwobar^{H}\Htwobar \rcb^{-1}\Hthreebar^{H}\rabs. \label{eq:lmcoopouter2}
\end{eqnarray}
Equation (\ref{eq:lmcoopouter2}) is obtained by using Lemma \ref{lemmaused3} and using the fact that $\log|\:.\:|$ is a monotonically increasing function on the cone of positive definite matrices.

Consider the following term in (\ref{eq:lmcoopouter2}):
\begin{eqnarray}
& & \Iden_{L_2N} + \rho\Hthreebar \lcb \Iden_{L_2M} + \rho^{\alpha}\Htwobar^{H}\Htwobar \rcb^{-1}\Hthreebar^{H} \nonumber \\
& & \stackrel{(a)}{=} \Iden_{L_2N} + \rho\Hthreebar \lcb \Iden_{L_2M} + \rho^{\alpha}\Uonetwobar\eigm_{12}\Uonetwobar^{H} \rcb^{-1}\Hthreebar^{H}, \nonumber \\
& & = \Iden_{L_2N} + \rho\Hthreebar\Uonetwobar \lcb \Iden_{L_2M} + \rho^{\alpha}\eigm_{12} \rcb^{-1}\Uonetwobar^{H}\Hthreebar^{H}, \nonumber \\
& & = \Iden_{L_2N} + \rho\Htildetwtwo \lcb \Iden_{L_2M} + \rho^{\alpha}\eigm_{12} \rcb^{-1}\Htildetwtwo^{H}, \quad \text{ where } \Htildetwtwo \triangleq \Hthreebar\Uonetwobar, \nonumber \\
& & \stackrel{(b)}{=} \Iden_{L_2N} + \rho\Htildetwtwo \lcb \Iden_{L_2M} + \rho^{\alpha}\lsqb\begin{array}{cc}
  \Sigma_{r} &\mathbf{0} \\
  \mathbf{0} &\mathbf{0}_{L_2M -r}
  \end{array}\rsqb \rcb^{-1}\Htildetwtwo^{H}, \qquad \text{where } r \triangleq \min\{L_2M, L_1N\}\nonumber \\ 
& & = \Iden_{L_2N} + \rho\Htildetwtwo \lsqb\begin{array}{cc}
  \lb \Iden_{r} + \rho^{\alpha}\Sigma_{r}\rb^{-1} &\mathbf{0} \\
  \mathbf{0} &\Iden_{L_2M -r}
  \end{array}\rsqb\Htildetwtwo^{H}, \nonumber \\
& & \stackrel{(c)}{=} \Iden_{L_2N} + \rho\lsqb \Htildetwtwo^{(a)} \quad \Htildetwtwo^{(b)} \rsqb \lsqb\begin{array}{cc}
  \lb \Iden_{r} + \rho^{\alpha}\Sigma_{r}\rb^{-1} &\mathbf{0} \\
  \mathbf{0} &\Iden_{L_2M -r}
  \end{array}\rsqb \lsqb \Htildetwtwo^{(a)} \quad \Htildetwtwo^{(b)} \rsqb^{H}, \nonumber \\
& & = \Iden_{L_2N} + \rho\Htildetwtwo^{(a)} \lb \Iden_{r} + \rho^{\alpha}\Sigma_{r}\rb^{-1}\Htildetwtwo^{(a)H} + \rho\Htildetwtwo^{(b)}\Iden_{L_2M -r}\Htildetwtwo^{(b)H}, \label{eq:lmcoopouter3}
\end{eqnarray}
where \textit{(a)} follows by taking EVD of $\Htwobar^{H}\Htwobar$, $\Uonetwobar \in \Complex^{L_2M \times L_2M}$, in \textit{(b)} $\Sigma_{r}$ contains non-zero singular values of $\Htwobar^{H}\Htwobar$ and $\mathbf{0}_{L_2M -r}$ is a zero matrix of dimension $(L_2M -r) \times (L_2M -r)$, in \textit{(c)} $\Htildetwtwo$ is partitioned into two sub-matrices $\Htildetwtwo^{(a)}$ and $\Htildetwtwo^{(b)}$ of dimensions $L_2N \times r$ and $L_2N \times (L_2M - r)$, respectively. 

Using (\ref{eq:lmcoopouter3}), and simplifying the outer bound in (\ref{eq:lmcoopouter2}) becomes
\begin{eqnarray}
& \mysum R_{i}  & \leq \log \labs \Iden_{L_{1}N} +  \rho\Honebar\Honebar^{H} + \rho^{\alpha} \Htwobar\Htwobar^{H}\rabs + \nonumber \\
& & \qquad \log \labs \Iden_{L_2N} + \rho\Htildetwtwo^{(a)} \lb \Iden_{r} + \rho^{\alpha}\Sigma_{r}\rb^{-1}\Htildetwtwo^{(a)H} + \rho\Htildetwtwo^{(b)}\Iden_{L_2M -r}\Htildetwtwo^{(b)H}\rabs \nonumber \\
& & = \log \labs \Iden_{L_{1}N} +  \rho\Honebar\Honebar^{H} + \rho^{\alpha} \Htwobar\Htwobar^{H}\rabs + \nonumber \\
& & \qquad \log \labs \Iden_{L_2N} + \rho^{1-\alpha}\Htildetwtwo^{(a)}\Sigma_{r}^{-1}\Htildetwtwo^{(a)H} + \rho\Htildetwtwo^{(b)}\Iden_{L_2M -r}\Htildetwtwo^{(b)H}\rabs + \oone. \label{eq:lmcoopouter4}
\end{eqnarray}
The above approximation holds at high SNR. The above equation can be simplified further, depending on the values of $M$, $N$ and $\alpha$.\\
\textbf{Case 1 $(M \leq N \:\text{and} \:0 \leq \alpha \leq 1)$}:\\
Using Lemma \ref{lemmaused1}, the outer bound in (\ref{eq:lmcoopouter4}) becomes:
\begin{eqnarray}
& & \mysum R_{i}  \leq r_{11}\log\rho + \min\lcb r_{12}, L_{1}N - r_{11}\rcb\alpha\log\rho + r_{22}^{(b)}\log\rho \nonumber \\
& & \qquad \qquad + \min\lcb r_{22}^{(a)}, L_2N - r_{22}^{(b)}\rcb(1-\alpha)\log\rho + \oone, \label{eq:lmcoopouter5}
\end{eqnarray}
where $r_{ij} \triangleq \text{rank}(\mathbf{\overline{H}}_{ij})$, $r_{22}^{(a)} \triangleq \text{rank}(\Htildetwtwo^{(a)})$ and $r_{22}^{(b)} \triangleq \text{rank}(\Htildetwtwo^{(b)})$. As the channel coefficients are drawn from a continuous distribution such as the Gaussian, the channel matrices are full rank with probability one. Hence, the outer bound in (\ref{eq:lmcoopouter5}) reduces to following form:
\begin{eqnarray}
 & \mysum R_{i} & \leq L_1M\log\rho + \min \lcb r, L_{1}N - L_1M\rcb\alpha\log\rho + (L_2M - r)\log\rho \nonumber \\
& & \qquad + \min\lcb \min\lcb L_2N, r\rcb, L_2N - (L_2M - r)\rcb(1-\alpha)\log\rho + \oone \nonumber \\
& & = L_1M\log\rho + \min \lcb r, L_{1}N - L_1M\rcb\alpha\log\rho + (L_2M - r)\log\rho \nonumber \\
& & \qquad + \min\lcb r, L_2N - (L_2M - r)\rcb(1-\alpha)\log\rho + \oone, \label{eq:lmcoopouter6}  \\
& & \text{where }  r \triangleq \min\{L_2M, L_1N\}. \nonumber
\end{eqnarray}
Hence, the sum GDOF is upper bounded as given below:
\begin{eqnarray}
d_{i_{1}} + \ldots + d_{i_{L}} \leq L_1M + \min \lcb r, L_{1}(N - M)\rcb\alpha + L_2M - r + \min\lcb r, L_2N - (L_2M - r)\rcb(1-\alpha). \label{eq:lmcoopouter7}
\end{eqnarray}
Note that $L$ users can be chosen among $K$ users in $\binom{K}{L}$ different ways, and each user appears in $\binom{K-1}{L-1}$ of these ways. By adding all inequalities like (\ref{eq:lmcoopouter7}), the sum GDOF is upper bounded as:
\begin{eqnarray}
& d_{i_{1}} + d_{i_{2}} + \ldots + d_{i_{K}}  & \leq \frac{K}{L}\lsqb L_1M + \min \lcb r, L_{1}(N - M)\rcb\alpha + (L_2M - r) + \right. \nonumber \\
 & & \left. \quad \min\lcb r, L_2N - (L_2M - r)\rcb(1-\alpha)\rsqb, \nonumber \\
& \text{or } d(\alpha) & \leq \frac{1}{L} \lsqb L_1M + \min \lcb r, L_{1}(N - M)\rcb\alpha + (L_2M - r)\right. \nonumber \\
& & \left. \quad + \min\lcb r, L_2N - (L_2M - r)\rcb(1-\alpha) \rsqb. \label{eq:lmcoopouter8}
\end{eqnarray}
Taking minimum of (\ref{eq:lmcoopouter8}) over all possible values of $L_{1}$ and $L_{2}$ results in Case 1 of Lemma  \ref{lemma-outer1}.\\
\textbf{Case 2 $(M \leq N \:\text{and}\: \alpha > 1)$}:\\
Hence, the outer bound in (\ref{eq:lmcoopouter4}) is simplified to following form using Lemma \ref{lemmaused1}:
\begin{eqnarray}
& \mysum R_{i} & \leq r\alpha\log\rho + \min\lcb L_1M, L_1N - r\rcb\log\rho + (L_2M - r)\log\rho + \oone. \label{eq:lmcoopouter9}
\end{eqnarray}
By following the same steps as in the previous case, the per user GDOF is upper bounded as given below:
\begin{eqnarray}
d(\alpha) \leq \frac{1}{L} \lsqb r\alpha + \min\lcb L_1M, L_1N - r\rcb + (L_2M - r)\rsqb. \label{eq:lmcoopouter10}
\end{eqnarray}
By taking minimum of (\ref{eq:lmcoopouter10}) over all possible values of $L_{1}$ and $L_{2}$ results in Case 2 of Lemma  \ref{lemma-outer1}.\\
\textbf{Case 3 $(M > N \:\text{and} \:0 \leq \alpha \leq 1)$}:\\
When $M > N$ and $0 \leq \alpha \leq 1$, the sum rate in (\ref{eq:lmcoopouter4}) reduces to following form by using Lemma  \ref{lemmaused1}:
\begin{eqnarray}
& \mysum R_{i} & \leq L_1N\log\rho +  \min \lcb L_2N, L_2M - r\rcb\log\rho + \nonumber \\
& & \quad\min\lcb \min\lcb L_2N,r \rcb, L_2N -\min \lcb L_2N, L_2M - r\rcb \rcb(1-\alpha)\log\rho + \oone\nonumber \\
& & = L_1N\log\rho + \min \lcb L_2N, L_2M - r\rcb\log\rho + \nonumber \\
& & \quad \min\lcb \min\lcb L_2N,r \rcb, L_2N -\min \lcb L_2N, L_2M - r\rcb \rcb(1-\alpha)\log\rho + \oone. \label{eq:lmcoopouter11}
\end{eqnarray}
Following the same steps as in Case 1, the per user GDOF is upper bounded as given below by using (\ref{eq:lmcoopouter11}):
\begin{eqnarray}
& d(\alpha) & \leq \frac{1}{L} \lsqb L_1N + \min \lcb L_2N, L_2M - r\rcb + \min\lcb \min\lcb L_2N,r \rcb, L_2N -\min \lcb L_2N, L_2M - r\rcb \rcb\right. \nonumber \\
& & \qquad \left. (1-\alpha)\rsqb. \label{eq:lmcoopouter12}
\end{eqnarray}
By taking minimum of (\ref{eq:lmcoopouter12}) over all possible values of $L_{1}$ and $L_{2}$ results in Case 3 of Lemma  \ref{lemma-outer1}.\\
\textbf{Case 4 $(M > N \:\text{and} \:\alpha \geq 1)$}:\\
Under this condition the outer bound in (\ref{eq:lmcoopouter4}) is simplified to following form by using Lemma \ref{lemmaused1}:
\begin{eqnarray}
& \mysum R_{i} & \leq r\alpha\log\rho + \min\lcb L_1N, L_1N - r\rcb\log\rho + \min\lcb L_2N, L_2M - r\rcb\log\rho + \oone \nonumber \\
& & =  r\alpha\log\rho + \lb L_1N - r\rb\log\rho + \min\lcb L_2N, L_2M - r\rcb\log\rho + \oone. \label{eq:lmcoopouter13}
\end{eqnarray}
Following the same steps as in case 1, the per user GDOF is upper bounded as:
\begin{eqnarray}
& d(\alpha) & \leq \frac{1}{L} \lsqb L_1N + r(\alpha - 1) + \min\lcb L_2N, L_2M - r\rcb\rsqb. \label{eq:lmcoopouter14}
\end{eqnarray}
Taking minimum of (\ref{eq:lmcoopouter14}) over all possible values of $L_{1}$ and $L_{2}$ results in Case 4 of Lemma \ref{lemma-outer1}. This completes the proof of Lemma \ref{lemma-outer1}.
\subsection{Proof of Theorem \ref{theorem-outer2}}\label{sec:appendouter2-th}
The signal received at receiver $i$ is:
\begin{equation}
\mathbf{y}_{i} = \Hii\mathbf{x}_{i} + \sumneq\Hij\mathbf{x}_{j} + \mathbf{z}_{j}. \label{eq:sideouter0}
\end{equation}
Define the following quantity
\begin{equation}
\mathbf{s}_{j,\mathcal{B}} \triangleq \displaystyle\sum_{i \in \mathcal{B}}\Hji\mathbf{x}_{i} + \mathbf{z}_{j}, \label{eq:sideouter0a}
\end{equation}
where $\mathcal{B} \subseteq \{1,2,\ldots,K\}$ is a subset of users.

The rate of first user is upper bounded as follows:
\begin{eqnarray}
& nR_{1} & \stackrel{(a)}{\leq} I(\mathbf{x}_{1}^{n};\mathbf{y}_{1}^{n})  + n\epsilon_{n}, \nonumber \\
& & \stackrel{(b)}{\leq} I(\mathbf{x}_{1}^{n};\mathbf{y}_{1}^{n},\mathbf{s}_{2,1}^{n}) + n\epsilon_{n}, \nonumber \\
& & \stackrel{(c)}{=} h(\mathbf{s}_{2,1}^{n}) - h(\mathbf{z}_{2}^{n}) + h(\mathbf{y}_{1}^{n}|\mathbf{s}_{2,1}^{n}) - h(\mathbf{y}_{1}^{n}|\mathbf{s}_{2,1}^{n},\mathbf{x}_{1}^{n}) + n\epsilon_{n},\nonumber \\
& & \stackrel{(d)}{\leq} h(\mathbf{s}_{2,1}^{n}) - h(\mathbf{z}_{2}^{n}) + h(\mathbf{y}_{1}^{n}|\mathbf{s}_{2,1}^{n}) - h(\mathbf{s}_{1,2}^{n}) + n\epsilon_{n}, \label{eq:sideouter9}
\end{eqnarray}
where \textit{(a)} follows due to Fano's inequality; \textit{(b)} is due to the genie giving side information to Receiver 1; \textit{(c)} follows from chain rule of mutual information and \textit{(d)} results due to the fact that the differential entropy can not increase by conditioning and the last differential entropy term is conditioned on $\mathbf{x}_{i},i = 1,\ldots, K$ and ~$i \neq 2 $.

Similarly, the rate of the $K$th user is upper bounded as follows:
\begin{eqnarray}
& nR_{K} & \leq I(\mathbf{x}_{K}^{n};\mathbf{y}_{K}^{n}) + n\epsilon_{n}, \nonumber \\
& & \leq I(\mathbf{x}_{K}^{n};\mathbf{y}_{K}^{n}, \mathbf{s}_{K-1,K}^{n}) + n\epsilon_{n}, \nonumber \\
& & \leq h(\mathbf{s}_{K-1,K}^{n})- h(\mathbf{z}_{K-1}^{n}) + h(\mathbf{y}_{K}^{n}|\mathbf{s}_{K-1,K}^{n}) - h(\mathbf{s}_{K,K-1}^{n}) + n\epsilon_{n}. \label{eq:sideouter10}
\end{eqnarray}
The rate of users $i = 2, 3, \ldots, K-1$ are upper bounded as follows:
\begin{eqnarray}
& nR_{i} & \leq \mutuali + n\epsilon_{n}, \nonumber \\
& & \leq I(\mathbf{x}_{i}^{n};\mathbf{y}_{i}^{n},\siminus^{n}) + n\epsilon_{n}, \nonumber \\
& & = h(\siminus^{n}) - h(\ziminus^{n}) + h(\Yi^{n}|\siminus^{n}) - h(\Yi^{n}|\siminus^{n},\myxi^{n}) + n\epsilon_{n} \nonumber \\
& & \leq  h(\siminus^{n}) - h(\ziminus^{n}) + h(\Yi^{n}|\siminus^{n}) - h(\Yi^{n}|\siminus^{n},\{\mathbf{x}_{j}^{n}\}_{j=1, j \neq i+1}^{K}) +  n\epsilon_{n},\nonumber \\
& & =  h(\siminus^{n}) - h(\ziminus^{n}) + h(\Yi^{n}|\siminus^{n}) - h(\mathbf{s}_{i,i+1}^{n}) +  n\epsilon_{n}. \label{eq:sideouter11}
\end{eqnarray}
Again, the rate of users $i = 2, 3, \ldots, K-1$ can also be bounded as given below:
\begin{eqnarray}
& nR_{i} & \leq \mutuali + n\epsilon_{n}, \nonumber \\
& & \leq I(\mathbf{x}_{i}^{n};\mathbf{y}_{i}^{n},\siplus^{n}) + n\epsilon_{n}, \nonumber \\
& & \leq h(\siplus^{n}) - h(\ziplus^{n}) + h(\Yi^{n}|\siplus^{n}) - h(\mathbf{s}_{i,i-1}^{n}) +  n\epsilon_{n}. \label{eq:sideouter12}
\end{eqnarray}
Summing all the inequalities in (\ref{eq:sideouter9}), (\ref{eq:sideouter10}), (\ref{eq:sideouter11}) and (\ref{eq:sideouter12}), the sum rate is bounded as follows:
\begin{eqnarray}
& & n\lsqb R_{1} + 2\displaystyle\sum_{i=2}^{K-1}R_{i} + R_{K} \rsqb  - n\epsilon_{n} \nonumber \\
& & = h(\mathbf{s}_{2,1}^{n}) - h(\mathbf{z}_{2}^{n}) + h(\mathbf{y}_{1}^{n}|\mathbf{s}_{2,1}^{n}) - h(\mathbf{s}_{1,2}^{n}) + \displaystyle\sum_{i=2}^{K-1}\lsqb h(\siminus^{n}) - h(\ziminus^{n}) + h(\Yi^{n}|\siminus^{n}) - h(\mathbf{s}_{i,i+1}^{n})\right. +\nonumber \\
& & \: \left.  h(\siplus^{n}) - h(\ziplus^{n}) + h(\Yi^{n}|\siplus^{n}) - h(\mathbf{s}_{i,i-1}^{n}) \rsqb + 
h(\mathbf{s}_{K-1,K}^{n})- h(\mathbf{z}_{K-1}^{n}) + h(\mathbf{y}_{K}^{n}|\mathbf{s}_{K-1,K}^{n}) - h(\mathbf{s}_{K,K-1}^{n}), \nonumber \\
& & \leq \displaystyle\sum_{i=1}^{K-1} h(\Yi^{n} | \siplus^{n}) + \displaystyle\sum_{i=2}^{K} h(\Yi^{n} | \siminus^{n}) -h(\mathbf{z}_{1}^{n}) - 2\displaystyle\sum_{i=2}^{K-1}h(\mathbf{z}_{i}^{n}) - h(\mathbf{z}_{K}^{n}), \nonumber \\
& & \stackrel{(a)}{\leq} n \displaystyle\sum_{i=1}^{K-1} h(\Yi^{*} | \siplus^{*}) + n\displaystyle\sum_{i=2}^{K} h(\Yi^{*} | \siminus^{*}) - nh(\mathbf{z}_{1}) - 2n\displaystyle\sum_{i=2}^{K-1}h(\Zi) - nh(\mathbf{z}_{K}),  \nonumber 
\end{eqnarray}
\begin{equation}
\text{or } R_{1} + 2\displaystyle\sum_{i=2}^{K-1}R_{i} + R_{K} \leq  \displaystyle\sum_{i=1}^{K-1} h(\Yi^{*} | \siplus^{*}) + \displaystyle\sum_{i=2}^{K} h(\Yi^{*} | \siminus^{*}) - h(\mathbf{z}_{1}) - 2\displaystyle\sum_{i=2}^{K-1}h(\Zi) - h(\mathbf{z}_{K}). \label{eq:sideouter13}
\end{equation}
where \textit{(a)} follows from applying Lemma \ref{lemmaused2} as in the proof of Theorem \ref{theorem-outer1}.

The conditional differential entropy terms in (\ref{eq:sideouter13}) are simplified as follows. Consider a particular $i$ in the first summation term in (\ref{eq:sideouter13}):
\begin{equation}
h(\Yi^{*}|\siplus^{*}) = \log\labs \pi e\Sigma_{\Yi^{*}|\siplus^{*}}\rabs, \label{eq:sideouter4}
\end{equation}
where
\begin{equation}
\Sigma_{\Yi^{*}|\siplus^{*}} = \expec\lsqb\Yi^{*}\Yi^{*H}\rsqb - \expec\lsqb\Yi^{*}\siplus^{*H}\rsqb \expec\lsqb\siplus^{*}\siplus^{*H}\rsqb^{-1}\expec\lsqb\siplus^{*}\Yi^{*H}\rsqb. \label{eq:sideouter5}
\end{equation}
The individual terms in (\ref{eq:sideouter5}) are obtained as follows:
\begin{eqnarray}
& & \expec\lsqb\Yi^{*}\Yi^{*H}\rsqb = \Iden_{N_i} + \Hii\Pdashi\Hii^{H} + \sumneq \Hij\Pj\Hij^{H},  \label{eq:sideouter5a}\\
& & \expec\lsqb\Yi^{*}\siplus^{*H}\rsqb = \Hii \Pdashi\Hplus^{H},  \label{eq:sideouter5b}\\
& & \expec\lsqb\siplus^{*}\siplus^{*H}\rsqb = \Iden_{N_i} + \Hplus\Pdashi\Hplus^{H},  \label{eq:sideouter5c}
\end{eqnarray}
Hence, using Woodbury identity \cite{seber1} and simplifying, (\ref{eq:sideouter5}) becomes:
\begin{eqnarray}
& \Sigma_{\Yi^{*}|\siplus^{*}} & = \Iden_{N_i} + \sumneq\Hij\Pj\Hij^{H} + \Hii\Pdashi^{1/2}\lb \Iden_{M_i}  + \Pdashi^{1/2}\Hplus^{H}\Hplus\Pdashi^{1/2}\rb^{-1}\Pdashi^{1/2}\Hii^{H}. \nonumber
\end{eqnarray}
Finally we obtain:
\begin{equation}
h(\Yi^{*} | \siplus^{*}) = \log\labs \pi e \lsqb \Iden_{N_i} + \sumneq \Hij\Pj\Hij^{H} + \Hii\Pdashi^{1/2}\lb \Iden_{M_i}  + \Pdashi^{1/2}\Hplus^{H}\Hplus\Pdashi^{1/2} \rb^{-1}\Pdashi^{1/2}\Hii^{H}\rsqb\rabs, \label{eq:sideouter14}
\end{equation}
In a similar manner, it can also be shown that:
\begin{equation}
h(\Yi^{*} | \siminus^{*}) = \log\labs \pi e \lsqb \Iden_{N_i} + \sumneq \Hij\Pj\Hij^{H} + \Hii\Pdashi^{1/2}\lb\Iden_{M_i}  + \Pdashi^{1/2}\Hminus^{H}\Hminus\Pdashi^{1/2} \rb^{-1}\Pdashi^{1/2}\Hii^{H}\rsqb\rabs. \label{eq:sideouter15}
\end{equation}
Finally, the sum rate is upper bounded by using (\ref{eq:sideouter14}) and (\ref{eq:sideouter15}) in (\ref{eq:sideouter13}):
\begin{eqnarray}
& & R_{1} + 2\displaystyle\sum_{i=2}^{K-1}R_{i} + R_{K} \nonumber \\
& & \leq \displaystyle\sum_{i=1}^{K-1}\log\labs  \Iden_{N_i} + \sumneq \Hij\Pj\Hij^{H} + \Hii\Pdashi^{1/2}\lb \Iden_{M_i} + \Pdashi^{1/2}\Hplus^{H}\Hplus\Pdashi^{1/2} \rb^{-1}\Pdashi^{1/2}\Hii^{H}\rabs  \nonumber \\
& & \quad + \displaystyle\sum_{i=2}^{K}\log\labs \Iden_{N_i} + \sumneq \Hij\Pj\Hij^{H} + \Hii\Pdashi^{1/2}\lb \Iden_{M_i} + \Pdashi^{1/2}\Hminus^{H}\Hminus\Pdashi^{1/2} \rb^{-1}\Pdashi^{1/2}\Hii^{H}\rabs. \label{eq:sideouter16}
\end{eqnarray}
The proof is complete.
\subsection{Proof of Lemma \ref{lemma-outer2}}\label{sec:appendouter2-lemma}
Following two cases are considered to simplify the outer bound stated in Theorem \ref{theorem-outer2}.\\
\textbf{Case 1 $(M \leq N)$}:
For the symmetric case, applying Lemma \ref{lemmaused2} and simplifying (\ref{eq:outerth2}) for high SNR, the statement of Theorem \ref{theorem-outer2} becomes:
\begin{eqnarray}
& R_{1} + 2\displaystyle\sum_{i=2}^{K-1}R_{i} + R_{K} & \leq \displaystyle\sum_{i=1}^{K-1}\log\labs  \mathbf{I}_{N} + \rho^{\alpha}\sumneq \Hij\Hij^{H} + \rho^{1-\alpha}\Hii\lb \Hplus^{H}\Hplus \rb^{-1}\Hii^{H}\rabs \nonumber \\
& & \quad + \displaystyle\sum_{i=2}^{K}\log\labs \mathbf{I}_{N} + \rho^{\alpha}\sumneq \Hij\Hij^{H} + \rho^{1-\alpha}\Hii\lb \Hminus^{H}\Hminus \rb^{-1}\Hii^{H}\rabs + \oone. \label{eq:sideouter26}
\end{eqnarray}
The outer bound is simplified further under the following cases using Lemma \ref{lemmaused1}.\\
\emph{Weak Interference Case $\lb0 \leq \alpha \leq \frac{1}{2}\rb$}: In this case, (\ref{eq:sideouter26}) becomes:
\begin{eqnarray}
& R_{1} + 2\displaystyle\sum_{i=2}^{K-1}R_{i} + R_{K} & \leq \displaystyle\sum_{i=1}^{K-1}\lsqb M (1-\alpha)\log\rho + \min\lcb \min\lcb N, (K-1)M \rcb, N-M\rcb\alpha\log\rho\rsqb + \nonumber \\
& & \: \displaystyle\sum_{i=2}^{K}\lsqb M(1-\alpha)\log\rho + \min\lcb \min\lcb N, (K-1)M \rcb, N-M\rcb\alpha\log\rho\rsqb +\oone,\nonumber \\
& \text{or } 2(K-1)R_{i} & \leq 2(K-1)\lsqb M(1-\alpha)\log\rho + \min\lcb \min\lcb N, (K-1)M \rcb, N-M\rcb\alpha\log\rho\rsqb + \oone, \nonumber \\
& \text{or } R_{i} & \leq M(1-\alpha)\log\rho + \min\lcb \min\lcb N, (K-1)M \rcb, N-M\rcb\alpha\log\rho + \oone. \label{eq:sideouter27}
\end{eqnarray}
The per user GDOF is thus upper bounded as:
\begin{equation}
d(\alpha) \leq M(1-\alpha) + \min\lcb \min\lcb N, (K-1)M \rcb, N-M\rcb\alpha. \label{eq:sideouter28}
\end{equation}

\emph{Moderate Interference Case $\lb\frac{1}{2} \leq \alpha \leq 1\rb$}: In this regime, (\ref{eq:sideouter26}) reduces to the following form:
\begin{eqnarray}
& & R_{1} + 2\displaystyle\sum_{i=2}^{K-1}R_{i} + R_{K}  \leq \displaystyle\sum_{i=1}^{K-1}\lsqb \min\lcb(K-1)M,N\rcb\alpha  + \min\lcb M, N - \min\lcb N, (K-1)M\rcb \rcb(1-\alpha) \rsqb \log\rho  \nonumber \\
& & \quad + \displaystyle\sum_{i=2}^{K}\lsqb \min\lcb(K-1)M,N\rcb\alpha  + \min\lcb M, N - \min\lcb N, (K-1)M\rcb \rcb(1-\alpha) \rsqb \log\rho + \oone, \nonumber \\
& & \text{or } R_{i}  \leq \lsqb \alpha\min\lcb(K-1)M,N\rcb + \min\lcb M, N - \min\lcb N, (K-1)M\rcb \rcb(1-\alpha)\rsqb\log\rho + \oone. \nonumber \\ \label{eq:sideouter28b}
\end{eqnarray}
Hence, the per user GDOF is upper bounded as follows:
\begin{equation}
 d(\alpha)  \leq \min\lcb(K-1)M,N\rcb\alpha  + \min\lcb M, N - \min\lcb N, (K-1)M\rcb \rcb(1-\alpha). \label{eq:sideouter29}
\end{equation}
\emph{High Interference Case $\lb \alpha \geq 1 \rb$}:In this regime, (\ref{eq:sideouter26}) becomes
\begin{eqnarray}
& & R_{1} + 2\displaystyle\sum_{i=2}^{K-1}R_{i} + R_{K}  \leq \displaystyle\sum_{i=1}^{K-1} \min\lcb N, (K-1)M \rcb\alpha\log\rho + \displaystyle\sum_{i=2}^{K} \min\lcb N, (K-1)M \rcb\alpha\log\rho + \oone, \nonumber \\
& &\text{or } R_{i}  \leq \min\lcb N, (K-1)M \rcb\alpha + \oone. \nonumber
\end{eqnarray}
The per user GDOF  is upper bounded as:
\begin{eqnarray}
& d(\alpha) \leq \min\lcb N, (K-1)M \rcb\alpha. \label{eq:sideouter30}
\end{eqnarray}
As the per user GDOF in this case exceeds the interference free GDOF, this bound is not helpful for high interference regime.\\
\textbf{Case 2 $(M > N)$}:\\
When $M >N$, (\ref{eq:sideouter16}) is simplified as follows:
\begin{eqnarray}
& & R_{1} + \displaystyle\sum_{i=2}^{K-1} R_{i} + R_{K} \nonumber \\
& & \stackrel{(a)}{\leq}\displaystyle\sum_{i=1}^{K-1} \log\labs \Iden_N + \rho^{\alpha}\sumneq\Hij\Hij^{H} + \rho\Hii\lb \Iden_{M} + \rho^{\alpha} \mathbf{U}_{i+1,i}\Sigma_{i+1,i}\mathbf{U}_{i+1,i}^{H}\rb^{-1} \Hii^{H}\rabs \nonumber \\
& & \quad + \displaystyle\sum_{i=2}^{K} \log\labs \Iden_N + \rho^{\alpha}\sumneq\Hij\Hij^{H} + \rho\Hii\lb \Iden_{M}+\rho^{\alpha} \mathbf{U}_{i-1,i}\Sigma_{i-1,i}\mathbf{U}_{i-1,i}^{H}\rb^{-1} \Hii^{H}\rabs  \nonumber \\
& & = \displaystyle\sum_{i=1}^{K-1} \log\labs \Iden_N + \rho^{\alpha}\sumneq\Hij\Hij^{H} + \rho\Hii\mathbf{U}_{i+1,i}\lb \Iden_{M} + \rho^{\alpha} \Sigma_{i+1,i}\rb^{-1} \mathbf{U}_{i+1,i}^{H}\Hii^{H}\rabs  \nonumber \\
& & \quad +\displaystyle\sum_{i=2}^{K} \log\labs \Iden_N + \rho^{\alpha}\sumneq\Hij\Hij^{H} + \rho\Hii\mathbf{U}_{i-1,i}\lb \Iden_{M}+\rho^{\alpha} \Sigma_{i-1,i}\rb^{-1}\mathbf{U}_{i-1,i}^{H} \Hii^{H}\rabs  \nonumber \\
& & \stackrel{(b)}{=} \displaystyle\sum_{i=1}^{K-1} \log\labs \Iden_N + \rho^{\alpha}\sumneq\Hij\Hij^{H} + \rho\Hitilde\lb \Iden_{M} + \rho^{\alpha} \lsqb\begin{array}{cc}
  \Sigma_{N}^{i+1,i} &\mathbf{0} \\
  \mathbf{0} &\mathbf{0}_{M-N}
  \end{array}\rsqb\rb^{-1} \Hitilde^{H}\rabs   \nonumber \\
& & \quad + \displaystyle\sum_{i=2}^{K} \log\labs \Iden_N + \rho^{\alpha}\sumneq\Hij\Hij^{H} + \rho\Hitilde\lb \Iden_{M}+\rho^{\alpha} \lsqb\begin{array}{cc}
  \Sigma_{N}^{i-1,i} &\mathbf{0} \\
  \mathbf{0} &\mathbf{0}_{M-N}
  \end{array}\rsqb \rb^{-1}\Hitilde^{H}\rabs   \nonumber \\
& & \stackrel{(c)}{=}\displaystyle\sum_{i=1}^{K-1} \log\labs \Iden_N + \rho^{\alpha}\sumneq\Hij\Hij^{H} + \rho^{1-\alpha}\Hitilde^{(a)}\lb \Sigma_{N}^{i+1,i} \rb^{-1}\Hitilde^{(a)H} + \rho \Hitilde^{(b)}\Iden_{M-N}\Hitilde^{(b)H}\rabs  \nonumber \\
& & \quad + \displaystyle\sum_{i=2}^{K} \log\labs \Iden_N + \rho^{\alpha}\sumneq\Hij\Hij^{H} + \rho^{1-\alpha}\Hitilde^{(a)}\lb \Sigma_{N}^{i-1,i} \rb^{-1}\Hitilde^{(a)H} + \rho \Hitilde^{(b)}\Iden_{M-N}\Hitilde^{(b)H}\rabs + \oone, \nonumber \\ \label{eq:sideouter31}
\end{eqnarray}
where, \textit{(a)} is obtained by taking the EVD of $\Hij^{H}\Hij$; $\mathbf{U}_{ij} \in \Complex^{M\times M}$ and $\Sigma_{ij} \in \Complex^{M\times M}$; in \textit{(b)}, $\Hitilde \triangleq \mathbf{H}_{ii}\mathbf{U}_{ij}$ and $\Sigma_{N}^{j,i}$ contains $N$ non-zero singular values of $\Hij^{H}\Hij$; and in \textit{(c)}, $\Hitilde$ is partitioned into submatrices $\Hitilde^{(a)}$ and $\Hitilde^{(b)}$ of dimension $N \times N$ and $N \times (M-N)$, respectively.

Equation (\ref{eq:sideouter31}) is further simplified based on the value of $\alpha$, as follows.\\
\emph{Weak Interference Case $\lb0 \leq \alpha \leq \frac{1}{2}\rb$}:\\
Consider a specific $i$ in (\ref{eq:sideouter31}):
\begin{eqnarray}
& &\log\labs \Iden_N + \rho^{\alpha}\sumneq\Hij\Hij^{H} + \rho^{1-\alpha}\Hitilde^{(a)}\lb \Sigma_{N}^{i+1,i} \rb^{-1}\Hitilde^{(a)H} + \rho \Hitilde^{(b)}\Iden_{M-N}\Hitilde^{(b)H}\rabs \nonumber \\
& & \stackrel{(a)}{=} \min\lcb N, M-N\rcb \log\rho + \min\lcb N, N - \min\lcb N, M-N\rcb \rcb(1-\alpha)\log\rho + \nonumber \\
& & \qquad \qquad\min\lcb \min\lb (K-1)M,N\rb, \{N - \min\lcb N, M-N\rcb - N\}^{+}\rcb \alpha\log\rho  + \oone \nonumber \\
& & = \min\lcb N, M-N\rcb\log\rho + \lb N -  \min\lcb N, M-N\rcb \rb(1-\alpha)\log\rho + \oone, \label{eq:sideouter32}
\end{eqnarray}
where \textit{(a)} is obtained by using Lemma \ref{lemmaused4}.

From (\ref{eq:sideouter31}) and (\ref{eq:sideouter32}), the outer bound becomes:
\begin{eqnarray}
& & R_{i} \leq N(1-\alpha)\log\rho + \min\lcb N, M-N\rcb\alpha\log\rho +  \oone. \label{eq:sideouter33}
\end{eqnarray}
In weak interference case, the per user GDOF is upper bounded as:
\begin{equation}
d(\alpha) \leq N(1-\alpha) + \min\lcb N, M-N\rcb\alpha. \label{eq:sideouter34}
\end{equation}
\emph{Moderate Interference Case $\lb\frac{1}{2} \leq \alpha \leq 1\rb$}: \\
Consider a specific $ i $ in (\ref{eq:sideouter31}):
\begin{eqnarray}
& &\log\labs \Iden_N + \rho^{\alpha}\sumneq\Hij\Hij^{H} + \rho^{1-\alpha}\Hitilde^{(a)}\lb \Sigma_{N}^{i+1,i} \rb^{-1}\Hitilde^{(a)H} + \rho \Hitilde^{(b)}\Iden_{M-N}\Hitilde^{(b)H}\rabs \nonumber 
\end{eqnarray}
\begin{eqnarray}
& & \stackrel{(a)}{=}  \min\lcb N, M-N\rcb\log\rho + \min\lcb \min\lcb(K-1)M,N\rcb, N -  \min\lcb N, M-N\rcb\rcb \alpha\log\rho + \nonumber \\
& & \qquad \qquad \min\lcb N, \lb N -  \min\lcb N, M-N\rcb -  \min\lcb (K-1)M,N\rcb\rb^{+}\rcb(1-\alpha)\log\rho  + \oone \nonumber \\
& & = \min\lcb N, M-N\rcb\log\rho + \min\lcb N, N - \min\lcb N, M-N\rcb\rcb \alpha\log\rho + \oone, \label{eq:sideouter35}
\end{eqnarray}
where \textit{(a)} is obtained by using Lemma \ref{lemmaused4}.

From (\ref{eq:sideouter31}) and (\ref{eq:sideouter35}), the outer bound becomes:
\begin{eqnarray}
R_{i} \leq N\alpha\log\rho + \min\lcb N, M-N\rcb(1-\alpha)\log\rho + \oone. \label{eq:sideouter36}
\end{eqnarray}
In moderate interference case, per user GDOF is upper bounded as follows:
\begin{equation}
d(\alpha) \leq N\alpha + \min\lcb N, M-N\rcb(1-\alpha). \label{eq:sideouter37}
\end{equation}
\emph{High Interference Case $\lb \alpha \geq 1\rb$}:\\
In this case, the outer bound in (\ref{eq:sideouter31}) simplifies as follows:
\begin{eqnarray}
& \mysumK R_{i} & \leq \displaystyle\sum_{i=1}^{K-1} \log\labs \Iden_N + \rho^{\alpha}\sumneq\Hij\Hij^{H} +  \rho \Hitilde^{(b)}\Iden_{M-N}\Hitilde^{(b)H}\rabs + \nonumber \\
& & \quad \displaystyle\sum_{i=2}^{K} \log\labs \Iden_N + \rho^{\alpha}\sumneq\Hij\Hij^{H} +  \rho \Hitilde^{(b)}\Iden_{M-N}\Hitilde^{(b)H}\rabs + \oone \nonumber \\
& & \stackrel{(a)}{=} \displaystyle\sum_{i=1}^{K-1} \lsqb \min\lcb N, (K-1)M \rcb\alpha\log\rho + \min\lcb \min\lcb N, M-N\rcb, N -  \min\lcb N, (K-1)M \rcb \rcb\log\rho \rsqb + \nonumber \\
& & \displaystyle\sum_{i=2}^{K} \lsqb \min\lcb N, (K-1)M \rcb\alpha\log\rho + \min\lcb \min\lcb N, M-N\rcb -  \min\lcb N, (K-1)M \rcb \rcb\log\rho \rsqb + \oone \nonumber \\
& \text{or } R_{i} & \leq N\alpha\log\rho + \oone, \label{eq:sideouter38}
\end{eqnarray}
where \textit{(a)} is obtained by using Lemma \ref{lemmaused1}.

The per user GDOF in case of high interference is upper bounded as follows:
\begin{equation}
d(\alpha) \leq N\alpha. \label{eq:sideouter39}
\end{equation}
But the outer bound in this case exceeds the interference free GDOF i.e. $N$ as $\alpha \geq 1$ . Hence, this outer bound is not useful when $\alpha \geq 1$.

By combining (\ref{eq:sideouter28}), (\ref{eq:sideouter29}), (\ref{eq:sideouter34}) and (\ref{eq:sideouter37}) results in Lemma \ref{lemma-outer2}. This completes the proof.
\subsection{Proof of Theorem \ref{theorem-outer3}}\label{sec:appendouter3-th}
Define $\mathbf{S}_{j,\mathcal{B}}$ as in (\ref{eq:sideouter0a}). Let $\mathcal{A} = \{1,2,\ldots,K\}$ be the set of all transmitters. $\mathcal{A-B}$ is the complement of $\mathcal{B}$ in $\mathcal{A}$. The $ith$ transmitter and receiver are assumed to have $M_i$ and $N_i$ antennas, respectively. Following the procedure given in \cite{gou3} and using the Lemma \ref{lemmaused2}, the sum rate can be bounded as follows. 
 \begin{eqnarray}
&  R_{1} + 2\displaystyle\sum_{i=2}^{K-1}R_{i} + R_{K}  & \leq h\lb\Yone^{*}|\Skone^{*}\rb + h\lb\Yk^{*}|\Sonek^{*}\rb + \displaystyle\sum_{i=2}^{K-1} h\lb\Yi^{*}|\Skonetwoi^{*},\Soneacomp^{*} \rb  \nonumber \\
& & \: + \displaystyle\sum_{i=2}^{K-1} h \lb \Yi^{*}|\Sonektwoi^{*},\Skacomp^{*} \rb \nonumber \\
& & \: - h \lb \Zone\rb -2 n\displaystyle\sum_{i=2}^{K-1}h\lb \Zi\rb - h \lb \Zk\rb.  \label{eq:outerthree3}
\end{eqnarray}
The above expression is simplified for the SIMO case in \cite{gou3}. Here, since the transmitters could also have multiple antennas, the individual terms in (\ref{eq:outerthree3}) need to be evaluated as follows. The first term $h\lb\Yone^{*}|\Skone^{*}\rb$ in (\ref{eq:outerthree3}) is similar to the evaluation of conditional differential entropy in the proof of Theorem \ref{theorem-outer2}. On simplification, first term becomes:
\begin{equation}
h\lb \Yone^{*}|\Skone^{*}\rb = \log\labs \pi e \lsqb \Iden_{N_1} + \displaystyle\sum_{j=2}^{K}\Honej\Pj\Honej^{H} + \Honeone \Pone^{1/2}\lcb \Iden_{M_1} + \Pone^{1/2}\Hkone^{H}\Hkone\Pone^{1/2}\rcb^{-1}\Pone^{1/2}\Honeone^{H}\rsqb\rabs. \label{eq:outerthree4}
\end{equation}
Similarly, the term $h\lb\Yk^{*}|\Sonek^{*}\rb $ simplifies to
\begin{equation}
h\lb\Yk^{*}|\Sonek^{*}\rb  = \log \labs \pi e \lsqb \Iden_{N_K} + \displaystyle\sum_{j=1}^{K-1}\HKj\Pj\HKj^{H} + \Hkk\Pk^{1/2}\lcb \Iden_{M_K} + \Pk^{1/2}\Honek^{H}\Honek\Pk^{1/2}\rcb\Pk^{1/2}\Hkk^{H}\rsqb\rabs. \label{eq:outerthree5}
\end{equation}
Now consider the term $h \lb \Yi^{*}|\Skonetwoi^{*},\Soneacomp^{*}\rb$. In this case,
\begin{eqnarray}
 & &\Yi^{*} = \Hii\myxi^{*} + \sumneq \Hij\myxj^{*} + \Zi, \nonumber \\
 & & \Skonetwoi^{*} = \displaystyle\sum_{j \in \{1,2,\ldots,i\}}\HKj\myxj^{*} + \Zk, \nonumber \\
 \text{and } & &\Soneacomp^{*} = \displaystyle\sum_{j \in \usrind }\Honej\myxj^{*} + \Zone. \nonumber 
\end{eqnarray}
The conditional differential entropy becomes
\begin{equation}
h \lb \Yi^{*}|\Skonetwoi^{*},\Soneacomp^{*}\rb = \log\labs\pi e \Sigma_{\Yi^{*}|\Skonetwoi^{*},\Soneacomp^{*}} \rabs, \label{eq:outerthree7}
\end{equation}
where,
\begin{eqnarray}
& & \Sigma_{\Yi^{*}|\Skonetwoi^{*},\Soneacomp^{*}} = \expec \lsqb\Yi^{*}\Yi^{*H}\rsqb - \expec \lsqb\Yi^{*}\Sbar^{*H} \rsqb \expec \lsqb \Sbar^{*}\Sbar^{*H} \rsqb^{-1} \expec \lsqb \Sbar^{*}\Yi^{*H} \rsqb, \label{eq:outerthree8}\\
\text{and } & & \Sbar^{*} = \lsqb \begin{array}{ll}
                         \Skonetwoi^{*T} & \Soneacomp^{*T}
                        \end{array} \rsqb^{T}. \nonumber 
\end{eqnarray}
The output at receiver $i\:(i \neq 1,K)$ can also be expressed as follows.
\begin{eqnarray}
\Yi^{*} = \Hionebar\Xonebar + \Hitwobar\Xtwobar + \Zi, \label{eq:outerthree9}
\end{eqnarray}
\begin{eqnarray}
\text{where }& & \Xonebar = \lsqb \begin{array}{llll}
             \mathbf{x}_{1}^{*T} & \mathbf{x}_{2}^{*T} &\ldots &\mathbf{x}_{i}^{*T}
            \end{array} \rsqb^{T}, \quad \Xtwobar = \lsqb \begin{array}{llll}
		\mathbf{x}_{i+1}^{*T} &\mathbf{x}_{i+2}^{*T} \ldots &\mathbf{x}_{K}^{*T} 
\end{array} \rsqb^{T}. \nonumber
\end{eqnarray}
and $\Hionebar$ and $\Hitwobar$ are defined as in (\ref{eq:outerth3b}). The two side information terms can also be expressed as follows.
\begin{eqnarray}
& & \Skonetwoi = \Hkonebar\Xonebar + \Zk, \label{eq:outerthree10} \\
& & \Soneacomp = \Honekbar\Xtwobar + \mathbf{z}_{1}. \label{eq:outerthree11}
\end{eqnarray}
Now consider the evaluation of individual terms in (\ref{eq:outerthree8}):
\begin{eqnarray}
& \expec \lsqb \Yi^{*}\Yi^{*H}\rsqb & = \Iden_{N_{i}} + \Hionebar\Pionebar\Hionebar^{H} + \Hitwobar\Pitwobar\Hitwobar^{H}, \label{eq:outerthree12}\\
& \expec \lsqb \Yi^{*}\Sbar^{*H}\rsqb & = \lsqb \begin{array}{lll}
             \Hionebar\Pionebar\Hkonebar^{H} & \: &\Hitwobar\Pitwobar\Honekbar^{H}
            \end{array} \rsqb.  \label{eq:outerthree13} \\
& \expec \lsqb \Sonek^{*}\Sonek^{*H} \rsqb & = \lsqb \begin{array}{ll}
             \Iden_{N_{i}} + \Hkonebar\Pionebar\Hkonebar^{H}  &\mathbf{0} \\
	     \mathbf{0} &\Iden_{N_{i}} + \Honekbar\Pitwobar\Honekbar^{H}
            \end{array}\rsqb, \label{eq:outerthree15}
\end{eqnarray}

Hence, (\ref{eq:outerthree8}) simplifies to
\begin{eqnarray}
& & \Sigma_{\Yi^{*}|\Skonetwoi^{*},\Soneacomp^{*}} \nonumber \\
& & = \Iden_{N_i} + \Hionebar\Pionebar\Hionebar^{H} + \Hitwobar\Pitwobar\Hitwobar^{H} - \Hionebar\Pionebar\Hkonebar^{H}\lsqb \Iden_{N_{i}} + \Hkonebar\Pionebar\Hkonebar^{H} \rsqb^{-1}\Hkonebar\Pionebar\Hionebar^{H} -  \nonumber \\
& & \qquad \quad \Hitwobar\Pitwobar\Honekbar^{H} \lsqb \Iden_{N_{i}} + \Honekbar\Pitwobar\Honekbar^{H}\rsqb^{-1}\Honekbar\Pitwobar\Hitwobar^{H} \nonumber  
\end{eqnarray}
\begin{eqnarray}
& & = \Iden_{N_i} + \Hionebar\Pionebar^{1/2} \lcb \Iden_{N_i} - \Pionebar^{1/2}\Hkonebar^{H} \lsqb \Iden_{N_{K}} + \Hkonebar\Pionebar^{1/2}\Pionebar^{1/2}\Hkonebar^{H} \rsqb^{-1}\Hkonebar\Pionebar^{1/2} \rcb\Pionebar^{1/2}\Hionebar^{H} + \nonumber \\
& & \qquad \quad \Hitwobar\Pitwobar^{1/2}\lcb \Iden_{N_i} - \Pitwobar^{1/2}\Honekbar^{H}\lsqb \Iden_{N_1} + \Honekbar\Pitwobar^{1/2}\Pitwobar^{1/2}\Honekbar^{H}\rsqb^{-1}\Honekbar\Pitwobar^{1/2} \rcb\Pitwobar^{1/2}\Hitwobar^{H} \nonumber \\
& & = \Iden_{N_i} + \Hionebar\Pionebar^{1/2} \lcb  \Iden_{\Mri} + \Pionebar^{1/2}\Hkonebar^{H}\Hkonebar\Pionebar^{1/2}\rcb^{-1}\Pionebar^{1/2}\Hionebar^{H} + \nonumber \\
& & \qquad \Hitwobar\Pitwobar^{1/2} \lcb \Iden_{\Msi} + \Pitwobar^{1/2}\Honekbar^{H}\Honekbar \Pitwobar^{1/2}\rcb^{-1}\Pitwobar^{1/2}\Hitwobar^{H} \label{eq:outerthree16}
\end{eqnarray}
where $\Mri = \dimone$, $\Msi = \dimtwo$ and the last equation follows from the Woodbury identity \cite{seber1}. The quantity $\Pionebar$ and $\Pitwobar$ are as defined in (\ref{eq:outerth3b}).

Hence, (\ref{eq:outerthree7}) becomes
\begin{eqnarray}
&  h \lb \Yi^{*}|\Skonetwoi^{*},\Soneacomp^{*}\rb &  = \log\labs\pi e \lsqb \Iden_{N_i} + \Hionebar\Pionebar^{1/2} \lcb  \Iden_{\Mri} + \Pionebar^{1/2}\Hkonebar^{H}\Hkonebar\Pionebar^{1/2}\rcb^{-1}\Pionebar^{1/2}\Hionebar^{H} + \right. \right. \nonumber \\
& & \qquad \left. \left. \Hitwobar\Pitwobar^{1/2} \lcb \Iden_{\Msi} + \Pitwobar^{1/2}\Honekbar^{H}\Honekbar \Pitwobar^{1/2}\rcb^{-1}\Pitwobar^{1/2}\Hitwobar^{H} \rsqb \rabs. \label{eq:outerthree17}
\end{eqnarray}
In a similar manner, it can be shown that
\begin{eqnarray}
& & h \lb \Yi^{*}|\Sonektwoi^{*},\Skacomp^{*} \rb = \log  \labs\pi e \lsqb \Iden_{N_{i}} + \Hithreebar\Pithreebar^{1/2} \lcb \Iden_{\Mri^{'}} + \Pithreebar^{1/2}\Honetwobar^{H}\Honetwobar\Pithreebar^{1/2}\rcb^{-1}\Pithreebar^{1/2}\Hithreebar^{H} + \right.\right.\nonumber \\
& & \left.\left. \Hifourbar\Pifourbar^{1/2}\lcb \Iden_{\Msi^{'}} + \Pifourbar^{1/2}\HKthreebar^{H}\HKthreebar\Pifourbar^{1/2} \rcb^{-1}\Pifourbar^{1/2}\Hifourbar^{H} \rabs \rsqb, \label{eq:outerthree18}
\end{eqnarray}
where $\Mri^{'} = \displaystyle\sum_{j=2}^{i}M_{j} + M_{K}$ and $\Msi^{'} = M_{1} + \displaystyle\sum_{j=i+1}^{K-1}M_{j}$, and $\Pithreebar$
and $\Pifourbar$ are as defined in (\ref{eq:outerth3b}).

By combining (\ref{eq:outerthree4}), (\ref{eq:outerthree5}), (\ref{eq:outerthree17}) and (\ref{eq:outerthree18}) results in Theorem \ref{theorem-outer3}.
\subsection{Proof of Lemma \ref{lemma-outer3}}\label{sec:appendouter3-lemma}
For symmetric case, using Lemma \ref{lemmaused3}, the sum rate outer bound in Theorem \ref{theorem-outer3} reduces to the following form:
\begin{eqnarray}
& & R_{1} + \displaystyle\sum_{i=2}^{K-1}R_{i} + R_{K}   \leq \log\labs \Iden_{N} + \rho^{\alpha}\displaystyle\sum_{j=2}^{K}\Honej\Honej^{H} + \rho\Honeone \lcb \Iden_{M} + \rho^{\alpha}\Hkone^{H}\Hkone \rcb^{-1}\Honeone^{H}\rabs  \nonumber \\
& & \quad + \displaystyle\sum_{i=2}^{K-1}\log \labs \Iden_{N} + \Hionebar\lcb \Iden_{\Mri} + \Hkonebar^{H}\Hkonebar\rcb^{-1}\Hionebar^{H} + \Hitwobar\lcb \Iden_{\Mri} + \Honekbar^{H}\Honekbar \rcb^{-1}\Hitwobar^{H}\rabs  \nonumber \\
& & \quad + \displaystyle\sum_{i=2}^{K-1}\log \labs \Iden_{N} + \Hithreebar\lcb \Iden_{\Mri} + \Honetwobar^{H}\Honetwobar\rcb^{-1}\Hithreebar^{H} +\Hifourbar\lcb \Iden_{\Mri} + \HKthreebar^{H}\HKthreebar \rcb^{-1}\Hifourbar^{H} \rabs   \nonumber \\
& & \quad + \log\labs \Iden_{N} +\rho^{\alpha} \displaystyle\sum_{j=1}^{K-1}\HKj\HKj^{H} + \rho\Hkk\lcb \Iden_{M} + \rho^{\alpha}\Honek^{H}\Honek \rcb^{-1}\Hkk^{H}\rabs, \label{eq:lmouterthree1}
\end{eqnarray}
where $M_{r_{i}} = iM$ and with a slight abuse of notation, and 
\begin{eqnarray}
& & \Hionebar = \lsqb \sqrt{\rho^{\alpha}}\mathbf{H}_{i1} \: \sqrt{\rho^{\alpha}}\mathbf{H}_{i2} \ldots \: \sqrt{\rho}\mathbf{H}_{ii} \rsqb, \quad \Hitwobar = \lsqb \sqrt{\rho^{\alpha}}	\mathbf{H}_{i,i+1} \: \sqrt{\rho^{\alpha}}\mathbf{H}_{i,i+2} \: \ldots \: \sqrt{\rho^{\alpha}}\mathbf{H}_{iK} \rsqb, \nonumber \\
& & \Hkonebar = \lsqb \sqrt{\rho^{\alpha}}\mathbf{H}_{K1} \: \sqrt{\rho^{\alpha}}\mathbf{H}_{K2} \: \ldots \: \sqrt{\rho^{\alpha}}\mathbf{H}_{Ki} \rsqb,\: \Honekbar = \lsqb 
                   \sqrt{\rho^{\alpha}}\mathbf{H}_{1,i+1} \: \sqrt{\rho^{\alpha}}\mathbf{H}_{1,i+2} \:\ldots\: \sqrt{\rho^{\alpha}}\mathbf{H}_{1K} \rsqb, \nonumber \\
& & \Honetwobar = \lsqb \sqrt{\rho^{\alpha}}\mathbf{H}_{1K} \: \sqrt{\rho^{\alpha}}\mathbf{H}_{12} \:\ldots\: \sqrt{\rho^{\alpha}}\mathbf{H}_{1i}
                  \rsqb ,\: \HKthreebar = \lsqb \sqrt{\rho^{\alpha}}\mathbf{H}_{K1} \:\sqrt{\rho^{\alpha}}\mathbf{H}_{K,i+1} \:\ldots \: \sqrt{\rho^{\alpha}}\mathbf{H}_{K,K-1} \rsqb, \nonumber \\
& & \Hithreebar = \lsqb \sqrt{\rho^{\alpha}}\mathbf{H}_{iK} \: \sqrt{\rho^{\alpha}}\mathbf{H}_{i2} \: \ldots \:\sqrt{\rho}\mathbf{H}_{ii} \rsqb, \: \Hifourbar = \lsqb	\sqrt{\rho^{\alpha}}\mathbf{H}_{i1} \: \sqrt{\rho^{\alpha}}\mathbf{H}_{i,i+1} \: \ldots \: \sqrt{\rho^{\alpha}}\mathbf{H}_{i,K-1} \rsqb. \label{eq:lmouterthree2}
\end{eqnarray}

Now consider the following term in (\ref{eq:lmouterthree1})
\begin{eqnarray}
& &  \log \labs \Iden_{N} + \Hionebar\lcb \Iden_{\Mri} + \Hkonebar^{H}\Hkonebar\rcb^{-1}\Hionebar^{H} + 
\Hitwobar\lcb \Iden_{\Mri}+ \Honekbar^{H}\Honekbar \rcb^{-1}\Hitwobar^{H}\rabs. \label{eq:lmouterthree3}
\end{eqnarray}
In the above equation, consider the following term
\begin{eqnarray}
& & \Hionebar\lcb \Iden_{\Mri} + \Hkonebar^{H}\Hkonebar\rcb^{-1}\Hionebar^{H} \nonumber \\
& & = \lsqb \sqrt{\rho^{\alpha}}\mathbf{H}_{i1} \: \sqrt{\rho^{\alpha}}\mathbf{H}_{i2} \ldots \: \sqrt{\rho}\mathbf{H}_{ii} \rsqb\lcb \Iden_{\Mri} + \lsqb \sqrt{\rho^{\alpha}}\mathbf{H}_{K1} \: \sqrt{\rho^{\alpha}}\mathbf{H}_{K2} \: \ldots \: \sqrt{\rho^{\alpha}}\mathbf{H}_{Ki} \rsqb^{H}  \right. \nonumber \\
& & \left. \lsqb \sqrt{\rho^{\alpha}}\mathbf{H}_{K1} \: \sqrt{\rho^{\alpha}}\mathbf{H}_{K2} \: \ldots \: \sqrt{\rho^{\alpha}}\mathbf{H}_{Ki} \rsqb \rcb^{-1}\lsqb \sqrt{\rho^{\alpha}}\mathbf{H}_{i1} \: \sqrt{\rho^{\alpha}}\mathbf{H}_{i2} \ldots \: \sqrt{\rho}\mathbf{H}_{ii} \rsqb^{H}, \nonumber \\
& & = \rho^{\alpha}\lsqb \mathbf{H}_{i1} \: \mathbf{H}_{i2} \ldots \: \sqrt{\rho^{1-\alpha}}\mathbf{H}_{ii} \rsqb\lcb \Iden_{\Mri} + \rho^{\alpha}\lsqb \mathbf{H}_{K1} \: \mathbf{H}_{K2} \: \ldots \: \mathbf{H}_{Ki} \rsqb^{H}  \lsqb \mathbf{H}_{K1} \: \mathbf{H}_{K2} \: \ldots \: \mathbf{H}_{Ki} \rsqb \rcb^{-1} \nonumber \\
& & \qquad \qquad\lsqb \mathbf{H}_{i1} \: \mathbf{H}_{i2} \ldots \: \sqrt{\rho^{1-\alpha}}\mathbf{H}_{ii} \rsqb ^{H}. \label{eq:lmouterthree4}
\end{eqnarray}
In a similar way, it can be shown that
\begin{eqnarray}
& & \Hitwobar\lcb \Iden_{\Msi}+ \Honekbar^{H}\Honekbar \rcb^{-1}\Hitwobar^{H} \nonumber \\
& & = \rho^{\alpha}\lsqb \mathbf{H}_{i,i+1} \: \mathbf{H}_{i,i+2} \: \ldots \: \mathbf{H}_{iK} \rsqb \lcb \Iden_{\Msi} + \rho^{\alpha}\lsqb \mathbf{H}_{1,i+1} \: \mathbf{H}_{1,i+2} \:\ldots\: \mathbf{H}_{1K}  \rsqb^{H}\lsqb \mathbf{H}_{1,i+1} \: \mathbf{H}_{1,i+2} \:\ldots\: \mathbf{H}_{1K}\rsqb\rcb^{-1} \nonumber \\
& & \qquad \qquad \lsqb \mathbf{H}_{i,i+1} \: \mathbf{H}_{i,i+2} \: \ldots \: \mathbf{H}_{iK} \rsqb^{H}. \label{eq:lmouterthree5}
\end{eqnarray}
From (\ref{eq:lmouterthree4}) and (\ref{eq:lmouterthree5}), for large $\rho$, (\ref{eq:lmouterthree3}) becomes:
\begin{eqnarray}
& &  \log\labs \Iden_{N} + \lsqb \mathbf{H}_{i1} \: \mathbf{H}_{i2} \ldots \: \sqrt{\rho^{1-\alpha}}\mathbf{H}_{ii} \rsqb\lcb \lsqb \mathbf{H}_{K1} \: \mathbf{H}_{K2} \: \ldots \: \mathbf{H}_{Ki} \rsqb^{H}  \lsqb \mathbf{H}_{K1} \: \mathbf{H}_{K2} \: \ldots \: \mathbf{H}_{Ki} \rsqb \rcb^{-1} \nonumber \right. \nonumber \\ 
& & \left. \qquad \lsqb \mathbf{H}_{i1} \: \mathbf{H}_{i2} \ldots \: \sqrt{\rho^{1-\alpha}}\mathbf{H}_{ii} \rsqb ^{H} + \lsqb \mathbf{H}_{i,i+1} \: \mathbf{H}_{i,i+2} \: \ldots \: \mathbf{H}_{iK} \rsqb \lcb  \lsqb \mathbf{H}_{1,i+1} \: \mathbf{H}_{1,i+2} \:\ldots\: \mathbf{H}_{1K}  \rsqb^{H} \right.\right. \nonumber \\
& & \left. \left. \lsqb \mathbf{H}_{1,i+1} \: \mathbf{H}_{1,i+2} \:\ldots\: \mathbf{H}_{1K}\rsqb\rcb^{-1} \lsqb \mathbf{H}_{i,i+1} \: \mathbf{H}_{i,i+2} \: \ldots \: \mathbf{H}_{iK} \rsqb^{H}\rabs + \oone \nonumber \\
& & \stackrel{(a)}{=} \log\labs \Iden_{N} + \lsqb \mathbf{H}_{i1} \: \mathbf{H}_{i2} \ldots \: \sqrt{\rho^{1-\alpha}}\mathbf{H}_{ii} \rsqb \lsqb \mathbf{H}_{i1} \: \mathbf{H}_{i2} \ldots \: \sqrt{\rho^{1-\alpha}}\mathbf{H}_{ii} \rsqb ^{H}\rabs + \oone \nonumber \\
& & \stackrel{(b)}{=} \log\labs \Iden_{N} + \rho^{1-\alpha}\mathbf{H}_{ii}\mathbf{H}_{ii}^{H} \rabs + \oone, \label{eq:lmouterthree6}
\end{eqnarray}
where \textit{(a)} is obtained by using the fact that the terms containing inverses are independent of $\alpha$ and are invertible when $\frac{N}{M} < K \leq \frac{N}{M}+1$, and \textit{(b)} is obtained by taking the constant terms into the $\oone$ approximation.

Similarly, it can be shown that
\begin{eqnarray}
& & \log \labs \Iden_{N} + \Hithreebar\lcb \Iden_{\Mri} + \Honetwobar^{H}\Honetwobar\rcb^{-1}\Hithreebar^{H} +\Hifourbar\lcb \Iden_{\Mri} + \HKthreebar^{H}\HKthreebar \rcb^{-1}\Hifourbar^{H} \rabs \nonumber \\
& & = \log\labs \Iden_{N} + \rho^{1-\alpha}\mathbf{H}_{ii}\mathbf{H}_{ii}^{H} \rabs + \oone. \label{eq:lmouterthree7}
\end{eqnarray}
Using (\ref{eq:lmouterthree6}) and (\ref{eq:lmouterthree7}), for large $\rho$, the sum rate bound in (\ref{eq:lmouterthree1}) reduces to
\begin{eqnarray}
& R_{1} + 2\displaystyle\sum_{i=2}^{K-1} R_{i} + R_{K} &\leq \log\labs \Iden_{N} + \rho^{\alpha}\displaystyle\sum_{j=2}^{K}\Honej\Honej^{H} + \rho^{1-\alpha}\Honeone \lcb \Hkone^{H}\Hkone \rcb^{-1}\Honeone^{H}\rabs + \nonumber \\
& & \displaystyle\sum_{i=2}^{K-1}\log\labs \Iden_{N} + \rho^{1-\alpha}\mathbf{H}_{ii}\mathbf{H}_{ii}^{H} \rabs +  \displaystyle\sum_{i=2}^{K-1} \log\labs \Iden_{N} + \rho^{1-\alpha}\mathbf{H}_{ii}\mathbf{H}_{ii}^{H} \rabs +  \nonumber \\
& & \log\labs \Iden_{N} + \rho^{\alpha}\displaystyle\sum_{j=1}^{K-1}\HKj\HKj^{H} + \rho^{1-\alpha}\Hkk\lcb \Honek^{H}\Honek \rcb^{-1}\Hkk^{H}\rabs + \oone. \label{eq:lmouterthree8}
\end{eqnarray}
The outer bound in (\ref{eq:lmouterthree8}) is further simplified based on the range of $\alpha$ as follows.\\
\emph{Weak Interference Case $\lb0 \leq \alpha \leq \frac{1}{2}\rb$}: \\
In this case, using Lemma \ref{lemmaused1}, the sum rate bound in (\ref{eq:lmouterthree8}) simplifies to
\begin{eqnarray}
& & R_{1} + 2\displaystyle\sum_{i=2}^{K-1} R_{i} + R_{K} \nonumber \\
& & \leq M(1-\alpha)\log\rho + \min\lcb \min \lcb N, (K-1)M \rcb, N-M \rcb\alpha\log\rho + \displaystyle\sum_{i=2}^{K-1}M(1-\alpha)\log\rho +  \nonumber \\
& & \quad \displaystyle\sum_{i=2}^{K-1}M(1-\alpha)\log\rho + M(1-\alpha)\log\rho + \min\lcb \min \lcb N, (K-1)M \rcb, N-M \rcb\alpha\log\rho + \oone \nonumber \\
& & = \lsqb 2\min\lcb \min \lb N, (K-1)M \rb, N-M \rcb\alpha + 2(K-1)M(1-\alpha)\rsqb\log\rho + \oone. \label{eq:lmouterthree9}
\end{eqnarray}
Hence, in the symmetric case,
\begin{eqnarray}
& 2(K-1)R_{i} &\leq 2(K-1)M(1-\alpha)\log\rho + 2\lb  N-M \rb \alpha \log \rho + \oone. \label{eq:lmouterthree10}
\end{eqnarray}
Thus, in the weak interference case, the per user GDOF is upper bounded as
\begin{equation}
d(\alpha) \leq M(1-\alpha) + \frac{1}{K-1}\lb N-M\rb   \alpha. \label{eq:lmouterthree11}
\end{equation}
\emph{Moderate Interference Case $\lb\frac{1}{2} \leq \alpha \leq 1\rb$}:\\
In this case, using Lemma \ref{lemmaused1}, (\ref{eq:lmouterthree8}) simplifies to
\begin{eqnarray}
& & R_{1} + 2\displaystyle\sum_{i=2}^{K-1}R_{i} + R_{K} \nonumber \\
& & \leq \lsqb \alpha\min\lcb N, (K-1)M\rcb + \min\lcb M, N - \min\lcb N, (K-1)M\rcb\rcb(1-\alpha) + \displaystyle\sum_{i=2}^{K-1}M(1-\alpha)\right. \nonumber \\
& & \: \left. + \displaystyle\sum_{i=2}^{K-1}M(1-\alpha) + \alpha\min\lcb N, (K-1)M\rcb + \min\lcb M, N - \min\lcb N, (K-1)M\rcb\rcb(1-\alpha)\rsqb\log\rho + \oone \nonumber \\
& & = 2\alpha\min\lcb N, (K-1)M\rcb\log\rho + 2\min\lcb M, N - \min\lcb N, (K-1)M\rcb\rcb(1-\alpha)\log\rho\nonumber \\
& & \: \: + 2(K-2)M(1-\alpha)\log\rho + \oone. \nonumber
\end{eqnarray}
Hence, in the symmetric case,
\begin{eqnarray}
& 2(K-1)R_{i} & \leq 2\alpha(K-1)M\log\rho + 2\min\lcb M, N - (K-1)M\rcb(1-\alpha) \log\rho\nonumber \\
& & \: \: + 2(K-2)M(1-\alpha)\log\rho + \oone. \label{eq:lmouterthree12}
\end{eqnarray}
As $\frac{N}{M} < K \leq \frac{N}{M}+1$, $\min\lcb M, N - (K-1)M\rcb=N - (K-1)M$, and hence, the per user GDOF in the moderate interference regime is upper bounded as given below:
\begin{eqnarray}
& d(\alpha) & \leq M\alpha + \frac{1}{K-1}(N-M)(1-\alpha). \label{eq:lmouterthree13}
\end{eqnarray}
\emph{High Interference Case $\lb \alpha \geq 1\rb$}:\\
In this case, it can be shown that the sum rate bound in (\ref{eq:lmouterthree8}) leads to $d(\alpha) \leq \alpha M$, which exceeds the interference free GDOF. Hence, the upper bound reduces to $d(\alpha) \leq M$.

Finally, combining (\ref{eq:lmouterthree11}) and (\ref{eq:lmouterthree13}) results in Lemma \ref{lemma-outer3}.
\subsection{Proof of Theorem \ref{th:compareouter}}\label{sec:th-combineouter}
In the initial part of the proof, the outer bound in Lemma \ref{lemma-outer1} is simplified. Then, in specific cases, the performance of the outer bound is characterized as a function of $K$, $M$, $N$ and $\alpha$.\\
\textit{Weak $(0 \leq \alpha \leq \frac{1}{2})$ and moderate $(\frac{1}{2} \leq \alpha \leq 1)$ interference regime:}\\
When $M \leq N$, for a specific $L_1$ and $L_2$ $( 0 < L_1 + L_2 \leq K)$, the outer bound in Lemma \ref{lemma-outer1} is of the following form:
\begin{equation}
d(\alpha) \leq \frac{1}{L}\lsqb L_1M + \min\lcb r, L_1(N-M)\rcb\alpha + L_r + \min\lcb r, L_2N-L_r \rcb(1-\alpha)\rsqb, \label{eq:minouter8}
\end{equation}
where $r \triangleq \min\lcb L_2M, L_1N\rcb$ and $L_r \triangleq L_2M - r$. The RHS in \eqref{eq:minouter8} is simplified under the following cases.\\
\textbf{Case 1:} When $\min\lcb L_2M, L_1N\rcb = L_2M$, we have
\begin{eqnarray}
& & \frac{L_2}{L_1} \leq \frac{N}{M}. \label{eq:minouter9}
\end{eqnarray}
Under this condition, (\ref{eq:minouter8}) becomes
\begin{eqnarray}
& d(\alpha) & \leq \frac{1}{L}\lsqb L_1M + \min\lcb L_2M, L_1(N-M)\rcb\alpha + \min\lcb L_2M, L_2N\rcb(1-\alpha)\rsqb \nonumber \\
& & = \frac{1}{L}\lsqb LM + \min\lcb L_2M, L_1(N-M) \rcb\alpha - L_2M\alpha\rsqb. \label{eq:minouter10}
\end{eqnarray}
Equation (\ref{eq:minouter10}) is simplified under the following cases:\\
\textbf{Case 1(a):} When $\min\lcb L_2M, L_1(N-M)\rcb = L_1(N-M)$, we have 
\begin{eqnarray}
 \frac{L_2}{L_1} \geq \frac{N}{M}-1. \label{eq:minouter11}
\end{eqnarray}
Combining this with (\ref{eq:minouter9}), we have 
\begin{equation}
\frac{N}{M} - 1 \leq \frac{L_2}{L_1} \leq \frac{N}{M}. \label{eq:minouter12}
\end{equation}
Under this condition, (\ref{eq:minouter10}) becomes
\begin{equation}
d(\alpha) \leq M(1-\alpha) + \frac{L_1}{L}N\alpha. \label{eq:minouter13}
\end{equation}
\textbf{Case 1(b):} When $\min\lcb L_2M, L_1(N-M)\rcb = L_2M$, then (\ref{eq:minouter10}) becomes
\begin{equation}
d(\alpha) \leq M. \label{eq:minouter14}
\end{equation}
This case is not useful, as the RHS is equal to the interference free GDOF.\\
\textbf{Case 2:} When $\min\lcb L_2M, L_1N\rcb = L_1N$, we have
\begin{equation}
\frac{L_2}{L_1} \geq \frac{N}{M}. \label{eq:minouter15a}
\end{equation}
In this case, (\ref{eq:minouter8}) becomes
\begin{eqnarray}
& d(\alpha) & \leq \frac{1}{L}\lsqb L_1M + \min\lcb L_1N, L_1(N-M)\rcb\alpha + L_2M-L_1N  \right. \nonumber \\
& & \qquad \qquad \left. + \min\lcb L_1N, L_2N-L_2M+L_1N\rcb(1-\alpha)\rsqb \nonumber \\
& & = M - \frac{L_1}{L}M\alpha. \label{eq:minouter15}
\end{eqnarray}
\textit{High interference regime $(\alpha \geq 1)$:}\\
In the high interference regime, for a specific $L_1$ and $L_2$ $( 0 < L_1 + L_2 \leq K)$, Lemma \ref{lemma-outer1} is of the following form:
\begin{eqnarray}
& d(\alpha) & \leq \frac{1}{L}\lsqb r\alpha + \min\lcb L_1M, L_1N-r\rcb + L_r\rsqb. \label{eq:minouthigh1}
\end{eqnarray}
The above equation is simplified under the following cases.\\
\textbf{Case 1:} When $\min\lcb L_2M, L_1N\rcb = L_2M$, we have
\begin{equation}
\frac{L_2}{L_1} \leq \frac{N}{M}, \label{eq:minouthigh2}
\end{equation}
and (\ref{eq:minouthigh1}) becomes
\begin{equation}
d(\alpha) \leq \frac{1}{L}\lsqb L_2M\alpha + \min\lcb L_1M, L_1N - L_2M\rcb\rsqb. \label{eq:minouthigh3a}
\end{equation}
The above equation is further simplified under following cases.\\
\textbf{Case 1(a):} When $\min\lcb L_1M, L_1N - L_2M\rcb = L_1N - L_2M$, then the following condition is obtained:
\begin{equation}
\frac{L_2}{L_1} \geq \frac{N}{M}-1.  \label{eq:minouthigh3}
\end{equation}
In this case, (\ref{eq:minouthigh3a}) becomes
\begin{eqnarray}
& d(\alpha) & \leq \frac{1}{L}\lsqb L_2M\alpha + L_1N - L_2M\rsqb \nonumber \\
& & = N + \frac{L_2}{L}\lsqb M(\alpha-1) - N\rsqb. \label{eq:minouthigh4}
\end{eqnarray}
\textbf{Case 1(b):} When $\min\lcb L_1M, L_1N - L_2M\rcb = L_1M$,  (\ref{eq:minouthigh3a}) becomes
\begin{equation}
d(\alpha) \leq \frac{L_2\alpha+L_1}{L}M. \label{eq:minouthigh5}
\end{equation}
As $\alpha\geq 1$, this case is not useful as the RHS in the above equation exceeds the interference free GDOF.\\
\textbf{Case 2:} When $\min\lcb L_2M, L_1N\rcb = L_1N$,  
\begin{equation}
\frac{L_2}{L_1} \geq \frac{N}{M}, \label{eq:minouthigh6}
\end{equation}
and (\ref{eq:minouthigh1}) becomes
\begin{eqnarray}
& d(\alpha) &\leq \frac{1}{L}\lsqb L_1N\alpha + L_2M-L_1N \rsqb \nonumber \\
& & = N(\alpha-1) + \frac{L_2}{L}\lsqb M - N(\alpha-1)\rsqb. \label{eq:minouthigh7}
\end{eqnarray}

Due to the minimization involved in Lemma \ref{lemma-outer1}, it is not possible to characterize the performance of the outer bounds in all the cases. However, a tractable solution exists in the following cases.\\
\textbf{Case a $(K \geq N+M)$:} It is required to determine the value of $L_1$ and $L_2$, such that the outer bound in \ref{lemma-outer1} is minimized. First, the weak and moderate interference regimes are considered, followed by the high interference regime in the later part of the proof.

The RHS in (\ref{eq:minouter13}) is minimized when $\frac{L_1}{L}$ is minimized, under the constraint in (\ref{eq:minouter12}). In other words, $\frac{L}{L_1}$ or $\frac{L_2}{L_1}$ is required to be maximized to minimize the RHS in (\ref{eq:minouter13}). From (\ref{eq:minouter12}), it can be noticed that $\frac{L_2}{L_1}$ is maximized when $\frac{L_2}{L_1} = \frac{N}{M}$. As $K \geq M+N$, it is always possible to choose $L_1=M$ and $L_2=N$, and (\ref{eq:minouter13}) becomes
\begin{equation}
d(\alpha) \leq M - \frac{M^2}{M+N}\alpha. \label{eq:minouter16}
\end{equation}
The RHS in (\ref{eq:minouter15}) minimized by choosing $\frac{L_1}{L}$ as large as possible, under the constraint in (\ref{eq:minouter15a}). Maximizing $\frac{L_1}{L}$ is the same as minimizing $\frac{L_2}{L_1}$. By choosing $L_1=M$ and $L_2=N$, the RHS in (\ref{eq:minouter15}) is minimized, and the outer bound reduces to following from:
\begin{equation}
d(\alpha) \leq M - \frac{M^2}{M+N}\alpha,  \label{eq:minouter17a}
\end{equation}
which is same as that in (\ref{eq:minouter16}). Hence, the outer bound in \ref{lemma-outer1} is minimized by choosing $L_1=M$ and $L_2=N$ and is given by (\ref{eq:minouter17a}).

Now, the outer bounds are compared in the following interference regimes. As $K \geq M+N$, the condition $\frac{N}{M}< K \leq \frac{N}{M}+1$ is not satisfied, and hence, the outer bound in Lemma 3 is not applicable. \\
\textit{Weak interference regime $(0 \leq \alpha \leq \frac{1}{2})$:} In the weak interference regime, the outer bound on the per user GDOF in Lemma 2 reduces to:
\begin{equation}
d(\alpha) \leq M(1-\alpha) + (N-M)\alpha. \label{eq:minouter18}
\end{equation}
The outer bound in (\ref{eq:minouter18}) exceeds that in (\ref{eq:minouter17a}), when
\begin{eqnarray}
& & M - \frac{M^2}{M+N}\alpha < M(1-\alpha) + (N-M)\alpha, \nonumber \\
\text{or } & & MN < N^{2} - M^{2}, \label{eq:minouter19}
\end{eqnarray}
which results in (\ref{eq:minouter1}).\\
\textit{Moderate interference regime $(\frac{1}{2} \leq \alpha \leq 1)$:} In the moderate interference regime, the outer bound in Lemma 2 reduces to
\begin{equation}
d(\alpha) \leq N\alpha. \label{eq:minouter20} 
\end{equation}
The outer bound in (\ref{eq:minouter17a}) is active as compared to (\ref{eq:minouter20}), when
\begin{eqnarray}
& & M - \frac{M^2}{M+N}\alpha < N\alpha, \nonumber \\
\text{or } & &  \alpha > \frac{M(M+N)}{N(M+N)+M^2}. \label{eq:minouter21}
\end{eqnarray}
Note that $\frac{M(M+N)}{N(M+N)+M^2} \leq 1$. The outer bound in (\ref{eq:minouter17a}) is active for the entire moderate interference regime, if 
\begin{eqnarray}
& & \frac{M(M+N)}{N(M+N)+M^2} < \frac{1}{2}, \nonumber \\
 \text{or } & & MN < N^2 - M^2. \label{eq:minouter22}
 \end{eqnarray}
Otherwise, when $\frac{1}{2} \leq \alpha \leq \frac{M(M+N)}{N(M+N)+M^2}$, the outer bound in (\ref{eq:minouter20}) is active, and when $\frac{M(M+N)}{N(M+N)+M^2} < \alpha \leq 1$, the outer bound in (\ref{eq:minouter17a}) is active. Combining these results in (\ref{eq:minouter2a}) and (\ref{eq:minouter2}).\\
\textit{High interference regime $(\alpha \geq 1)$:} The RHS in (\ref{eq:minouthigh4}) and (\ref{eq:minouthigh7}) need to be  minimized in cases $1$ and $2$ discussed in a previous page, respectively. Consider the minimization of (\ref{eq:minouthigh4}) first. When $M(\alpha-1) - N \geq 0$, the RHS in (\ref{eq:minouthigh4}) exceeds the interference free GDOF and this case is not useful. When $M(\alpha-1) - N < 0$, the RHS in (\ref{eq:minouthigh4}) is minimized by choosing $\frac{L_2}{L}$ as large as possible. From (\ref{eq:minouthigh2}), it can be noticed that (\ref{eq:minouthigh4}) is minimized by choosing $L_1=M$ and $L_2=N$, and (\ref{eq:minouthigh4}) becomes
\begin{equation}
d(\alpha) \leq \frac{MN\alpha}{M+N}.  \label{eq:minouter23}
\end{equation}
Now, the RHS in (\ref{eq:minouthigh7}) is required to be minimized. When $M - N(\alpha-1) < 0$, $\frac{L_2}{L}$ should be chosen as large as possible. By choosing $L_1=0$ and $L_2 > 0 $, $\frac{L_2}{L}$ is maximized, and (\ref{eq:minouthigh7}) becomes
\begin{equation}
d(\alpha) \leq M, \label{eq:minouter23a}
\end{equation}
which is not useful. When $M - N(\alpha-1) \geq 0$, $\frac{L_2}{L}$ or $\frac{L_2}{L_1}$ should be as low as possible. From (\ref{eq:minouthigh6}), it can be noticed that (\ref{eq:minouthigh7}) is minimized by choosing $L_1=M$ and $L_2=N$, and it reduces to
\begin{equation}
d(\alpha) \leq \frac{MN\alpha}{M+N}. \label{eq:minouter24}
\end{equation}
It can be noticed that in both the cases, the RHS are the same. But, (\ref{eq:minouter23}) and (\ref{eq:minouter24}) are active when $1 \leq \alpha \leq \frac{M+N}{M}$ and $1 \leq \alpha \leq \frac{M+N}{N}$, respectively. As $M \leq N$ and the RHS in  (\ref{eq:minouter24}) exceeds  the interference free GDOF per user, i.e., $M$, when $\alpha > \frac{M+N}{N}$, it is not required to consider the case $\frac{M+N}{N} < \alpha \leq \frac{M+N}{M}$. The outer bound in Lemmas 2 and 3 exceed the interference free GDOF in this case as mentioned in the proofs of these lemmas, and hence, these bounds are not taken into account in the high interference regime. Finally, taking the minimum of (\ref{eq:minouter24}) and $M$ results in (\ref{eq:minouter3}).\\
\textbf{Case b $(\frac{N}{M}+1 < K  < M + N \text{, where } \frac{N}{M} \text{ is an integer})$:} In this case, (\ref{eq:minouter13}) and  (\ref{eq:minouter15}) are minimized by choosing $L_1 =1$  and $L_2 = \frac{N}{M}$. This can be shown by following a similar procedure as in the previous case. Hence, the outer bound in \ref{lemma-outer1} in weak/moderate interference regime and high interference regime is of the same form as given in the first case of the Theorem.\\
\textbf{Case c $(\frac{N}{M} < K \leq \frac{N}{M}+1)$:} In this case, (\ref{eq:minouter13}) is minimized by choosing $L_1=1$ and $L_2 = K-1$. It is easy to verify that this choice of $L_1$ and $L_2$ maximizes $\frac{L_1}{L}$ and also satisfies the constraint in (\ref{eq:minouter12}). Hence, (\ref{eq:minouter13}) becomes
\begin{equation}
d(\alpha) \leq M(1-\alpha) + \frac{N\alpha}{K}.  \label{eq:minouter24a}
\end{equation}
In this case, the condition $\frac{L_2}{L_1} \geq \frac{N}{M}$ implies that (\ref{eq:minouter15}) arises only when $(K-1)M=N$. The RHS in (\ref{eq:minouter15}) is minimized by choosing $L_1=1$ and $L_2=K-1$, and (\ref{eq:minouter15}) becomes
\begin{eqnarray}
& d(\alpha) & \leq M - \frac{M\alpha}{K}.  \label{eq:minouter25}
\end{eqnarray}
With some algebraic manipulation, it can be shown that (\ref{eq:minouter25}) reduces to (\ref{eq:minouter24a}) when $(K-1)M=N$. Hence, choosing $L_1=1$ and $L_2 = K-1$ minimizes the outer bound in \ref{lemma-outer1}, and it is given by~(\ref{eq:minouter24a}).
 
The following interference regimes are considered for comparison with other outer bounds.\\
\textit{Weak interference regime $(0 \leq \alpha \leq \frac{1}{2})$:} In the weak interference regime, the outer bound in Lemma 2 reduces to:
\begin{equation}
d(\alpha) \leq M(1-\alpha) + (N-M)\alpha. \label{eq:minouter34b}
\end{equation}
Comparing (\ref{eq:minouter34b}) with (\ref{eq:minouter24a}), the following condition is obtained:
\begin{eqnarray}
& & M(1-\alpha) + \frac{N\alpha}{K} \leq M(1-\alpha) + (N-M)\alpha, \nonumber \\
\text{or } & & 2M \leq N, \label{eq:minouter35}
\end{eqnarray}
which is always satisfied in this case. Hence, the outer bound in (\ref{eq:minouter34b}) is loose compared to the outer bound in (\ref{eq:minouter24a}). Now, the outer bound in (\ref{eq:minouter24a}) is compared with the outer bound in Lemma 3.
\begin{eqnarray}
& & M(1-\alpha) +  \frac{1}{K-1}(N-M)\alpha \leq M(1-\alpha) + \frac{N\alpha}{K}, \nonumber \\
\text{or } & & KM \leq N, \label{eq:minouter35b}
\end{eqnarray}
which is satisfied in this case. Hence, Lemma 3 is active in the entire weak interference regime, which results in (\ref{eq:minouter4}).\\
\textit{Moderate interference regime $(\frac{1}{2} \leq \alpha \leq 1)$:} In the moderate interference regime, it is easy to see that the outer bound in Lemma 2 is loose compared to the outer bound in \ref{lemma-outer1}. Lemma 3 is tighter than the outer bound in (\ref{eq:minouter24a}), when
\begin{eqnarray}
& & M\alpha + \frac{1}{K-1}(N-M)(1-\alpha) \leq M(1-\alpha) + \frac{N\alpha}{K}, \nonumber \\
\text{or } & & \alpha \leq \frac{K}{2K-1}. \label{eq:minouter36}
\end{eqnarray}
Consequently, Lemma \ref{lemma-outer1} is active when $\frac{K}{2K-1} < \alpha \leq 1$. Taking the minimum of these two outer bounds results in (\ref{eq:minouter5}).\\
\textit{High interference regime $(\alpha \geq 1)$:} By employing the similar procedure as followed in the weak/moderate interference regime, it can be shown that the outer bound in Lemma \ref{lemma-outer1} is minimized by choosing $L_1=1$ and $L_2 = K-1$. Also, the outer bound in Lemma \ref{lemma-outer2} and \ref{lemma-outer3} are loose compared to Lemma \ref{lemma-outer1}, as they exceed the interference free GDOF per user, i.e., $M$. In this case, Lemma \ref{lemma-outer1} reduces to
\begin{equation}
d(\alpha) \leq \frac{1}{K}\lsqb N + (K-1)M(\alpha-1)\rsqb. \label{eq:minouter37}
\end{equation}
Finally, taking the minimum of (\ref{eq:minouter37}) and $M$ results in (\ref{eq:minouter6}), which completes the proof.\\
\emph{Remark:} There are a few other cases where it is possible to exactly characterize the performance of these outer bounds. For example, when $K < M + N < aK$, and integer $a \geq 2$, $\frac{M}{a}$ and $\frac{N}{a}$ are integers, choosing $L_1 = \frac{M}{a}$ and $L_2 = \frac{N}{a}$ minimizes the outer bound in Lemma~\ref{lemma-outer1}, and the outer bound is the same as given in the first case of the Theorem. In this case, $\frac{N}{M}$ need not be an integer. 

\subsection{Proof of Theorem \ref{th:theorem1_intnoise}}\label{sec:th-intasnoise}
When interference is treated as noise, the rate achieved by the individual user in case of MIMO GIC is
\begin{eqnarray}
& R_{j} & \geq \log \labs \Iden_{N} + \lcb \Iden_{N} + \rho^{\alpha}\sumneqmod\Hnew_{ji}\Pnew_{i}\Hnew_{ji}^{H}\rcb^{-1}\rho\Hjj\Pnew_{j}\Hjj^{H}\rabs, \nonumber \\
& & = \log\labs \Iden_{N} + \rho\Hjj\Pnew_{j}\Hjj^{H} + \rho^{\alpha}\sumneqmod\Hji\Pnew_{i}\Hji^{H}\rabs - \log\labs \Iden_{N} + \rho^{\alpha}\sumneqmod\Hji\Pnew_{i}\Hji^{H}\rabs, \nonumber \\
& & = r\log\rho + \min\lcb r^{'}, N- r\rcb\alpha\log\rho - \alpha r^{'}\log\rho + \oone, \label{eq:inner-intasnoise1}
\end{eqnarray}
where $r = \text{rank}(\Hjj\Pnew_{j}\Hjj^{H})$ and $r^{'} = \text{rank}\lb \sumneqmod\Hji\Pnew_{i}\Hji^{H} \rb$. 
The last equation is obtained using Lemma \ref{lemmaused1}. As the input covariance matrix is full rank, (\ref{eq:inner-intasnoise1}) becomes:
\begin{eqnarray}
& R_{j} & \geq M\log\rho + \min\lcb \min\lcb(K-1)M,N \rcb, N-M\rcb\alpha\log\rho - \min\lcb(K-1)M,N\rcb\alpha\log\rho + \oone. \label{eq:inner-intasnoise1a}
\end{eqnarray}
The achievable rate in (\ref{eq:inner-intasnoise1a}) can be further simplified by considering the following cases:
\subsubsection{Case 1 $\lb \frac{N}{M} < K \leq \frac{N}{M} + 1\rb$}
Here, $N-M \leq (K-1)M \leq N$, and hence, the rate in (\ref{eq:inner-intasnoise1a}) reduces to
\begin{eqnarray}
& R_{j} & \geq  M\log\rho + (N-M)\alpha\log\rho -  (K-1)M\alpha\log\rho + \oone. \label{eq:inner-intasnoise2}
\end{eqnarray}
Thus, the per user GDOF that can be achieved in this case is
\begin{eqnarray}
& d(\alpha) & \geq M + (N-M)\alpha -  (K-1)M\alpha, \nonumber \\
& & = M + (N - KM)\alpha.\label{eq:inner-intasnoise4}
\end{eqnarray}
\subsubsection{Case 2 $\lb K > \frac{N}{M} + 1\rb$}
Here, $\min\lcb (K-1)M,N\rcb = N$, and hence, the rate in (\ref{eq:inner-intasnoise1a}) becomes:
\begin{eqnarray}
& R_{j} & \geq M\log\rho + \min\lcb N , N-M\rcb\alpha\log\rho - N\alpha\log\rho + \oone, \nonumber \\
& &= M(1-\alpha)\log\rho + \oone. \label{eq:inner-intasnoise6}
\end{eqnarray}
The achievable per user GDOF in this case is:
\begin{equation}
d(\alpha) \geq M(1-\alpha). \label{eq:inner-intasnoise7}
\end{equation}
By combining (\ref{eq:inner-intasnoise4}) and (\ref{eq:inner-intasnoise7}) results in Theorem \ref{th:theorem1_intnoise}.
\subsection{Proof of Theorem \ref{th-highint1}}\label{sec:appendhigint}
Due to the symmetry of the problem, it is sufficient to consider the GDOF achieved by any particular user, say user $1$. Also consider the user subset $S \subseteq \{ 2, \ldots, K\}$, and let $S' \triangleq S \cup \{1\}$, i.e., $S$ is a subset of users excluding user $1$, while $S'$ always includes user $1$. The number of users in the set $S$ is denoted by $|S| \leq K-1$ and number of users in $S'$ is $|S| + 1$. The achievable GDOF is obtained by considering two different cases:
\subsubsection{Case 1 $\lb \frac{N}{M} < K \leq \frac{N}{M} + 1 \rb$}
Now, using the MAC channel formed at the receiver of user $1$ with the signals from the user set $S$, the achievable sum rate is bounded as: 
\begin{eqnarray}
& \sums R_{j} & \leq \log | \Iden_{N}  + \rho^{\alpha}\sums\Honej\Pj\Honej^{H} |, \nonumber \\
& & = \min \lcb |S|M, N \rcb\alpha\log\rho + \oone, \nonumber \\
& \text{or } R_{j} & \leq  M\alpha\log\rho + \oone  \label{eq:myinnerstrong1}
 \end{eqnarray}
In the above equation, it can be noticed that $|S|_{\max} = K-1$ and hence, $\min \lcb |S|M, N \rcb = |S|M$. Similarly, using the MAC channel formed at the receiver of user $1$ with the signals from the user set $S'$, the achievable sum rate is bounded as:
\begin{eqnarray}
& \displaystyle\sum_{j \in S'} R_{j} & \leq \log | \Iden_{N} + \rho\Honeone\Pnew_{1}\Honeone^{H} + \rho^{\alpha}\displaystyle\sum_{j\in S}\Honej \Pj\Honej^{H}|, \nonumber \\
& & = \min(|S|M,N)\alpha\log\rho + \min\lcb M, N -  \min(|S|M,N)\rcb\log\rho + \oone, \nonumber \\
& & = |S|M\alpha\log\rho + \min\lcb M, N - |S|M\rcb \log\rho + \oone.  \label{eq:myinnerstrong2}
\end{eqnarray}
Above equation is obtained by using Lemma \ref{lemmaused1} and simplified further based on following conditions.

When $\min\lcb M, N - |S|M\rcb = N - |S|M$, then following condition is obtained:
\begin{eqnarray}
& & N - |S|M \leq M, \nonumber \\
& & \text{or } \frac{N}{M} \leq 1 + |S| \leq  \frac{N}{M} + 1, \qquad (\because K \leq \frac{N}{M} + 1) \label{eq:myinnerstrong3}
\end{eqnarray}
When the condition in (\ref{eq:myinnerstrong3}) is satisfied, (\ref{eq:myinnerstrong2}) reduces to following form:
\begin{eqnarray}
& \sum_{j \in S'} R_{j} & \leq |S|M\alpha\log\rho + (N - |S|M) \log\rho + \oone, \nonumber \\
&  \text{or }R_{j} &\leq \frac{|S|M(\alpha -1) + N}{1 + |S|} \log\rho + \oone. \label{eq:myinnerstrong4}
\end{eqnarray}
The right hand side above is minimized when $|S| = |S|_{\text{max}} = K-1$; and recall that $K \leq \frac{N}{M} + 1$. Hence, (\ref{eq:myinnerstrong4}) becomes:
\begin{eqnarray}
 R_{j} \leq \frac{(K-1)M(\alpha -1) + N}{K} \log\rho + \oone. \label{eq:myinnerstrong5}
\end{eqnarray}
When $\min\lcb M, N - |S|M\rcb = M$, then following condition is obtained:
\begin{eqnarray}
& & M \leq N - |S|M, \nonumber \\
& & \text{or } 1 + |S| \leq \frac{N}{M}. \label{eq:myinnerstrong6}
\end{eqnarray}
When above condition is satisfied, (\ref{eq:myinnerstrong2}) becomes:
\begin{eqnarray}
& \sum_{j \in S'} R_{j} & \leq |S|M\alpha\log\rho + M\log\rho + \oone, \nonumber \\
& \text{or } R_{j} & \leq \frac{|S|\alpha + 1}{|S|+1}M\log\rho + \oone. \label{eq:myinnerstrong7}
\end{eqnarray}
The term in the right hand side of the above equation is minimized when $|S|=0$ and it results in following equation.
\begin{eqnarray}
R_{j} \leq M\log\rho + \oone. \label{eq:myinnerstrong7b}
\end{eqnarray}

The achievable rate is obtained by taking minimum of (\ref{eq:myinnerstrong1}), (\ref{eq:myinnerstrong5}) and (\ref{eq:myinnerstrong7b}). It can be observed that (\ref{eq:myinnerstrong1}) becomes superfluous given (\ref{eq:myinnerstrong7b}). Finally, taking minimum of (\ref{eq:myinnerstrong5}) and (\ref{eq:myinnerstrong7b}) results in case 1 of Theorem \ref{th-highint1}.

\subsubsection{Case  2 $\lb K > \frac{N}{M} + 1\rb$}
Now, using the MAC channel formed at the receiver of user $1$ with the signals from the user set $S$, the achievable sum rate is: 
\begin{eqnarray}
& \sums R_{j} & \leq \log | \Iden_{N}  + \rho^{\alpha}\sums\Honej\Pj\Honej^{H} |, \nonumber \\
& & = \min \lcb |S|M, N \rcb\alpha\log\rho + \oone. \label{eq:innerstrongint2}
 \end{eqnarray}
If $\min \lcb |S|M, N \rcb = |S|M$, then the above equation simplifies to:
\begin{equation}
R_{j} \leq  M\alpha\log\rho + \oone \label{eq:innerstrongint3}
\end{equation}
If $\min \lcb |S|M, N \rcb = N$, then (\ref{eq:innerstrongint2}) simplifies to:
\begin{equation}
R_{j} \leq \frac{N\alpha}{|S|}\log\rho + \oone,  \label{eq:innerstrongint4} \text{ where } |S| \leq K-1.
\end{equation}
Similarly, using the MAC channel formed at the receiver of user $1$ with the signals from the user set $S'$, the achievable sum rate is bounded as:
\begin{eqnarray}
 \sum_{j \in S'} R_{j} \leq \log | \Iden_{N} + \rho\Honeone\Pnew_{1}\Honeone^{H} + \rho^{\alpha}\displaystyle\sum_{j\in S}\Honej \Pj\Honej^{H}|. \label{eq:innerstrongint5}
\end{eqnarray}
Again, using Lemma \ref{lemmaused1}, above simplifies to:
 \begin{eqnarray}
 \sum_{j \in S'} R_j \leq  \min\{|S|M,N\}\alpha\log\rho + \min\{M, N - \min(|S|M,N)\}\log\rho + \oone, \label{eq:innerstrongint6}
\end{eqnarray}
The above equation is simplified further under following cases.\\
\textbf{Case a}:\\
If $\min\{|S|M,N\} = |S|M$, (\ref{eq:inner-intasnoise6}) becomes:
\begin{eqnarray}
 \sum_{j \in S'} R_j \leq  |S|M\alpha\log\rho + \min\{M, N - |S|M\}\log\rho + \oone, \label{eq:innerstrongint8}
\end{eqnarray}
If $\min\{M, N - |S|M\} = N - |S|M$, then we have the following condition:
\begin{eqnarray}
& & N - |S|M \leq M, \nonumber \\
\text{or } & &  \frac{N}{M} - 1 \leq |S|. \label{eq:innerstrongint9}
\end{eqnarray}
Hence, we obtain the following condition:
\begin{eqnarray}
& & \frac{N}{M} - 1 \leq |S| \leq \frac{N}{M}. \label{eq:innerstrongint10}
\end{eqnarray}
When the condition in (\ref{eq:innerstrongint10}) is satisfied, (\ref{eq:innerstrongint6}) becomes:
\begin{eqnarray}
 &  \sum_{j \in S'} R_j & \leq  |S|M\alpha\log\rho + (N - |S|M)\log\rho + \oone, \nonumber \\
& & = \lsqb |S|M(\alpha - 1) + N \rsqb \log\rho + \oone, \nonumber \\
& \text{or } R_{j} & \leq \frac{|S|M(\alpha - 1) + N }{1 + |S|}\log\rho + \oone. \label{eq:innerstrongint11}
\end{eqnarray}
The value of  $|S|$ is required to be chosen such that (\ref{eq:innerstrongint11}) is minimized, and the condition in (\ref{eq:innerstrongint10}) is also satisfied. Above equation is simplified in the later part of derivation.

When $\min\{M, N - |S|M\} = M$, then following condition is obtained:
\begin{eqnarray}
& & M \leq N - |S|M, \nonumber \\
& & \text{or } |S| \leq \frac{N}{M} - 1. \label{eq:innerstrongint12}
\end{eqnarray}
Hence, following condition is deduced:
\begin{eqnarray}
|S| \leq \frac{N}{M} - 1 < \frac{N}{M}. \label{eq:innerstrongint13}
\end{eqnarray}
Under this condition, (\ref{eq:innerstrongint8}) becomes:
\begin{eqnarray}
& \sum_{j \in S'} R_j & \leq  |S|M\alpha\log\rho + M \log\rho + \oone, \nonumber \\
& & = (|S|\alpha + 1)M\log\rho + \oone, \nonumber \\
\text{or } & R_{j} &\leq \frac{|S|\alpha + 1}{|S| + 1}M\log\rho + \oone. \label{eq:innerstrongint14}
\end{eqnarray}
The above needs to be minimized for $|S| \leq \frac{N}{M} - 1$, and $|S| = 0$ minimizes it. This results in 
\begin{eqnarray}
R_{j} \leq M\log\rho + \oone. \label{eq:innerstrongint15}
\end{eqnarray}

\textbf{Case b}:\\
When $\min\lcb|S|M, N\rcb = N$, we have $|S|M \geq N$ and (\ref{eq:innerstrongint6}) reduces to:
\begin{eqnarray}
& \sum_{j \in S'} R_j & \leq N\alpha\log\rho + \oone, \nonumber \\
& \text{or } R_{j} & \leq \frac{N\alpha}{|S| + 1}\log\rho + \oone, \label{eq:innerstrongint17}
\end{eqnarray}
Above equation is minimized when $|S|$ takes the maximum value. As $K > \frac{N}{M} + 1$, the maximum value of $|S|$ is $|S|_{\text{max}} = K -1$. Hence, (\ref{eq:innerstrongint17}) becomes:
\begin{eqnarray}
 R_{j} \leq \frac{N\alpha}{K}\log\rho + \oone \label{eq:innerstrongint18}
\end{eqnarray}
Finally taking minimum of (\ref{eq:innerstrongint3}), (\ref{eq:innerstrongint4}), (\ref{eq:innerstrongint11}), (\ref{eq:innerstrongint15}) and (\ref{eq:innerstrongint18}) the achievable GDOF is obtained as described below.

Given (\ref{eq:innerstrongint15}), the equation in (\ref{eq:innerstrongint3}) becomes superfluous as $\alpha > 1$. Similarly given (\ref{eq:innerstrongint18}), the equation in (\ref{eq:innerstrongint4}) is redundant. It can also be observed that (\ref{eq:innerstrongint11}) is redundant given (\ref{eq:innerstrongint18}). It can be proved as follows:

Let assume that $\frac{N}{M}$ is an integer. Then there are two possible values of $|S|$ which satisfies the condition in (\ref{eq:innerstrongint10}). When $|S| = \frac{N}{M} - 1$, (\ref{eq:innerstrongint11}) reduces to following form.
\begin{eqnarray}
& R_{j} & \leq \frac{\lb \frac{N}{M} -1\rb M(\alpha - 1) + N}{\frac{N}{M}}\log\rho + \oone  \label{eq:innerstrongint19a}
\end{eqnarray}
At $\alpha=1$, the right hand side reduces to $M$ and with increase in $\alpha$, the right hand side in (\ref{eq:innerstrongint19a}) also increases. Hence, given (\ref{eq:innerstrongint15}), (\ref{eq:innerstrongint19a}) is redundant. When $|S| = \frac{N}{M}$, (\ref{eq:innerstrongint11}) becomes
\begin{eqnarray}
R_{j} \leq \frac{N\alpha}{\frac{N}{M} + 1}\log\rho + \oone \label{eq:innerstrongint19b}
\end{eqnarray}
As $K > \frac{N}{M} + 1$, given (\ref{eq:innerstrongint18}), the rate in (\ref{eq:innerstrongint19b}) becomes superfluous. When $\frac{N}{M}$ is not an integer, then there is only one value of $|S|$ i.e., $|S| = \lfloor \frac{N}{M}\rfloor$ satisfies the condition in (\ref{eq:innerstrongint10}). With some algebraic manipulations, it can be shown that (\ref{eq:innerstrongint11}) is redundant given (\ref{eq:innerstrongint18}). Finally by taking minimum of (\ref{eq:innerstrongint15}) and (\ref{eq:innerstrongint18}) results in second part of Theorem  \ref{th-highint1}.
\subsection{Proof of Theorem \ref{th:moderateint1}}\label{sec:appendmoderateint1}
First we calculate the rate obtained due to the private part of the message. As the private message is decoded last, the rate of the private message is obtained by treating all remaining users' private messages as noise. Due to symmetry of the problem, it is sufficient to consider only one particular user. The rate achieved by the private part is
\begin{eqnarray}
& R_{p,j} &\leq \log\labs\Iden_{N} + \lb \Iden_{N} + \sumneq\Hji\Pdashi\Hji^{H}\rb^{-1}\rho^{1-\alpha} \Hjj\Pj\Hjj^{H}\rabs \nonumber \\
& & = M(1-\alpha)\log\rho + \oone \label{eq:innermodrateint1}
\end{eqnarray}
Next, we calculate the rate obtained due to the common part of the message. The following two cases are considered to obtain the achievable GDOF.
\subsubsection{Case 1 $\lb \frac{N}{M} < K \leq \frac{N}{M} + 1\rb$}
For the common part, different subsets of users are considered as in Appendix \ref{sec:appendhigint}. Consider the set $S^{'} \subseteq \{1,2,\ldots,K\}$, where user  1 is always included in the subset. Since, common messages need to be decodable at every receiver, user 1 should be able to decode the other users' common messages as well as its own common message. While decoding the common message, it should treat all other users' private messages as well as its own private message as noise. The common messages form a MAC channel at Receiver 1. Similarly, $K-1$ MAC channel will be formed at every receiver. The achievable rate is the intersection of $K$ such MAC regions. Due to symmetry of the problem, it is sufficient to consider only one specific receiver. The achievable rate due to the signals from $S^{'}$ is:
 \begin{equation}
\displaystyle\sum_{j \in S'} R_{c,j} \leq \log\labs  \Iden_{N} + \lb \Iden_{N} + \displaystyle\sum_{j \in S}\Honej\Pj\Honej^{H} + \rho^{1-\alpha}\Honeone\Pone\Honeone^{H} \rb^{-1}\displaystyle\sum_{j \in S'} P_{c,j}\Honej\Pj\Honej^{H}\rabs \label{eq:innermodrateint2}
\end{equation}
Here
\begin{eqnarray}
P_{c,j} = \left\{ 
\begin{array}{l l}
  \rho^{\alpha}-1 & \quad \mbox{for $ j\neq 1$ }\\
  \rho - \rho^{1-\alpha} & \quad \mbox{for $j=1$}\\ \end{array} \right. \label{eq:innermodrateint3}
\end{eqnarray}
Equation (\ref{eq:innermodrateint2}) becomes:
\begin{eqnarray}
 & \sum_{j \in S'} R_{c,j} &\leq \log|\Iden_{N} + \displaystyle\sum_{j\in S }\Honej\Pj\Honej^{H} + \rho^{1-\alpha}\Honeone\Pone \Honeone^{H} + \lb\rho -\rho^{1-\alpha}\rb\Honeone\Pone\Honeone^{H} \nonumber \\
& & \quad + \lb\rho^{\alpha} -1 \rb \displaystyle\sum_{j\in S}\Honej\Pj\Honej^{H}|- \log\labs \Iden_{N} + \rho^{1-\alpha}\Honeone\Pone\Honeone^{H}\rabs, \nonumber \\
& & = \log|\Iden_{N} + \rho\Honeone\Pone\Honeone^{H} + \rho^{\alpha}\displaystyle\sum_{j\in S }\Honej\Pj\Honej^{H}|- \log\labs \Iden_{N} + \rho^{1-\alpha}\Honeone\Pone\Honeone^{H}\rabs + \oone, \nonumber \\
& & = M\log\rho + \min\lcb\min\lcb N, |S|M \rcb, N-M \rcb\alpha\log\rho - M(1-\alpha)\log\rho + \oone, \nonumber \\
& & =  M\alpha\log\rho + \min\lcb\min\lcb N, |S|M \rcb, N-M \rcb\alpha\log\rho + \oone, \nonumber \\
&  \text{or } \sum_{j \in S'} R_{c,j} & \leq  M\alpha\log\rho + \min\lcb |S|M , N-M \rcb\alpha\log\rho + \oone. \label{eq:innermoderateint4}
\end{eqnarray}
Equation (\ref{eq:innermoderateint4}) is obtained by using Lemma \ref{lemmaused1} and $|S|M \leq (K-1)M \leq N$.
The above equation is simplified under following cases.\\
\textbf{Case  a}:\\
When $\min\lcb |S|M , N-M \rcb = N-M$, then we have the following condition:
\begin{eqnarray}
& & N -M \leq |S|M, \nonumber \\
& & \text{or } N \leq (1 + |S|)M. \label{eq:innermoderateint5}
\end{eqnarray}
Since $\frac{N}{M} < K \leq \frac{N}{M}+1$ and $|S|\leq K-1$, the above equation can only be satisfied for $|S| = K-1$, and hence (\ref{eq:innermoderateint4}) becomes:
\begin{eqnarray}
& \sum_{j \in S'} R_{c,j} & \leq M\log\rho + (N-M)\alpha\log\rho + \oone \nonumber \\
& & = N\alpha\log\rho + \oone \nonumber \\
& \text{or }R_{c,j} & \leq \frac{N\alpha}{1 + |S|}\log\rho + \oone \nonumber \\
& & = \frac{N\alpha}{K}\log\rho + \oone \qquad  \label{eq:innermoderateint6}
\end{eqnarray}
\textbf{Case b}:\\
When $\min\lcb |S|M , N-M \rcb = |S|M$, then we have $(1 + |S|)M \leq N$.
This condition is satisfied when $|S| < K-1$ and (\ref{eq:innermoderateint4}) simplifies to:
\begin{eqnarray}
& \sum_{j \in S'} R_{c,j} & \leq  M\alpha\log\rho + |S|M\alpha\log\rho + \oone, \nonumber \\
& \text{or } R_{c,j} & \leq  M \alpha\log\rho + \oone.  \label{eq:innermoderateint6b}
\end{eqnarray}
Now consider the user subset $S \subseteq \{ 2, \ldots, K\}$. A MAC channel is formed at Receiver 1 due to the signals from users in $S$. The achievable sum rate in this case is:
\begin{eqnarray}
& \sums R_{c,j} & \leq \log\labs\Iden_{N}  + \lb \Iden_{N} +\displaystyle\sums\Honej\Pj\Honej^{H} + \rho^{1-\alpha}\Honeone\Pone\Honeone^{H}\rb^{-1} (\rho^{\alpha}-1)\sums\Honej\Pj\Honej^{H}\rabs, \nonumber \\
& & = |S|M\alpha\log\rho + \min\lcb M, N- |S|M\rcb(1- \alpha)\log\rho - M(1-\alpha)\log\rho + \oone, \nonumber \\
& & (\because \alpha \geq 1- \alpha) \label{eq:innermoderateint7}
\end{eqnarray}
Equation (\ref{eq:innermoderateint7}) is simplified under following cases:\\
\textbf{Case a}:\\
When $\min\lcb M, N- |S|M\rcb = N - |S|M$, then $N - |S|M \leq M$.
This condition is satisfied when $|S| = K-1$, and (\ref{eq:innermoderateint7}) reduces to:
\begin{eqnarray}
& \sums R_{c,j} & \leq (K-1)M\alpha\log\rho + (N - (K-1)M)(1- \alpha)\log\rho - M(1-\alpha)\log\rho + \oone, \nonumber \\
& & = \lsqb M \lcb \alpha\lb 2K-1 \rb - K \rcb + N(1- \alpha) \rsqb \log\rho + \oone, \nonumber \\
& \text{or } R_{c,j} &\leq \frac{1}{K-1}\lsqb M \lcb \alpha\lb 2K-1 \rb - K \rcb + N(1- \alpha) \rsqb \log\rho + \oone. \label{eq:innermoderateint9}
\end{eqnarray}
\textbf{Case b}:\\
When $\min\lcb M, N- |S|M\rcb = M$, it results in $(1 + |S|)M \leq N$.
Under this condition, (\ref{eq:innermoderateint7}) reduces to
\begin{eqnarray}
& \sums R_{c,j} & \leq |S|M\alpha\log\rho + \oone, \nonumber \\
& \text{or }  R_{c,j} &\leq  M\alpha\log\rho + \oone. \label{eq:innermoderateint11}
\end{eqnarray}
The achievable rate is obtained by taking minimum of (\ref{eq:innermoderateint6}), (\ref{eq:innermoderateint6b}), (\ref{eq:innermoderateint9}) and (\ref{eq:innermoderateint11}). As $N < KM$, (\ref{eq:innermoderateint6b}) and (\ref{eq:innermoderateint11}) becomes superfluous given (\ref{eq:innermoderateint6}). The achievable GDOF by common part of the message is thus given by
\begin{eqnarray}
d_c(\alpha) \geq \min\lcb \frac{N\alpha}{K}, \frac{1}{K-1}\lsqb M \lcb \alpha\lb 2K-1 \rb - K \rcb + N(1- \alpha) \rsqb \rcb \label{eq:innermoderateint12}
\end{eqnarray}
The total GDOF achievable by private part and common part together is:
\begin{eqnarray}
& d(\alpha) & = d_p(\alpha) + d_c(\alpha) \nonumber \\
& & \geq M(1-\alpha) +  \min\lcb \frac{N\alpha}{K}, \frac{1}{K-1}\lsqb M \lcb \alpha\lb 2K-1 \rb - K \rcb + N(1- \alpha) \rsqb \rcb \label{eq:innermoderateint13}
\end{eqnarray}
This completes the proof for first case of Theorem \ref{th:moderateint1}.
\subsubsection{Case 2 $\lb K > \frac{N}{M} + 1 \rb$}
We consider the set $S^{'} \subseteq \{1,2,\ldots,K\}$ as in Case 1. The achievable rate in this case using (\ref{eq:innermodrateint2}) and (\ref{eq:innermodrateint3}) simplifies  to:
\begin{eqnarray}
& \sum_{j \in S'} R_{c,j} &\leq \log|\Iden_{N} + \rho\Honeone\Pone\Honeone^{H} + \rho^{\alpha}\displaystyle\sum_{j\in S }\Honej\Pj\Honej^{H}|- \log\labs \Iden_{N} + \rho^{1-\alpha}\Honeone\Pone\Honeone^{H}\rabs + \oone, \nonumber \\
& & = M\log\rho + \min\lcb\min\lcb N, |S|M \rcb, N-M \rcb\alpha\log\rho - M(1-\alpha)\log\rho + \oone. \label{eq:innermoderateint14}
\end{eqnarray}
Above equation is simplified under following cases.\\
\textbf{Case a}:\\
When $\min\{N,|S|M\} = N$, then $ \frac{N}{M} \leq |S|$.
The maximum value of $|S|$ which satisfies this condition is $K-1$. Under this condition, (\ref{eq:innermoderateint14}) becomes:
\begin{eqnarray}
& \sum_{j \in S'} R_{c,j} &\leq M\log\rho + \min\lcb N , N-M \rcb\alpha\log\rho - M(1-\alpha)\log\rho + \oone, \nonumber \\
& & = M\log\rho + (N-M)\alpha\log\rho - M(1-\alpha)\log\rho + \oone, \nonumber \\
& & = N\alpha\log\rho + \oone, \nonumber \\
& \text{or } R_{c,j} &\leq \frac{N\alpha}{1 + |S|}\log\rho + \oone. \label{eq:innermoderateint16}
\end{eqnarray}
Above equation is minimized when $|S|$ takes its maximum value i.e., $|S| = K-1$ and (\ref{eq:innermoderateint16}) becomes:
\begin{eqnarray}
R_{c,j} \leq \frac{N\alpha}{K}\log\rho + \oone. \label{eq:innermoderateint17}
\end{eqnarray}
\textbf{Case b}:\\
When $\min\{N,|S|M\} = |S|M$, then $|S| \leq \frac{N}{M}$.
Equation (\ref{eq:innermoderateint14}) becomes:
\begin{eqnarray}
& \sum_{j \in S'} R_{c,j} &\leq  M\log\rho + \min\lcb |S|M , N-M \rcb\alpha\log\rho - M(1-\alpha)\log\rho + \oone. \label{eq:innermoderateint19}
\end{eqnarray}
When $\min\lcb |S|M , N-M \rcb = N- M$, then we have
\begin{eqnarray}
& & \frac{N}{M} -1 \leq  |S|. \label{eq:innermoderateint20}
\end{eqnarray}
Hence, we have
\begin{eqnarray}
& & \frac{N}{M} \leq |S| + 1 \leq \frac{N}{M} + 1. \label{eq:innermoderateint21}
\end{eqnarray}
The achievable rate in this case is
\begin{eqnarray}
& \sum_{j \in S'} R_{c,j} &\leq  M\log\rho + (N-M)\alpha\log\rho - M(1-\alpha)\log\rho + \oone, \nonumber \\
& & = N\alpha\log\rho + \oone, \nonumber \\
& \text{or } R_{c,j} & \leq \frac{N\alpha}{1 + |S|}\log\rho + \oone. \label{eq:innermoderateint22}
\end{eqnarray}
The above equation is minimized by taking largest integer value of $|S|$ which satisfies the condition in (\ref{eq:innermoderateint21}).

When $\min\lcb |S|M , N-M \rcb = |S|M$, then we have
\begin{eqnarray}
& & 1 + |S| \leq \frac{N}{M}. \label{eq:innermoderateint23}
\end{eqnarray}
Then the following condition is deduced:
\begin{eqnarray}
|S| < 1 + |S| \leq \frac{N}{M}. \label{eq:innermoderateint24}
\end{eqnarray}
Under this condition, (\ref{eq:innermoderateint19}) becomes:
\begin{eqnarray}
& \sum_{j \in S'} R_{c,j} &\leq  M\alpha\log\rho + |S|M\alpha\log\rho + \oone, \nonumber \\
& & =  (1 + |S|)M\alpha\log\rho + \oone, \nonumber \\
& \text{or } R_{c,j} &\leq  M\alpha\log\rho + \oone. \label{eq:innermoderateint25}
\end{eqnarray}

Now consider the user set $S \subseteq \{2, 3, \ldots, K\}$. As this forms a MAC channel at receiver 1, from (\ref{eq:innermoderateint7}), the following rate equation is obtained:
\begin{eqnarray}
& \sums R_{c,j} & \leq \log|\Iden_{N}  + \rho^{\alpha}\displaystyle\sum_{j\in S }\Honej\Pj\Honej^{H} + \rho^{1-\alpha}\Honeone\Pone\Honeone^{H}| - \log| \Iden_{N} + \rho^{1-\alpha}\Honeone\Pone\Honeone^{H}|, \nonumber \\
& & = \min\lcb N, |S|M\rcb \alpha\log\rho + \min\lcb M, N- \min\lb N, |S|M\rb \rcb(1- \alpha)\log\rho - M(1-\alpha)\log\rho + \oone. \nonumber \\ \label{eq:innermoderateint26}
\end{eqnarray}
where the above uses the fact that $\alpha > 1-\alpha$ in the moderate interference regime. The above equation is simplified under following cases.\\
\textbf{Case a}:\\
When $\min\lcb N, |S|M\rcb = |S|M$, we have $|S| \leq \frac{N}{M}$.
When this condition is satisfied, (\ref{eq:innermoderateint26}) becomes
\begin{eqnarray}
& \sums R_{c,j} & \leq |S|M\alpha\log\rho + \min\lcb M, N- |S|M \rcb(1- \alpha)\log\rho - M(1-\alpha)\log\rho + \oone. \label{eq:innermoderateint27}
\end{eqnarray}
Above equation is further simplified by considering following two cases.

When $\min\lcb M, N- |S|M \rcb = N - |S|M$, we have $\frac{N}{M} - 1 \leq |S|$.
Hence, the following condition is obtained.
\begin{equation}
\frac{N}{M} - 1 \leq |S| \leq \frac{N}{M} \label{eq:innermoderateint29}
\end{equation}
When the above condition is satisfied, (\ref{eq:innermoderateint27}) becomes:
\begin{eqnarray}
& \sums R_{c,j} &\leq |S|M(2\alpha - 1)\log\rho + (N-M)(1-\alpha)\log\rho + \oone, \nonumber \\
& \text{or } R_{c,j} &\leq M(2\alpha-1)\log\rho + \frac{(N-M)(1-\alpha)}{|S|}\log\rho + \oone. \label{eq:innermoderateint30}
\end{eqnarray}
The above equation is required to be minimized by taking largest possible integer value of $|S|$, which also satisfies the condition in (\ref{eq:innermoderateint29}). Assume that $\frac{N}{M}$ is an integer. Equation (\ref{eq:innermoderateint30}) is minimized by taking $|S| = \frac{N}{M}$ and (\ref{eq:innermoderateint30}) becomes:
\begin{eqnarray}
&  R_{c,j} & \leq M(2\alpha-1)\log\rho + \frac{(N-M)(1-\alpha)}{\frac{N}{M}}\log\rho + \oone.\label{eq:innermoderateint31} 
\end{eqnarray}
When $\min\lcb M, N- |S|M \rcb = M$, then it results in following condition:
\begin{eqnarray}
|S| \leq \frac{N}{M} - 1 \label{eq:innermoderateint32} 
\end{eqnarray}
From (\ref{eq:innermoderateint26}) and (\ref{eq:innermoderateint32}), following condition is obtained:
\begin{eqnarray}
|S| \leq \frac{N}{M} - 1 < \frac{N}{M} \label{eq:innermoderateint33} 
\end{eqnarray}
Under this condition, (\ref{eq:innermoderateint27}) becomes:
\begin{eqnarray}
&  R_{c,j} &\leq |S|M\alpha\log\rho + M (1- \alpha)\log\rho - M (1 - \alpha)\log\rho + \oone, \nonumber \\
& & = |S|M\alpha\log\rho + \oone, \nonumber \\
& \text{or } R_{c,j} & \leq  M\alpha\log\rho + \oone. \label{eq:innermoderateint34}
\end{eqnarray}
\textbf{Case b}:\\
 When $\min \lcb N, |S|M\rcb = N$, then following condition is obtained:
 \begin{eqnarray}
 \frac{N}{M} \leq |S| \label{eq:innermoderateint35}
 \end{eqnarray}
 When the above equation is satisfied, (\ref{eq:innermoderateint26}) becomes:
 \begin{eqnarray}
 & \sums R_{c,j} & \leq N\alpha\log\rho - M(1-\alpha)\log\rho + \oone, \nonumber \\
 & \text{or } R_{c,j} & \leq \frac{N\alpha - M(1-\alpha)}{|S|}\log\rho + \oone. \label{eq:innermoderateint36}
 \end{eqnarray}
The above equation is minimized when $|S|$ takes the maximum value, and the maximum value of $|S|$ which satisfies the condition in (\ref{eq:innermoderateint35}) is $K-1$. As a result, (\ref{eq:innermoderateint36}) becomes:
\begin{eqnarray}
R_{c,j} & \leq \frac{N\alpha - M(1-\alpha)}{K-1}\log\rho + \oone. \label{eq:innermoderateint37}
\end{eqnarray}
The achievable rate by the common part when $K > \frac{N}{M} + 1$ of the message is obtained by taking minimum of (\ref{eq:innermoderateint17}), (\ref{eq:innermoderateint22}), (\ref{eq:innermoderateint25}), (\ref{eq:innermoderateint31}), (\ref{eq:innermoderateint34}) and (\ref{eq:innermoderateint37}). It can be observed that (\ref{eq:innermoderateint25}) and (\ref{eq:innermoderateint34}) are redundant given (\ref{eq:innermoderateint17}) as $N < KM$. As $K > \frac{N}{M} + 1$,  (\ref{eq:innermoderateint31}) is redundant given (\ref{eq:innermoderateint37}). Also, (\ref{eq:innermoderateint22}) is redundant given (\ref{eq:innermoderateint17}). It is important to notice that when $\frac{N}{M}$ is not an integer, the above argument remains valid as $\lfloor\frac{N}{M}\rfloor < \frac{N}{M}$. Now the achievable GDOF obtained by the common part of the message is:
\begin{eqnarray}
d_{c}(\alpha) \geq \min\lcb \frac{N\alpha}{K}, \frac{N\alpha - M(1-\alpha)}{K-1}\rcb. \label{eq:innermoderateint37b}
\end{eqnarray}
The per user GDOF achievable in this case is:
\begin{eqnarray}
& d(\alpha) & = d_{p}(\alpha) + d_{c}(\alpha) \nonumber \\
& & \geq M(1-\alpha) + \min\lcb \frac{N\alpha}{K}, \frac{N\alpha - M(1-\alpha)}{K-1}\rcb. \label{eq:innermoderateint38}
\end{eqnarray}
This completes the proof for case 2 of Theorem \ref{th:moderateint1}.
\subsection{Proof of Theorem \ref{th:lowinter1}}\label{sec:appendlowint1}
The rate achieved by the private part is same as that in case of moderate interference case and is given by:
\begin{eqnarray}
R_{p,j} \leq M(1-\alpha)\log\rho + \oone. \label{eq:innerlowint1}
\end{eqnarray}
In order to obtain the rate for common part the same procedure is followed as described in the moderate interference case. The following two cases are considered:
\subsubsection{Case 1 $\lb \frac{N}{M} < K \leq \frac{N}{M} + 1\rb$}
First consider the MAC channel formed at Receiver 1 due to the users in $S \subseteq \{2,\ldots,K\}$. The sum rate constraint in this case is:
\begin{eqnarray}
& \sums R_{c,j} & \leq \log|\Iden_{N}  + \rho^{\alpha}\displaystyle\sum_{j\in S }\Honej\Pj\Honej^{H} + \rho^{1-\alpha}\Honeone\Pone\Honeone^{H}| - \log| \Iden_{N} + \rho^{1-\alpha}\Honeone\Pone\Honeone^{H}| \nonumber \\
& & = M(1-\alpha)\log\rho + \min\lcb \min\lcb N, |S|M\rcb, N- M\rcb\alpha\log\rho - M(1-\alpha)\log\rho + \oone \nonumber \\
& & = \min\lcb |S|M, N- M\rcb\alpha\log\rho + \oone.  \label{eq:innerlowint2}
\end{eqnarray}
When $\min\lcb |S|M, N- M\rcb = N - M$, then $N \leq (1 + |S|)M $. This is possible when $|S| = K-1$ and (\ref{eq:innerlowint2}) becomes:
\begin{eqnarray}
& \sums R_{c,j} & \leq (N-M)\alpha\log\rho + \oone, \nonumber \\
& \text{or } R_{c,j} & \leq \frac{N-M}{K-1}\alpha\log\rho + \oone. \label{eq:innerlowint4}
\end{eqnarray}
When  $\min\lcb |S|M, N- M\rcb = |S|M$, $(1 + |S|)M  \leq N$. This condition results when $|S| < K-1$, and under this condition (\ref{eq:innerlowint2}) reduces to:
\begin{eqnarray}
& \sums R_{c,j} & \leq |S|M\alpha\log\rho + \oone, \nonumber \\
& \text{or } R_{c,j} & \leq  M\alpha\log\rho + \oone. \label{eq:innerlowint6}
\end{eqnarray}
Now consider the user set $S^{'} = S \cup \{1\}$, where user 1 is always included. The sum rate constraint for common part of the message:
\begin{eqnarray}
& \sum_{j \in S'} R_{c,j} & \leq \log|\Iden_{N} + \rho\Honeone\Pone\Honeone^{H} + \rho^{\alpha}\displaystyle\sum_{j\in S }\Honej\Pj\Honej^{H}|- \log\labs \Iden_{N} + \rho^{1-\alpha}\Honeone\Pone\Honeone^{H}\rabs + \oone \nonumber \\
& & = M\log\rho + \min\lcb \min\lcb N, |S|M\rcb, N- M\rcb\alpha\log\rho - M(1-\alpha)\log\rho + \oone \nonumber \\
& & = \alpha M\log\rho + \min\lcb \min\lcb N, |S|M\rcb, N- M\rcb\alpha\log\rho + \oone. \label{eq:innerlowint7}
\end{eqnarray}
As $K \leq \frac{N}{M} + 1$, we have $(K-1)M \leq N$ or $|S|M \leq N$. Equation (\ref{eq:innerlowint7}) further simplifies to:
\begin{eqnarray}
& \sum_{j \in S'} R_{c,j} & \leq  M\alpha\log\rho + \min\lcb  |S|M, N- M\rcb\alpha\log\rho + \oone. \label{eq:innerlowint8}
\end{eqnarray}
It can be noticed that (\ref{eq:innermoderateint4}) and (\ref{eq:innerlowint8}) are same. and hence can be simplified as in case of (\ref{eq:innermoderateint4}) to obtain following equations.

When $\min\lcb  |S|M, N- M\rcb = N-M$, then (\ref{eq:innerlowint8}) becomes
\begin{equation}
R_{c,j} \leq \frac{N\alpha}{K}\log\rho + \oone. \label{eq:innerlowint11}
\end{equation}
When $\min\lcb  |S|M, N- M\rcb = |S|M$, then (\ref{eq:innerlowint8}) becomes
\begin{equation}
R_{c,j}  \leq  M\alpha\log\rho + \oone. \label{eq:innerlowint13}
\end{equation}
The achievable rate by common part of the message is obtained by taking minimum of (\ref{eq:innerlowint4}), (\ref{eq:innerlowint6}), (\ref{eq:innerlowint11}) and (\ref{eq:innerlowint13}). With some algebraic manipulation, it can be shown that given (\ref{eq:innerlowint4}), all the remaining equations become superfluous. The achievable GDOF due to common part of the message is:
\begin{eqnarray}
d_{c}(\alpha) \geq \frac{N-M}{K-1}\alpha. \label{eq:innerlowint14} 
\end{eqnarray}
The per user GDOF achievable in this case is:
\begin{eqnarray}
d(\alpha) \geq M(1-\alpha) + \frac{N-M}{K-1}\alpha. \label{eq:innerlowint15} 
\end{eqnarray}
\subsubsection{Case 2 $\lb K > \frac{N}{M} + 1\rb$}
As in the previous case, first consider the MAC channel formed at Receiver 1, due to the users in $S \subseteq \{2,\ldots,K\}$. From (\ref{eq:innerlowint2}), the sum rate constraint in this case becomes:
\begin{eqnarray}
& \sums R_{c,j} & \leq \min\lcb \min\lcb N, |S|M\rcb, N- M\rcb\alpha\log\rho + \oone. \label{eq:innerlowint16} 
\end{eqnarray}
When $\min\lcb N, |S|M\rcb = N$, then we have $ N \leq |S|M$. Under this condition, (\ref{eq:innerlowint16}) reduces to:
\begin{equation}
 R_{c,j} \leq \frac{N-M}{|S|}\alpha\log\rho + \oone.  \label{eq:innerlowint18}
\end{equation}
Above equation is minimized when $|S| = |S|_{\text{max}} = K-1$ and (\ref{eq:innerlowint18}) becomes:
\begin{eqnarray}
R_{c,j} \leq \frac{N-M}{K-1}\alpha\log\rho + \oone. \label{eq:innerlowint19}
\end{eqnarray}
When $\min\lcb N, |S|M\rcb = |S|M$, then we have $|S| \leq \frac{N}{M}$, and (\ref{eq:innerlowint16}) becomes:
\begin{eqnarray}
& \sums R_{c,j} &\leq \min\lcb |S|M, N-M \rcb\alpha\log\rho + \oone. \label{eq:innerlowint21}
\end{eqnarray}
When $\min\lcb |S|M, N-M \rcb = N-M$, then we have $N \leq (1 + |S|)M$ and (\ref{eq:innerlowint16}) becomes
\begin{eqnarray}
& \sums R_{c,j} &\leq (N-M)\alpha\log\rho + \oone, \nonumber \\
& \text{or } R_{c,j} &\leq \frac{N-M}{K-1}\alpha\log\rho + \oone. \:\:(\because \text{for } |S| = K-1, \text{ RHS is minimized} )\label{eq:innerlowint22}
\end{eqnarray}
When $\min\lcb |S|M, N-M \rcb = |S|M$, then $(1 + |S|)M \leq N$, and (\ref{eq:innerlowint16}) becomes
\begin{eqnarray}
R_{c,j} \leq  M\alpha\log\rho + \oone. \label{eq:innerlowint24}
\end{eqnarray}
Now consider the user set $S^{'} =  S \cup \{1\}$, where user 1 is always included. By following the same procedure as in the previous case, the following equation similar to (\ref{eq:innerlowint7}) is obtained
\begin{eqnarray}
& \sum_{j \in S'} R_{c,j} & \leq  M\alpha\log\rho + \min\lcb \min\lcb N, |S|M\rcb, N- M\rcb\alpha\log\rho  + \oone. \label{eq:innerlowint25}
\end{eqnarray}
When $\min\lcb N, |S|M\rcb = |S|M$, then following equation is obtained:
\begin{eqnarray}
& \sum_{j \in S'} R_{c,j} & \leq  M\alpha\log\rho + \min\lcb |S|M, N-M\rcb\alpha\log\rho + \oone. \label{eq:innerlowint26} 
\end{eqnarray}
By following the same procedure as in previous case, above equation is further simplified and following rate constraints are obtained under following conditions:

When $\frac{N}{M} - 1 \leq |S| \leq \frac{N}{M}$, then we have:
\begin{eqnarray}
R_{c,j} \leq \frac{N\alpha}{1 + |S|}\log\rho + \oone. \label{eq:innerlowint27} 
\end{eqnarray}
The above equation is minimized when $|S|$ takes its maximum value, i.e., $|S| = \frac{N}{M}$ (assume $\frac{N}{M}$ is an integer) and (\ref{eq:innerlowint27}) becomes:
\begin{eqnarray}
R_{c,j} \leq \frac{N\alpha}{1 + \frac{N}{M}}\log\rho + \oone.  \label{eq:innerlowint28}
\end{eqnarray}
When $|S| < 1 + |S| \leq \frac{N}{M}$, (\ref{eq:innerlowint26}) reduces to:
\begin{eqnarray}
R_{c,j} \leq  M\alpha\log\rho + \oone. \label{eq:innerlowint29}
\end{eqnarray}
When $N -M\leq |S|M$, using $|S|_{\max} = K-1$ in this case, (\ref{eq:innerlowint26}) becomes
\begin{eqnarray}
R_{c,j} \leq \frac{N\alpha}{K}\log\rho + \oone. \label{eq:innerlowint30}
\end{eqnarray}
When $\min\lcb N,|S|M\rcb = N$, then we have $N \leq |S|M$ and (\ref{eq:innerlowint25}) becomes
\begin{eqnarray}
 & \sum_{j \in S'} R_{c,j} & \leq  M\alpha\log\rho + \min\lcb N, N- M\rcb\alpha\log\rho  + \oone, \nonumber \\
 & \text{or } R_{c,j} & \leq \frac{N\alpha}{1 + |S|}\log\rho + \oone. \label{eq:innerlowint30a}
\end{eqnarray}
In the above equation, the right hand side is minimized when $|S|$ takes the maximum value. In this case, its maximum value is $|S| = |S|_{\max} = K-1$ and (\ref{eq:innerlowint30a}) becomes
\begin{equation}
R_{c,j}  \leq \frac{N\alpha}{K}\log\rho + \oone. \label{eq:innerlowint30b}
\end{equation}
Finally, the achievable rate by common part of the message is obtained by taking minimum of (\ref{eq:innerlowint19}), (\ref{eq:innerlowint22}), (\ref{eq:innerlowint24}), (\ref{eq:innerlowint28}), (\ref{eq:innerlowint29}), (\ref{eq:innerlowint30}) and (\ref{eq:innerlowint30b}). Given (\ref{eq:innerlowint30}), (\ref{eq:innerlowint28}) becomes redundant as $K > \frac{N}{M} + 1$. It is important to note that when $\frac{N}{M}$ is not an integer, then (\ref{eq:innerlowint27}) is minimized when $|S| = \lfloor \frac{N}{M}\rfloor$  and as $K > \frac{N}{M}+1 >  \lfloor \frac{N}{M}\rfloor + 1$, the above mentioned result still remains valid. Given (\ref{eq:innerlowint19}) and (\ref{eq:innerlowint22}), (\ref{eq:innerlowint24}), (\ref{eq:innerlowint30}) and (\ref{eq:innerlowint30b}) become redundant. Finally, the GDOF achievable by the common part of the message:
\begin{eqnarray}
d_{c}(\alpha) \geq \frac{N-M}{K-1}\alpha. \label{eq:innerlowint31} 
\end{eqnarray}
The per user GDOF achievable in this case is:
\begin{eqnarray}
d(\alpha) \geq M(1-\alpha) + \frac{N-M}{K-1}\alpha. \label{eq:innerlowint32} 
\end{eqnarray}
The above equation is same as that in previous case. This completes the proof.
\subsection{Proof of Theorem \ref{innermaxhigh-th}}\label{sec:innermax-high}
Following two cases are considered.\\
\textbf{Case 1 $\lb \frac{N}{M} < K \leq \frac{N}{M} + 1\rb$}:\\
In this case, the HK-scheme achieves GDOF as given by (\ref{eq:strongth1a}). The HK scheme achieves the interference free GDOF provided
\begin{eqnarray}
& & \frac{1}{K}[\alpha(K-1)M+N-(K-1)M] \geq M, \nonumber \\
& & \text{or } \alpha \geq \frac{M(2K-1)-N}{M(K-1)}. \label{eq:achhighint1}
\end{eqnarray}
Hence, (\ref{eq:strongth1a}) can also be expressed as:
\begin{eqnarray}
& & d_{\text{HK}}(\alpha) \geq \lcb\begin{array}{ll}
    \frac{1}{K}[\alpha(K-1)M+N-(K-1)M] &\mbox{for $ 1 < \alpha < \frac{M(2K-1)-N}{M(K-1)}$} \\ 	
    M & \mbox{for $ \alpha \geq \frac{M(2K-1)-N}{M(K-1)}.$}
 \end{array}\right. \label{eq:achhighint2}
\end{eqnarray}
Comparing (\ref{eq:achhighint2}) with (\ref{eq:innercombine2}), it is easy to see that the HK-scheme outperforms ZF-receiving for all $\alpha > 1$. The HK-scheme also outperforms treating interference as noise.\\
\textbf{Case 2 $\lb K \geq \frac{N}{M} + 1 \rb$}:\\
In this case, from Theorem \ref{th-highint1}, the HK-scheme achieves a per user GDOF given by (\ref{eq:strongth1b}). Comparing the HK-scheme with IA, it is easy to show that the former outperforms the latter for $\alpha > \frac{KM}{N+M}$. Hence, we obtain the following result.
\begin{eqnarray}
& & d(\alpha) \geq \left\{\begin{array}{lll}
     \frac{MN}{M+N }    &\mbox{for $ 1 \leq \alpha \leq \frac{KM}{M+N}$} \\ 	
     \frac{\alpha N}{K} & \mbox{for $ \frac{KM}{M+N} < \alpha < \frac{KM}{N} $} \\
     M                   & \mbox{for $\alpha \geq \frac{KM}{N}. $}
 \end{array}\right. \label{eq:achhighint7}
\end{eqnarray}
Also the HK-scheme always outperforms ZF-receiving and treating interference as noise.
This completes the proof.
\subsection{Proof of Theorem \ref{innermaxmoderate-th}}\label{sec:innermax-moderate}
The following two cases are considered in this regime.\\
\textbf{Case 1 $\lb \frac{N}{M} < K \leq \frac{N}{M} + 1\rb$}:\\
In this case, from Theorem \ref{th:moderateint1}, the per user GDOF achievable by the HK-scheme can be expressed as:
\begin{eqnarray}
& & d_{\text{HK}}(\alpha) = \lcb\begin{array}{ll}
    M(1-\alpha) + \frac{1}{K-1}\lsqb M\lcb \alpha(2K-1)-K\rcb+N(1-\alpha)\rsqb&\mbox{for $\frac{1}{2} \leq \alpha \leq \frac{K}{2K-1} $} \\ 	
    M(1-\alpha) + \frac{N\alpha}{K}& \mbox{for $ \frac{K}{2K-1} \leq \alpha \leq 1. $} \end{array}\right. \label{eq:innercombine13a}
\end{eqnarray}
It can be shown that the HK-scheme performs better than treating interference as noise and ZF-receiving, with their performance coincide at $\alpha = 1$. In this case, IA is not applicable.\\
\textbf{Case 2 $\lb K > \frac{N}{M} + 1\rb$}:\\
In this case, the HK-scheme as well as IA perform better than ZF-receiving and treating interference as noise and at $\alpha=1$, the HK-scheme coincides with ZF-receiving. In this regime, the achievable per GDOF in Theorem \ref{th:moderateint1} simplifies to
\begin{eqnarray}
& & d_{\text{HK}}(\alpha) = \lcb\begin{array}{ll}
    M(1-\alpha) +  \frac{N\alpha - M(1-\alpha)}{K-1} &\mbox{for $\frac{1}{2} \leq \alpha \leq \frac{KM}{N+KM} $} \\ 	
    M(1-\alpha) + \frac{N\alpha}{K} &\mbox{for $ \frac{KM}{N+KM} < \alpha \leq 1 .$} \end{array}\right. \label{eq:innercombine13}
\end{eqnarray}
When $\frac{1}{2}\leq \alpha \leq \frac{KM}{N+KM}$, the HK-scheme outperforms IA when
\begin{eqnarray}
& & M (1-\alpha) + \frac{1}{K-1}\lsqb N\alpha - M(1-\alpha)\rsqb \geq \frac{MN}{M+N}, \label{eq:innercombine14} \\
 \text{or }& & \alpha \lsqb N - M(K-2)\rsqb \geq M \lsqb \frac{N - M(K-2)}{M+N}\rsqb. \label{eq:innercombine15}
\end{eqnarray}
When $N - M(K-2) \geq 0$, then following condition on $\alpha$ is obtained:
\begin{eqnarray}
& & \alpha \geq \frac{M}{M+N}, \label{eq:innercombine16}
\end{eqnarray}
which is satisfied for all $\alpha$ in the moderate interference regime and hence the HK-scheme always performs better than IA. 

When $N - M(K-2) < 0$, the following condition is obtained:
\begin{eqnarray}
& & \alpha < \frac{M}{M+N} \leq \frac{1}{2}.\label{eq:innercombine17}
\end{eqnarray}
In this case, it is not possible to find an $\alpha$ which satisfies the above condition, and hence, IA always outperforms the HK-scheme. 

When $\frac{KM}{N+KM} < \alpha \leq 1$, from (\ref{eq:innercombine13}), the HK-scheme outperforms IA when
\begin{eqnarray}
& & M(1-\alpha) + \frac{N\alpha}{K} \geq \frac{MN}{M+N}, \label{eq:innercombine18} \\
\text{or } & & \alpha \leq \frac{KM^{2}}{(M+N)(KM-N)}. \label{eq:innercombine19}
\end{eqnarray}
This case does not arise if $N=M(K-2)$. Rather, IA outperforms HK for $\frac{KM}{N+KM} < \alpha \leq 1$.

From (\ref{eq:innercombine16}), (\ref{eq:innercombine17}) and (\ref{eq:innercombine19}), following conditions are obtained.
\begin{enumerate}
 \item When $N - M(K-2)\geq 0$ i.e., $K< 2 + \frac{N}{M}$, then we have following conditions:
\begin{enumerate}
\item When $\frac{1}{2} \leq \alpha \leq \frac{KM}{N+KM}$, the HK-scheme performs better than IA and it achieves a per user GDOF of
\begin{eqnarray}
d(\alpha) \geq M(1-\alpha) +  \frac{N\alpha - M(1-\alpha)}{K-1}. \label{eq:innercombine22}
\end{eqnarray}
\item When $\frac{KM}{N+KM} < \alpha \leq \frac{KM^{2}}{(M+N)(KM-N)}$, the HK-scheme outperforms IA and achieves a per user GDOF of
\begin{eqnarray}
d(\alpha) \geq M(1-\alpha) + \frac{N\alpha}{K}. \label{eq:innercombine23}
\end{eqnarray}
\item When $\frac{KM^{2}}{(M+N)(KM-N)} < \alpha \leq 1$, IA performs the best and the following per user GDOF is achievable:
\begin{eqnarray}
d(\alpha) \geq \frac{MN}{M+N}. \label{eq:innercombine25}
\end{eqnarray}
\end{enumerate}
\item  When $N-M(K-2)<0$, i.e., $K > 2 + \frac{N}{M}$, IA performs better than the HK-scheme for $\frac{1}{2} \leq \alpha \le 1$, and the following per user GDOF is achievable:
\begin{eqnarray}
d(\alpha) \geq \frac{NM}{N+M}. \label{eq:innercombine26}
\end{eqnarray}
\end{enumerate}
This completes the proof.
\subsection{Proof of Theorem \ref{innermaxweak-th}}\label{sec:innermax-weak}
From Theorem \ref{th:lowinter1} and \ref{th:theorem1_intnoise}, the HK-scheme and treating interference as noise achieve following per user GDOF:
\begin{eqnarray}
& &  d_{\text{HK}}(\alpha) = M(1-\alpha) + \frac{1}{K-1}(N-M)\alpha, \label{eq:innercombine3} \\
\text{ and }& & d_{\text{int-noise}}(\alpha) = 
  \left\{\begin{array}{ll} 
    M + \alpha(N-KM) &\mbox{for $ \frac{N}{M} < K \leq \frac{N}{M} + 1 $} \\ 	
    M(1-\alpha) & \mbox{for $ K > \frac{N}{M} + 1.$}
 \end{array}\right. \label{eq:innercombine4}
\end{eqnarray}
When $K > \frac{N}{M} + 1$, from (\ref{eq:innercombine3}) and (\ref{eq:innercombine4}), it can be observed that the HK-scheme performs better than treating interference as noise. The per user GDOF achievable by the HK-scheme can also be expressed as follows:
\begin{equation}
d_{\text{HK}}(\alpha) = M + \frac{1}{K-1}(N-KM)\alpha.  \label{eq:innercombine5}
\end{equation}
When $\frac{N}{M} < K \leq \frac{N}{M} + 1$, since $\frac{1}{K-1}(N-KM)\alpha > (N-KM)\alpha$ and hence HK-scheme performs better than treating interference as noise in this case also. Now HK-scheme outperforms the IA scheme whenever
\begin{eqnarray}
& & M(1-\alpha) + \frac{1}{K-1}(N-M)\alpha > \frac{MN}{M+N}, \nonumber \\
\text{i.e., } & &\alpha > \frac{M^{2}}{M(N+M)-\frac{N^{2}-M^{2}}{K-1}}. \label{eq:innercombine6}
\end{eqnarray}
Since $\alpha < \frac{1}{2}$, the right hand side is less than $ \frac{1}{2}$, which requires
\begin{eqnarray}
 K > 2 + \frac{N}{M}. \label{eq:innercombine7}
\end{eqnarray}
Thus, when $K > \frac{N}{M}+2$ and (\ref{eq:innercombine6}) are satisfied, the HK-scheme performs better than IA.
Comparing the HK-scheme with ZF-receiving, it is easy to show that HK-scheme outperforms ZF-receiving for $\alpha \leq \frac{1}{2}$. The two scheme coincide at $\alpha = \frac{1}{2}$ when $K=2$.

To summarize, when $K > \frac{N}{M} + 2$, the per user GDOF that can be achieved in the weak interference regime is:
\begin{eqnarray}
& d(\alpha) & \geq \left\{\begin{array}{ll}
    M(1-\alpha) + \frac{1}{K-1}(N-M)\alpha &\mbox{for $0 \leq \alpha \leq \frac{M^{2}}{M(N+M)-\frac{N^{2}-M^{2}}{K-1}} $} \\ 	
    \frac{NM}{N+M} & \mbox{for $ \frac{M^{2}}{M(N+M)-\frac{N^{2}-M^{2}}{K-1}} < \alpha \leq \frac{1}{2}.$}
 \end{array}\right.\nonumber \\ \label{eq:innercombine11}
\end{eqnarray}
When $\frac{N}{M} < K \leq \frac{N}{M}+2$, HK-scheme alone performs better than the other schemes and the per user GDOF achievable by this scheme is:
\begin{eqnarray}
& d(\alpha) & \geq  M(1-\alpha) + \frac{1}{K-1}(N-M)\alpha. \label{eq:innercombine12}
\end{eqnarray}
This completes the proof.

\subsection{Proof of Theorem \ref{thm:NbyMnotInteger}}\label{sec:new result}
First, recall that the maximum of the achievable GDOF from the HK-scheme and IA outperforms the achievable GDOF from treating interference as noise or ZF-receiving for all values of $M$, $N$, $K$ and $\alpha$. Hence, the above result follows from carefully comparing the achievable GDOF from the HK-scheme and IA in the weak, moderate, and strong interference cases. 

 \emph{Weak interference case $(0 \le \alpha \le \frac{1}{2})$:} Comparing the achievable GDOF using IA, given by (\ref{eq:inneralign1}), with that achievable using the HK-scheme, given by (\ref{eq:weakth1}), it follows that the HK-scheme is active when 
\begin{equation}
\alpha \le \frac{(K-1)}{\lb R+1\rb \lb K - \tfrac{N}{M} \rb}.
\end{equation}
When $R = 1$, since $\frac{N}{M} \ge 1$, it is clear that the right hand side above exceeds $\frac{1}{2}$. Hence, the HK-scheme is active throughout the weak interference case. When $R >1$, the right hand side above is $\le \frac{1}{2}$, provided
\begin{equation}
K \geq \frac{N}{M} + 2 \frac{\frac{N}{M} - 1}{R-1}.
\end{equation} 
Notice that, in the last term above, the denominator is the floor of the numerator. Hence, the ratio is bounded above by $2$. Hence, for $K \geq \frac{N}{M} + 4$, the HK-scheme is active for the initial part of the weak interference case. IA is active in the later part of the weak interference case.  This completes the proof in the weak interference case.

 \emph{Moderate interference case $(\frac{1}{2} \le \alpha \le 1)$:} Consider the achievable GDOF using the HK-scheme given by (\ref{eq:moderateth1b}) for $K > \frac{N}{M} + 1$. The expression can be equivalently written as
\begin{equation}
d(\alpha) \geq \lcb\begin{array}{l l}
   M(1-\alpha) + \frac{N\alpha + M(1-\alpha)}{K-1} &\mbox{for $\frac{1}{2} \le  \alpha < \frac{1}{1 + \frac{N}{MK}}$} \\ 	
    M (1-\alpha) + \frac{N \alpha}{K} & \mbox{for $ \frac{1}{1 + \frac{N}{MK}} \le \alpha < 1.$}
 \end{array}\right. 
\end{equation}
Consider the first case above, i.e., when $\frac{1}{2} \le  \alpha < \frac{1}{1 + \frac{N}{MK}}$. It can be shown that the above achievable GDOF exceeds that achievable by IA, provided 
\begin{equation}
\alpha \le \frac{(K-1) - (R+1)}{\lb R+1\rb \lb \lb K-1\rb - \lb\tfrac{N}{M} + 1\rb \rb}.
\end{equation}
Now, the right hand side above is smaller than $\frac{1}{1 + \frac{N}{MK}}$ when
$K \geq \frac{N}{M} + 2 \frac{N}{RM}$,
which is always satisfied when $K \geq \frac{N}{M} + 4$. When $R = 1$, it is immediate to see that the right hand side above exceeds $\frac{1}{2}$, and hence, the HK-scheme is active for an initial portion of $\frac{1}{2} \le  \alpha < \frac{1}{1 + \frac{N}{MK}}$. When $R > 1$, the right hand side above is smaller than $\frac{1}{2}$ and hence IA is active throughout this range of $\alpha$, provided
\begin{equation}
K \geq \frac{N}{M} + 2 \frac{\frac{N}{M} - 1}{R - 1},
\end{equation}
which is satisfied when $K \geq \frac{N}{M} + 4$. \\
Next, consider the second case above, i.e., when $ \frac{1}{1 + \frac{N}{MK}} \le \alpha < 1.$ In this case, the HK-scheme outperforms IA when 
\begin{equation}
\alpha \le \frac{1}{(R+1) \lb 1 - \frac{N}{MK} \rb}.
\end{equation}
When the right hand side above is $\le \frac{1}{1 + \frac{N}{MK}}$, IA is active throughout this range of $\alpha$. This leads to 
\begin{equation}
K \ge \frac{N}{M} + 2 \frac{N}{RM}
\end{equation}
which is satisfied when $K \geq \frac{N}{M} + 4$. This completes the proof in the moderate interference case. 

 \emph{Strong interference case $(\alpha \ge 1)$:} In this case, from (\ref{eq:strongth1b}), the achievable GDOF from the HK scheme when $K \ge \frac{N}{M} + 4$ is given by 
\begin{equation}
d(\alpha) \geq \lcb\begin{array}{l l}
    \frac{N\alpha}{K} &\mbox{for $ 1 \le  \alpha < \frac{MK}{N}$} \\ 	
    M & \mbox{for $ \alpha \ge \frac{MK}{N}.$}
 \end{array}\right. 
\end{equation}
Comparing the above achievable GDOF using IA given by (\ref{eq:inneralign1}), one obtains
\begin{equation}
d(\alpha) \geq \lcb\begin{array}{l l}
\frac{RM}{R+1} & \mbox{for $1 \le \alpha \le \frac{MKR}{N(R+1)}$} \\
    \frac{N\alpha}{K} &\mbox{for $ \frac{MKR}{N(R+1)} <  \alpha \le \frac{MK}{N}$} \\ 	
    M & \mbox{for $ \alpha > \frac{MK}{N}.$}
 \end{array}\right. 
 \end{equation}
The statements of the theorem are now easily obtained by consolidating the above results. 
\subsection{Proof of Corollary \ref{tightcor}}\label{sec:tightoutercorl}
When $M=N$, treating interference as noise achieves a per user GDOF of $M(1-\alpha)$ (Theorem~\ref{th:lowinter1}). In Lemma~\ref{lemma-outer2}, when $M=N$, the outer bound reduces to $M(1-\alpha)$ and hence, treating interference as noise is GDOF optimal in the weak interference regime $(0 \leq \alpha \leq \frac{1}{2})$. 

In order to prove the tightness of the outer bounds when $\frac{N}{M} < K \leq \frac{N}{M}+1$, the following interference regimes are considered. \\
\textit{Weak interference regime $(0 \leq \alpha \leq \frac{1}{2})$:} In the weak interference regime, the following per user GDOF (Theorem~\ref{th:lowinter1}) is achievable:
\begin{equation}
d(\alpha) \geq M(1-\alpha) + \frac{1}{K-1}(N-M)\alpha \label{eq:thirdrev10}
\end{equation}
The outer bound in (\ref{eq:minouter4}) matches with the achievable GDOF in the above equation,  thus establishing the optimality of the achievable scheme.\\ 
\textit{Moderate interference regime $(\frac{1}{2} \leq \alpha \leq 1)$:} In the moderate interference regime, the following per user GDOF (Theorem~\ref{th:moderateint1}) is achievable:
\begin{eqnarray}
& d(\alpha) & \geq M(1-\alpha) + \min\lcb \frac{N\alpha}{K}, \frac{\lsqb M  \lcb (2K-1)\alpha - K\rcb + N(1-\alpha))\rsqb}{K-1} \rcb, \nonumber \\
& & = \lcb\begin{array}{l l}
	 M\alpha + \frac{1}{K-1}(N-M)(1-\alpha) &\mbox{for $\frac{1}{2} \leq \alpha \leq \frac{K}{2K-1}$} \\
         M(1-\alpha) + \frac{N\alpha}{K} & \mbox{for $\frac{K}{2K-1}< \alpha \leq 1$.}
\end{array}\right. \label{eq:thirdrev11}
\end{eqnarray}
The achievable GDOF matches with the outer bound in (\ref{eq:minouter5}), and hence, the outer bound is tight in this case.\\
\textit{High interference regime $(\alpha \geq 1)$:} In the high interference regime, the following per user GDOF (Theorem \ref{th-highint1}) is achievable:
\begin{eqnarray}
& d(\alpha) & \geq \min\lcb M, \frac{1}{K}\lsqb(K-1)M\alpha + N - (K-1)M \rsqb\rcb \nonumber \\
& & = \lcb\begin{array}{l l}
	 \frac{1}{K}\lsqb N + (K-1)M(\alpha-1)\rsqb &\mbox{for $1 \leq \alpha \leq \frac{2KM-(M+N)}{(K-1)M}$} \\
         M & \mbox{for $\alpha \geq \frac{2KM-(M+N)}{(K-1)M}$,}
\end{array}\right.  \label{eq:thirdrev12}
\end{eqnarray}
which matches with the outer bound in (\ref{eq:minouter6}). Thus, the outer bound is tight in the high interference regime also. This completes the proof.
\subsection{Comparison of the Outer Bound in Lemma~\ref{lemma-outer1} and the $Z$ Channel Outer Bound in \cite{zhengdao1}}\label{sec:compz}
In this subsection, we show that the sum DOF outer bound from Theorem 1 matches with the outer bound on the sum DOF of the MIMO GIC \cite{zhengdao1}. The latter is given by \cite{zhengdao1}:
\begin{equation}
d_{1} + d_{2} \leq \min\lcb\max\lcb N_{1}, M_{2} \rcb, M_{1} + M_{2}, N_{1} + N_{2} \rcb \label{eq:thirdrev0}
\end{equation}
In order to simplify the comparison between the two outer bounds, let us first assume that $M_{1} \leq N_{1} $ and $M_{2} \leq N_{2}$. Recall that, in Theorem 1, two groups with $L_1$ and $L_2$ users each are considered, such that $0 < L_{1}+L_{2} \leq K$. Let $M_{1} \triangleq L_1M$, $M_2 \triangleq L_2M$, $N_1 \triangleq L_1N$ and $N_2 = L_2N$. Then, the sum rate bound in (\ref{eq:lmcoopouter2}) in the paper is expressed as follows:
\begin{eqnarray}
 & R_{1} + R_{2}  &\leq \log \labs \Iden_{N_1} +  \rho\Honebar\Honebar^{H} + \rho^{\alpha} \Htwobar\Htwobar^{H}\rabs \nonumber \\
& & \qquad + \log \labs \Iden_{N_2} + \rho\Hthreebar \lcb \Iden_{M_2} + \rho^{\alpha}\Htwobar^{H}\Htwobar \rcb^{-1}\Hthreebar^{H}\rabs. \label{eq:thirdrev1}
\end{eqnarray}
where the effective channel $\mathbf{\overline{H}}_{ij} \in \Complex^{N_{i} \times M_{j}}$.

Using the procedure employed to obtain (\ref{eq:lmcoopouter3}) in the paper, (\ref{eq:thirdrev1}) reduces to the following form:
\begin{eqnarray}
& & R_1 + R_2 \leq \log \labs \Iden_{N_1} +  \rho\Honebar\Honebar^{H} + \rho^{\alpha} \Htwobar\Htwobar^{H}\rabs \nonumber \\
& & \qquad \quad\quad +\log \labs \Iden_{N_2} + \rho^{1-\alpha} \Htildetwtwo^{(a)} \Sigma_{r}^{-1}\Htildetwtwo^{(a)H} + \rho\Htildetwtwo^{(b)}\Iden_{M_2 -r}\Htildetwtwo^{(b)H}\rabs + \oone. \label{eq:thirdrev2} 
\end{eqnarray}
The symbols/notations defined in the above equation are the same as that defined in the proof of Lemma \ref{lemma-outer1} of the paper. The matrices $\Htildetwtwo^{(a)}$ and $\Htildetwtwo^{(a)}$ are of dimensions $N_2 \times r$ and $N_2 \times (M_2 -r)$, respectively, where $r \triangleq \min\{M_2, N_1\}$.  When $0 \leq \alpha \leq 1$, (\ref{eq:thirdrev2}) becomes
\begin{eqnarray}
& R_{1} + R_{2} &\leq M_{1}\log\rho + \min\lcb r, N_{1} - M_{1} \rcb\alpha\log\rho + (M_2 - r)\log\rho + \nonumber \\
& & \qquad \qquad \min\lcb \min\lcb N_2,r\rcb, N_2 - M_2 + r\rcb(1-\alpha)\log\rho + \oone. \label{eq:thirdrev3}
\end{eqnarray}
When $\alpha=1$, the bound on the sum DOF is
\begin{equation}
d_{1} + d_{2} \leq M_{1} + \min\lcb  r, N_{1} - M_{1} \rcb\ + M_{2} -r.\label{eq:thirdrev4}
\end{equation}

\textbf{Case 1 $(M_2 \leq N_1)$}: Under this condition, $r = M_2$, and (\ref{eq:thirdrev4}) becomes
\begin{equation}
d_{1} + d_{2} \leq M_{1} + \min\lcb M_2, N_1 - M_1\rcb. \label{eq:thirdrev5}
\end{equation}
When $M_1 + M_2 \geq N_{1}$, the above equation becomes
\begin{equation}
d_{1} + d_{2} \leq N_{1}. \label{eq:thirdrev6}
\end{equation}
As $M_{1} + M_{2} \geq N_{1}$, (\ref{eq:thirdrev0}) becomes
\begin{equation}
d_{1} + d_{2} \leq N_{1}. \label{eq:thirdrev7}
\end{equation}
Hence, the two outer bounds match in this case.

When $M_{1} + M_{2} < N_{1}$, then (\ref{eq:thirdrev5}) becomes
\begin{equation}
d_{1} + d_{2} \leq M_{1} + M_{2}. \label{eq:thirdrev8}
\end{equation}
In this case also, the above outer bound matches with the outer bound in (\ref{eq:thirdrev0}).

\textbf{Case 2 $(M_2 > N_1)$} Under this condition, $r = N_1$ and (\ref{eq:thirdrev4}) becomes
\begin{equation}
d_{1} + d_{2} \leq M_{2}. \label{eq:thirdrev9} 
\end{equation}
In this case also, the proposed outer bound coincides with the outer bound in \cite{zhengdao1}.

Proceeding in a similar way, it can be shown that the proposed outer bound matches with the outer bound on the sum DOF for the MIMO $Z$ GIC in the remaining cases $(M_1 > N_1 \text{ or } M_2 > N_2)$ as well.
\bibliographystyle{IEEEtranTCOM}
\bibliography{IEEEabrv,refs}
\newpage
\begin{figure}[ht]
\begin{center}
\setxysizeo
\epsffile{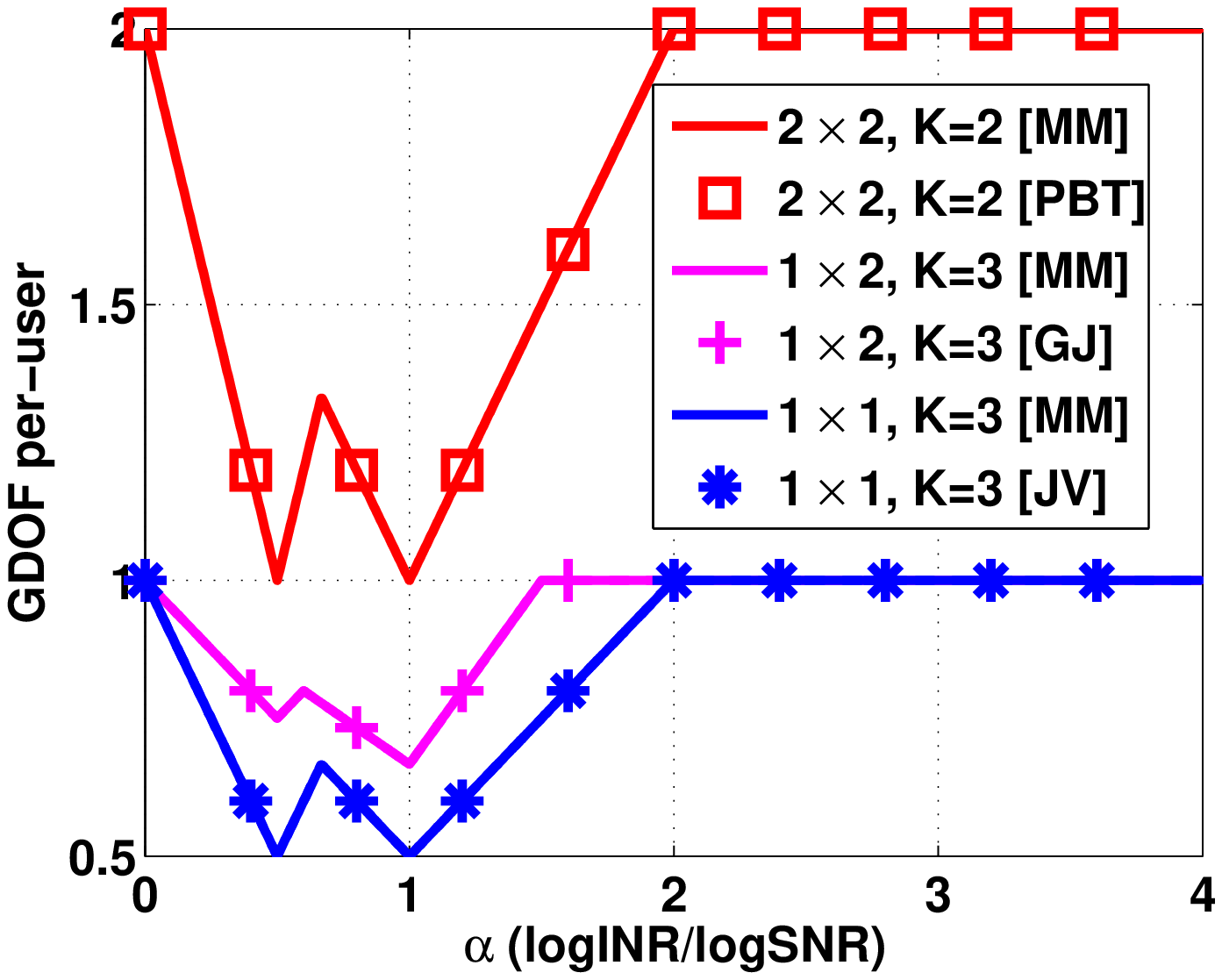}
\caption{Outer bound on per user GDOF for symmetric MIMO GIC with different antenna configuration and number of users. In the legend, MM stands for the outer bound derived in our work, PBT stands for the outer bound on GDOF result in \cite{tarokh1}, GJ stands for the outer bound on GDOF result in \cite{gou3}, and JV stands for the outer bound on GDOF result in \cite{jafar2}.} 
\label{fig:compare1}
\end{center}
\end{figure}

\begin{figure}[ht]
\begin{center}
\setxysizeo
\epsffile{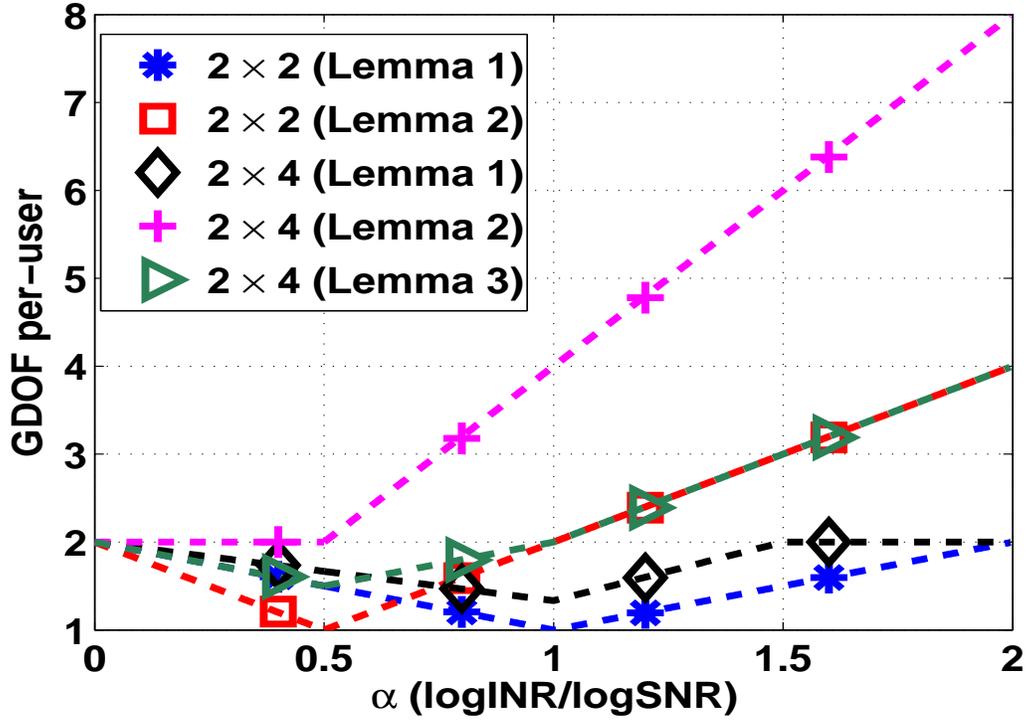}
\caption{Comparison of the different outer bounds on per user GDOF for the $K=3$ user symmetric MIMO GIC with $(M,N) = (2,2)$ and $(2,4)$.} 
\label{fig:outercomp1}
\end{center}
\end{figure}

\begin{figure}[ht]
\begin{center}
\setxysizeo
\epsffile{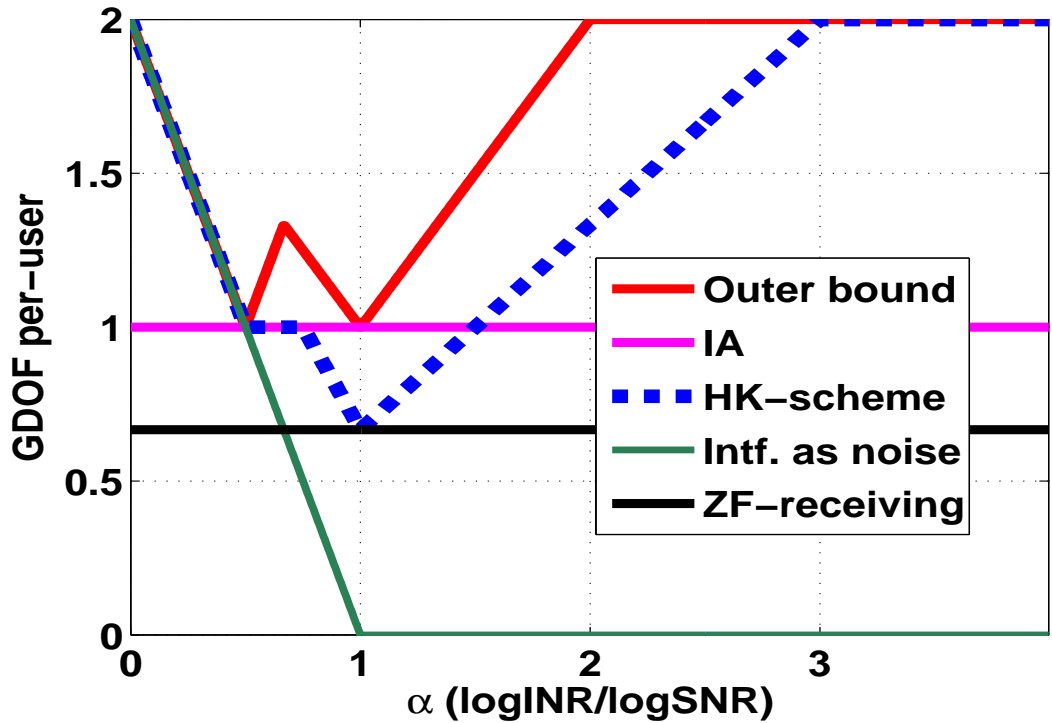}
\caption{GDOF for $K = 3$ user Interference Channel with $M=N=2$} 
\label{fig:P1_K3M2N2}
\end{center}
\end{figure}

%
%
%
%
%

\begin{figure}[ht]
\begin{center}
\setxysizeone
\epsffile{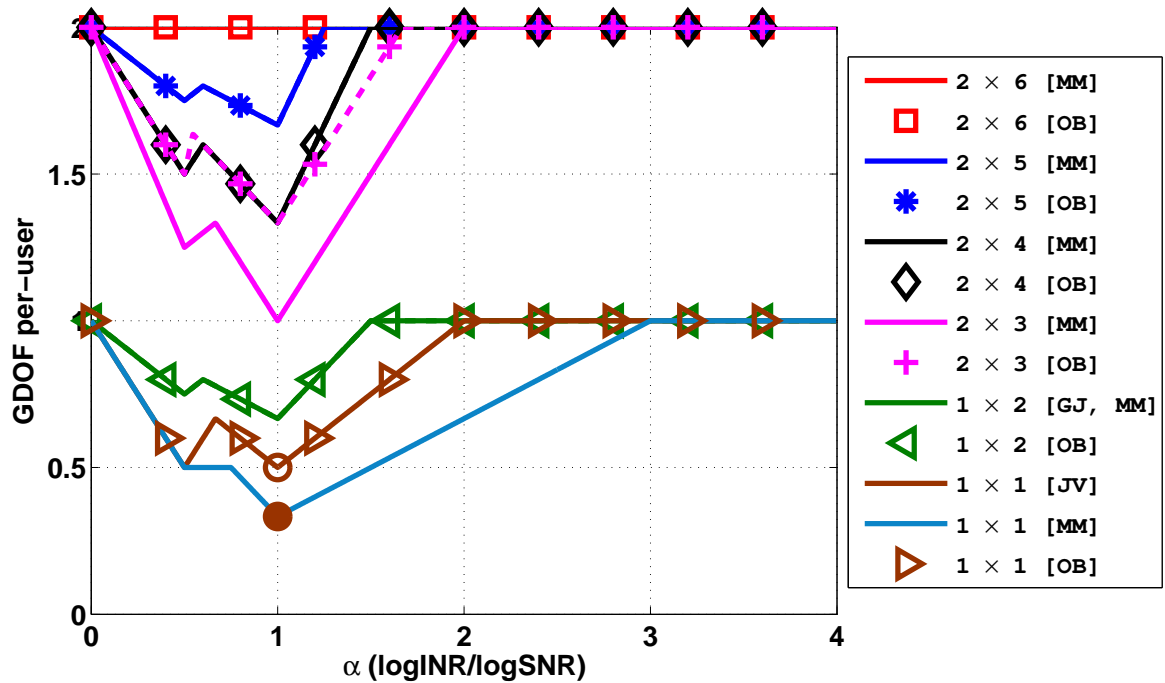}
\caption{The achievable GDOF for the $K = 3$ user symmetric MIMO GIC with different antenna configurations. In the legend, MM stands for inner bound derived in this work, JV stands for the achievable GDOF result in \cite{jafar2}, GJ stands for the achievable GDOF result in \cite{gou3}, and OB stands for the outer bound derived in this work.}
\label{fig:P_KMN}
\end{center}
\end{figure}

\begin{figure}[ht]
\begin{center}
\setxysizeo
\epsffile{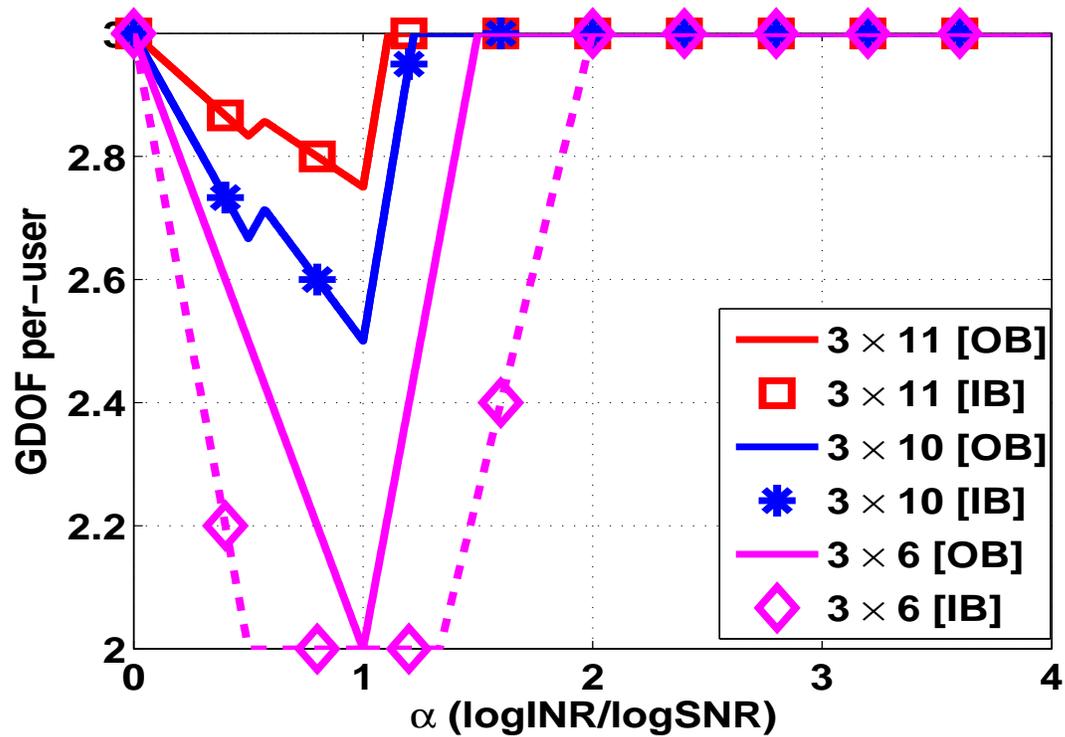}
\caption{Outer bound (OB) and inner bound (IB) on the per user GDOF for $K=4$ user MIMO GSIC with different antenna configurations.} 
\label{fig:P1_K4MN}
\end{center}
\end{figure}

\begin{figure}[ht]
\begin{center}
\setxysizeo
\epsffile{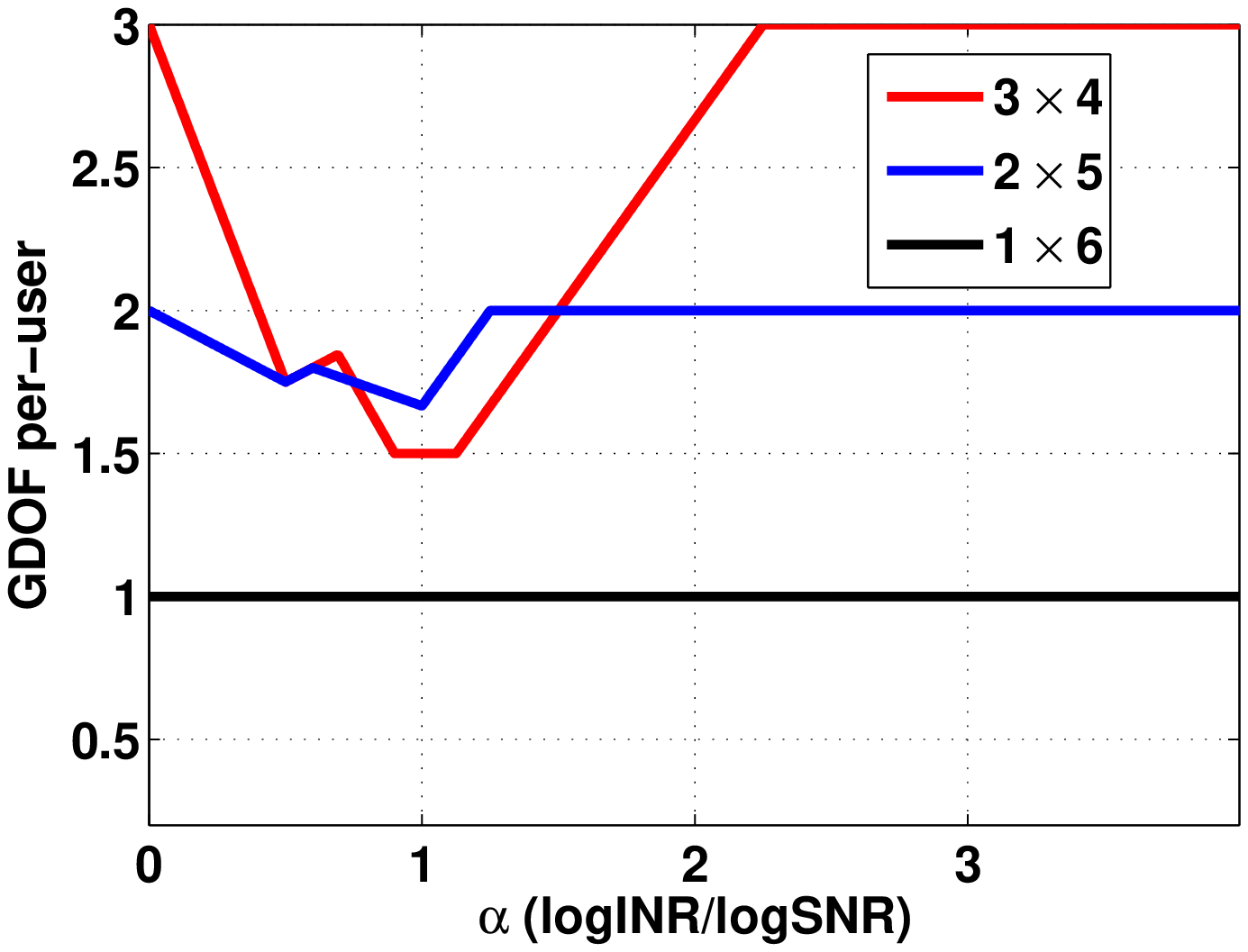}
\caption{The achievable GDOF for the $K = 3$ user symmetric MIMO GIC with different antenna configurations such that $M + N =7$.} 
\label{fig:K3MN7}
\end{center}
\end{figure}

\begin{figure}[ht]
\begin{center}
\setxysizeone
\epsffile{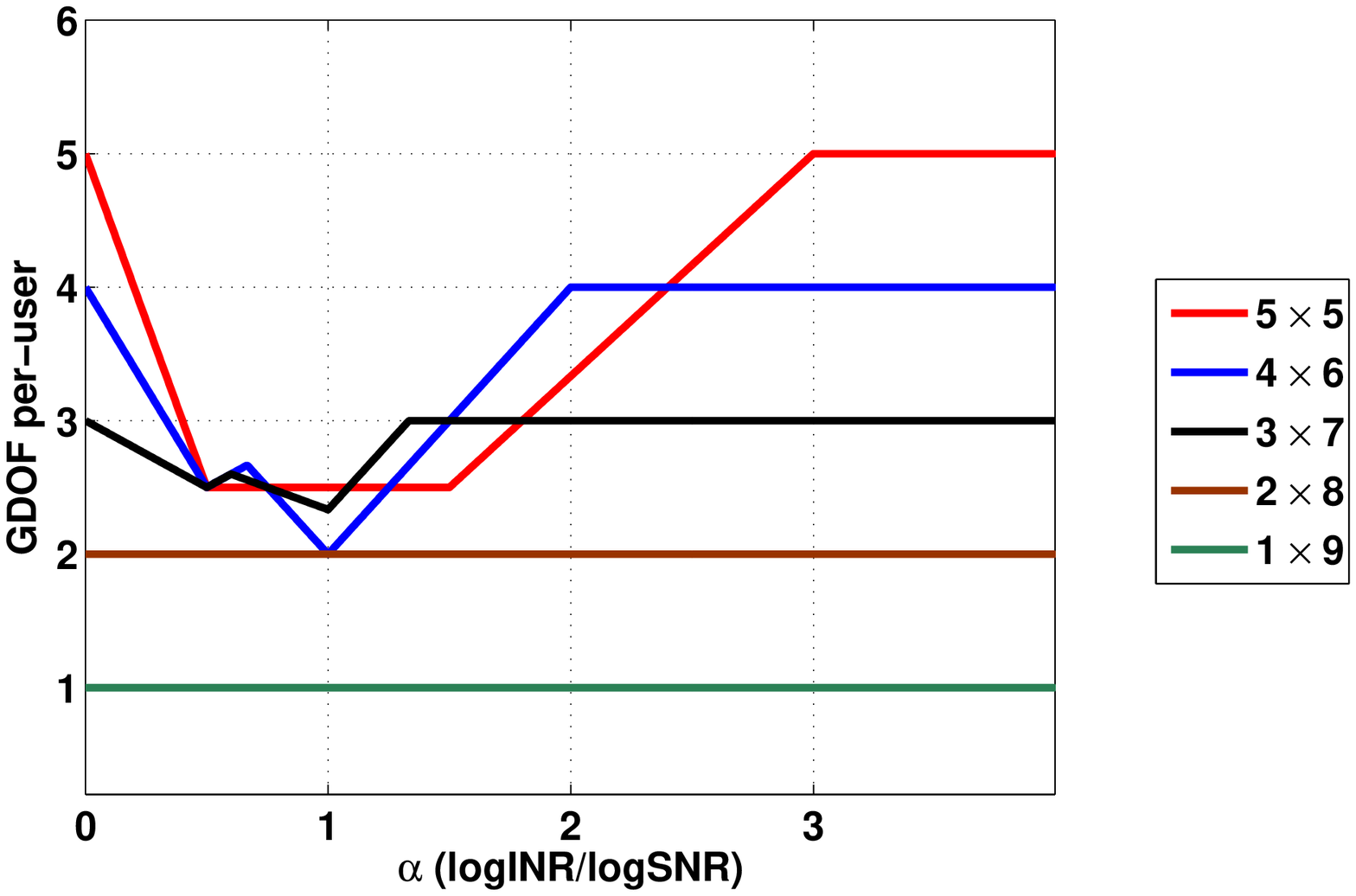}
\caption{The achievable GDOF for the $K = 3$ user symmetric MIMO GIC with different antenna configurations such that $M + N =10$.} 
\label{fig:K3MN10}
\end{center}
\end{figure}
\end{document}